\documentclass[11pt,a4paper]{amsart}

\usepackage{amsmath,amssymb,mathtools}
\usepackage[hmargin=3cm,vmargin=3.2cm]{geometry}
\usepackage[T1]{fontenc}
\usepackage[utf8]{inputenc}
\usepackage{enumitem}
\usepackage[dvipsnames]{xcolor}
\usepackage{cancel}
\usepackage{comment}

\numberwithin{equation}{section}

\newtheorem{theorem}{Theorem}[section]
\newtheorem{lemma}[theorem]{Lemma}
\newtheorem{definition}[theorem]{Definition}
\newtheorem{remark}[theorem]{Remark}
\newtheorem{corollary}[theorem]{Corollary}
\newtheorem{rmk}[theorem]{Remark}

\newtheorem{proposition}[theorem]{Proposition}

\newtheorem{example}[theorem]{Example}
\DeclareMathOperator{\End}{End}

\DeclareMathOperator{\lcm}{lcm}

\newcommand{\tnabla}{\widetilde\nabla}

\newcommand{\dd}{\mathrm d}

\newcommand{\calH}{\mathcal H}

\def\na{\nabla}

\def\O{\mathcal{O}_M}
\def\Le{\mathcal{L}_e}
\def\T{\mathcal{T}_M}
\def\OO{\Omega_M}

\def\na{\nabla}

\def\inw{\in\{1,\dots,n\}}

\def\na{\nabla}

\def\O{\mathcal{O}_M}
\def\Le{\mathcal{L}_e}
\def\T{\mathcal{T}_M}
\def\OO{\Omega_M}

\def\na{\nabla}

\def\proof{\noindent\hspace{2em}{\itshape Proof: }}

\def\endproof{\hspace*{\fill}~\qedsymbol\par\endtrivlist\unskip}
\newcommand{\eqa}{\begin{eqnarray}}
\newcommand{\eeqa}{\end{eqnarray}}
\newcommand{\beq}{\begin{equation}}
\newcommand{\eeq}{\end{equation}}

\newcommand{\tn}{\widetilde \nabla}

\begin{document}

\title[Cyclic F-manifolds, distinguished  connections and integrability]{Cyclic F-manifolds, distinguished  connections and integrability}

\author{Alessandro Arsie}
\address{Department of Mathematics and Statistics, The University of Toledo, 2801 W. Bancroft St., Toledo, OH 43606, USA}
\email{alessandro.arsie@utoledo.edu}

\author{Paolo Lorenzoni}
\address{Dipartimento di Matematica e Applicazioni, Universit\`a degli Studi di Milano-Bicocca, Via Roberto Cozzi 53, I-20125 Milano, Italy; and INFN Sezione di Milano-Bicocca}
\email{paolo.lorenzoni@unimib.it}

\date{}

\begin{abstract}
We show that the geometry of Hertling-Manin F-manifolds $(M,\circ,e)$ provide the appropriate theoretical framework for studying the integrability of quasilinear systems of first-order evolutionary partial differential equations of the form ${\bf u}_t=X\circ {\bf u}_x$ (F-systems) under the mild assumption that the unit vector field is cyclic with respect to the operator of multiplication by the vector field $X$. This approach is very general and allows us to treat even non-regular systems that were previously beyond the scope of existing techniques. Like in the regular case  the information about integrability is contained
 in a torsionless connection associated with the system and the integrability condition reduces to a geometric condition involving the Riemann tensor of the connection and the structure functions of the product. We prove that a locally conservative F-system is integrable and, in the analytic setting, also the converse statement, thereby providing a full characterization of integrability. Moreover, in the analytic case, we prove the existence of a family of analytic symmetries providing, in principle, the unique local analytic solution of the Cauchy problem through the generalised hodograph method. 
\end{abstract}

\maketitle
\tableofcontents

\section{Introduction}\label{SectionIntro}

\subsection{Motivation and historical background}\label{SubsectIntroMotivation}
A \emph{system of hydrodynamic type} is a first-order quasilinear evolutionary system of PDEs
\begin{equation}\label{sysHT}
u^i_t=V^i_j(u)u^j_x,\qquad i\inw.
\end{equation}
Such systems are pervasive in mathematical physics: they arise in gas dynamics, chromatography, the Whitham averaging of soliton equations, the theory of Frobenius and F-manifolds, and the geometry of Hamiltonian PDE hierarchies. Two structural properties of \eqref{sysHT} have played a central role in their study.

\medskip

\noindent\emph{Diagonalisability via Riemann invariants.} When the eigenvalues of $V$ are pairwise distinct, the existence of local coordinates $(r^1,\dots,r^n)$ in which \eqref{sysHT} takes the diagonal form
\begin{equation}\label{DsysHT}
r^i_t=v^i(r)\,r^i_x,\qquad i\inw,
\end{equation}
is equivalent to the vanishing of the \emph{Haantjes tensor} of $V$.

\medskip

\noindent\emph{Hamiltonianity.} The system \eqref{sysHT} admits a local Hamiltonian structure of Dubrovin--Novikov type if and only if there exists a flat metric $g$ such that
\begin{eqnarray}
\label{DN1intro}
g_{ik}V^k_j&=&g_{jk}V^k_i,\\
\label{DN2intro}
\nabla^g_iV^k_j&=&\nabla^g_jV^k_i,
\end{eqnarray}
where $\nabla^g$ denotes the Levi-Civita connection of $g$. In this case \eqref{sysHT} can be written in the Hamiltonian form
\begin{equation}\label{HsysHT}
u^i_t=P^{ij}\frac{\delta H}{\delta u^j},\qquad i\inw,
\end{equation}
with
\begin{equation}\label{LPBHT}
P^{ij}=g^{ij}\partial_x-\frac{1}{2}g^{ik}g^{js}\left(\partial_kg_{sl}+\partial_lg_{ks}-\partial_sg_{kl}\right)u^l_x,
\end{equation}
and $H=\int h(u)\,dx$ a local functional of hydrodynamic type, so that the variational derivatives of $H$ coincide with the partial derivatives of the density $h$.

A third notion --- \emph{integrability} --- was introduced by Tsarev \cite{Tsarev} for diagonal systems \eqref{DsysHT} as the compatibility of the linear system governing the characteristic velocities of the symmetries. Tsarev called such systems \emph{semi-Hamiltonian} and proved that all their local solutions can be obtained via the \emph{generalised hodograph method}. A foundational result, conjectured by Novikov (see  \cite{DN}, page 6) and proved by Tsarev in \cite{Tsarev}, asserts that strictly hyperbolic Hamiltonian systems of hydrodynamic type are automatically integrable: in the diagonal case, the existence of a (not necessarily flat) solution to \eqref{DN1intro}--\eqref{DN2intro} implies Tsarev's compatibility condition.

A nice characterisation of Tsarev's integrability condition was obtained by S\'evennec. He proved in  \cite{Sevennec} that semi-Hamiltonian systems are conservative systems (i.e.\ systems that can be written as systems of conservation laws) admitting Riemann invariants.

\subsection{Beyond the diagonal case}\label{SubsectIntroBeyondDiag}
In the non-diagonalisable setting the picture is considerably more delicate, because formal compatibility of the linear system for symmetries no longer implies compatibility of \eqref{DN1intro}--\eqref{DN2intro}. A productive way to organise the non-diagonalisable theory has emerged through the geometry of \emph{F-manifolds}. In a series of works \cite{LPVG1,LPVG2,LPVG3} Tsarev's theory has been extended to systems whose tensor $V$ is block-diagonal with each block of lower-triangular Toeplitz form
\begin{equation}\label{toeplitz}
V_{(\alpha)}=
\begin{bmatrix}
v^{1(\alpha)} & 0 & \dots & 0\cr
v^{2(\alpha)} & v^{1(\alpha)} & \dots & 0\cr
\vdots & \ddots & \ddots & \vdots\cr
v^{m_\alpha(\alpha)} & \dots & v^{2(\alpha)} & v^{1(\alpha)}
\end{bmatrix},
\end{equation}
which is the natural normal form arising from regular F-manifolds in the sense of David and Hertling \cite{DH}.

Non-diagonalisable hydrodynamic systems arise in many contexts: integrable hierarchies built from the Fr\"olicher--Nijenhuis bicomplex \cite{LM,LP2}; non-semisimple (bi-)flat F-manifolds and Dubrovin--Frobenius manifolds \cite{LPR,LP1}; and Hamiltonian and bi-Hamiltonian structures of hydrodynamic type and their deformations \cite{DVLS};  non-diagonalisable systems of hydrodynamic type related to integrable geodesic flows (see \cite{BKM1,BKM2}); reductions of the soliton gas kinetic equation (see \cite{FP}) and mKP hierarchy (see \cite{XF}).  It is also worth mentioning \cite{KK,KO}. Within F-manifold theory, the prototypical example (see for instance \cite{LPR,ALImrn}) is what we call an \emph{F-system}:
\beq\label{FsysIntro}
u^i_t=(X\circ u_x)^i=c^i_{jk}X^ju^k_x,\qquad i\inw,
\eeq
where $\circ$ is the commutative associative product on $\T$ and $X$ is a vector field. Hertling--Manin's compatibility condition \cite{HM} for $\circ$ implies the vanishing of the Haantjes tensor of $V:=X\circ$ and conversely, under the assumption of regularity the vanishing of the Haantjes tensor allows to write $V$ in the form $V:=X\circ$ (see \cite{BKM3}), so F-systems are the natural non-diagonalisable analogues of Tsarev's diagonal systems. In the semisimple case F-systems coincide with diagonal systems studied by Tsarev: the components of $X$ are the characteristic velocities, and the canonical coordinates that put the structure constants in the form $c^i_{jk}=\delta^i_j\delta^i_k$ play the role of Riemann invariants.

\subsection{The regular case: the geometric reformulation of integrability}\label{SubsectIntroGeomReform}
Regular F-manifolds are F-manifolds equipped by a cyclic Euler vector field.
In the non-semisimple regular case and in the holomorphic setting, David and Hertling \cite{DH} exhibit distinguished coordinates in which the structure constants take the constant form
\begin{equation*}\label{strconsts_reg}
c^{i(\alpha)}_{j(\beta)k(\gamma)}=\delta^\alpha_\beta\delta^\alpha_\gamma\delta^i_{j+k-1},
\end{equation*}
for all $\alpha,\beta,\gamma\in\{1,\dots,r\}$ and $i\in\{1,\dots,m_\alpha\}$, $j\in\{1,\dots,m_\beta\}$, $k\in\{1,\dots,m_\gamma\}$, and \eqref{FsysIntro} acquires the block lower-triangular Toeplitz shape \eqref{toeplitz}.\footnote{A similar result holds true in the real analytic and in the real smooth setting in the case of real eigenvalues, while in the case of complex conjugate eigenvalues one has to replace the block triangular Toeplitz form \eqref{toeplitz} with its coomplex analogue (we refer to \cite{BKM3} for details).}

In the regular case the integrability of an F-system \eqref{FsysIntro} admits an elegant geometric reformulation. Under mild non-degeneracy assumptions on the vector field $X$, there is a \emph{unique} torsionless connection $\tn$ associated with $X$, characterised by
\[
\tn e=0,\qquad d_{\tn}(X\circ)=0.
\]
The formal compatibility condition for the symmetries of \eqref{FsysIntro} can then be written as a curvature constraint on $\tn$ (the two equivalent  conditions below being related by the Bianchi identity):
\beq\label{shc-intro}
\tilde R^s_{lmi}c^j_{ks}+\tilde R^s_{lik}c^j_{ms}+\tilde R^s_{lkm}c^j_{is}=0,\qquad
\tilde R^j_{skl}c^s_{mi}+\tilde R^j_{smk}c^s_{li}+\tilde R^j_{slm}c^s_{ki}=0.
\eeq
An F-manifold equipped with a torsionless connection $\nabla$ that is compatible with $\circ$ and satisfies \eqref{shc-intro} is called an \emph{F-manifold with compatible connection (and flat unit)} \cite{LPR}. In this language, integrable F-systems are precisely those that locally endow their underlying F-manifold with the structure of an F-manifold with compatible connection. The notion of integrability deserves some additional explanations. Indeed, the conditions \eqref{shc-intro} can be interpreted as compatibility conditions in the sense of Darboux (see \cite{darboux} and also \cite{BJK}) only for a special class of regular F-systems called \emph{Darboux-Tsarev} systems. For this special class one can fully extend Tsarev's theory, proving that, locally and under some transversality conditions, any solution can be obtained in implicit form via the generalised hodograph method.  Tsarev's theory of integrability has been extended to regular F-systems in the papers \cite{LPVG1,LPVG2,LPVG3}. In the general case, the linear system for the symmetries and the linear system for densities of conservation laws cannot be written in Darboux form. In this case the conditions \eqref{shc-intro} can be interpreted as involutivity conditions in the sense of Cartan.

\subsection{Beyond the regular case: cyclic F-manifolds. Main results}\label{SubsectIntroMainResult}
The aim of this paper is going beyond the regular case considering F-manifolds equipped by a  cyclic vector field  (\emph{cyclic F-manifolds}). We can summarize  the main results of the paper as follows:
\begin{itemize}
\item We generalise the theory of integrability to F-systems associated with a cyclic F-manifold (\emph{cyclic F-systems}). This includes
\begin{enumerate}
\item The construction of a connection uniquely defined by the F-system and encoding its integrability properties. 
\item The study of the linear system of the symmetries and of the generalised hodograph method for this class of systems in the analytic setting. This part relies on the iterated application of the Cauchy-Kowalevski theorem, in the spirit of Cartan-K\"ahler theory.
\item  In the analytic setting,  we prove the existence of a family of symmetries providing the unique local analytic solution of the Cauchy problem through the generalised hodograph method.
\end{enumerate}
\item We prove that if a cyclic F-system \eqref{FsysIntro} admits a Hamiltonian structure of hydrodynamic type --- local or non-local --- then it is integrable, i.e.\ its associated connection $\tn$ satisfies \eqref{shc-intro}. The proof relies on a relation between F-manifolds with compatible connection
 and Riemannian F-manifolds generalising a result of \cite{ABLR}. This is a generalisation of  Novikov's Conjecture.
\item In the analytic setting, we prove that an F-system is integrable if and only if it is locally conservative. In the regular analytic case this means that an integrable system of hydrodynamic type defined by a $(1,1)$-tensor field with vanishing Haantjies tensor is integrable if and only if it is locally conservative.  This is a generalisation of S\'evennec's result on integrability of conservative systems admitting Riemann invariants. 
\end{itemize}

\noindent The restriction to the analytic (or holomorphic) category in the results above is essential rather than technical. The symmetry and hodograph theorems rest on the Cauchy--Kowalevski theorem, which is the only tool that treats all cyclic F-systems uniformly.

\medskip

\noindent\textbf{Remarks on notation}. We work indifferently in the category of $C^{\infty}$ manifolds, $C^{\omega}$ real analytic manifolds, or complex (holomorphic) manifolds. $\mathbb{K}=\mathbb{R}\text{\,or\,}\mathbb{C}$ depending on the context.  By a metric we mean a non-degenerate symmetric bilinear pairing on sections of $\T$ --- pseudo-Riemannian in the real case, complex-valued in the holomorphic case. We denote by $\O$ the structure sheaf of $M$, by $\T$ the tangent sheaf, by $\T^*$ the cotangent sheaf, and by $\OO^{\bullet}=\bigoplus_{k=0}^n\OO^k$ the graded sheaf of forms. The symbol $\circ$ always denotes the commutative associative product on $\T$ defined in \eqref{defFmani}, while $\cdot$ denotes the composition of morphisms. All sheaves involved are locally free of finite rank, so the reader may safely think in terms of vector bundles. The one exception is Theorem~\ref{thm:main-restated} in Section~\ref{SectionCKsymm}, which relies on the Cauchy--Kowalevski theorem and therefore requires the real analytic or the holomorphic setting; the proof works verbatim in both.

\medskip

\noindent\textbf{Acknowledgements}. The authors are thankful to GNFM (Gruppo Nazionale di Fisica Matematica) for supporting activities that contributed to the research reported in this paper. P. Lorenzoni is supported by funds of INFN (Istituto Nazionale di Fisica Nucleare) under IS-CSN4 Mathematical Methods of Nonlinear Physics.  We thank Sonmez \c{S}ahuto\u{g}lu for useful discussions.

\section{F-manifolds}\label{SectionFmnfs}
F-manifolds have been introduced by Hertling and Manin in \cite{HM}.
\begin{definition}\label{defFmani}
An \emph{F-manifold} is a manifold $M$ equipped with
\begin{itemize}
\item a morphism of sheaves $\circ:\T\times\T\rightarrow \T$  which is $\O$-bilinear, and it gives rise to a commutative associative product on $\T$. Moreover, it satisfies the following identity:
\[
\mathcal{L}_{X\circ Y} \circ=X\circ (\mathcal{L}_Y \circ) +Y\circ (\mathcal{L}_X\circ ),
\]
which can be shown to be equivalent (using commutativity and associativity of $\circ$ to re-order the factors) to
\begin{align}
&[X\circ Y,W\circ Z]-[X\circ Y, Z]\circ W-[X\circ Y, W]\circ Z\label{HMeq1free}\\
&-X\circ [Y, Z \circ W]+X\circ [Y, Z]\circ W +X\circ [Y, W]\circ Z\notag\\
&-Y\circ [X,Z\circ W]+Y\circ [X,Z]\circ W+Y\circ [X, W]\circ Z=0,\notag
\end{align}
for all local sections  $X,Y,W, Z$ of $\T$, where $[X,Y]$ is the Lie bracket and $\mathcal{L}_X$ is the Lie derivative. 
\newline
\item A distinguished global section $e$ of $\T$ such that 
\[e\circ X=X\] 
for all local sections $X$ of $\T$. This acts like the unit of $\circ$.
\end{itemize}
\end{definition}
In any coordinate system, the morphism $\circ$ can be expressed using its structure functions $c^k_{ij}$ via the formula $\partial_i\circ\partial_j=c^k_{ij}\partial_k$, where $\{\partial_1, \dots, \partial_{n}\}$, with $n=\text{dim}(M)$ is a local basis of $\T$ associated to the chosen coordinate system. Commutativity of $\circ$ translates into the equation 
\begin{equation}\label{comm1.eq}
X\circ Y=Y\circ X \text{ for all local sections $X,Y\in \T$ or } \;\; c^k_{ij}=c^k_{ji},
\end{equation}
while the associativity of $\circ$ is equivalent to 
\begin{equation}\label{asso1.eq}
(X\circ Y)\circ Z=X\circ (Y \circ Z) 
\end{equation}
for all local sections $X,Y\in \T$ or 
\begin{equation}\label{asso1.eqBis}
c^i_{sj}c^s_{km}=c^i_{sk}c^s_{jm}.
\end{equation}
Of course, since $X\circ \in \T\otimes_{\O}\OO^1$, we can also view the morphism $\circ$ as $\circ: \T\rightarrow \T\otimes_{\O}\T^*$.

Usually F-manifolds are equipped with additional structures. 
\begin{definition}\label{defFwithE}
An \emph{F-manifold with Euler vector field} is an F-manifold $M$ equipped with a global section of $\T$ satisfying
$$\mathcal{L}_E \circ=\circ.$$
\end{definition}

Other additional structures have been introduced motivated by the theory of integrable systems of hydrodynamic type \cite{LPR}. There the additional datum is a connection $\nabla$ satisfying suitable compatibility conditions coming from integrability
  of the associated integrable hierarchies. 

In particular, the additional geometric structure capturing integrability properties of F-systems  
\beq\label{Fsys}
u^i_t=(X\circ u_x)^i=c^i_{jk}X^ju^k_x,\qquad i\inw,
\eeq
is a torsionless  connection $\tnabla$ associated with the vector field $X$ defining  the system.

Under some additional assumptions on the product $\circ$ and on the  vector field $X$ this connection is uniquely defined. These assumptions are:
\newline
\newline
\emph{1. The regularity of the product}. 
\begin{definition}
Let $(M, \circ, e)$ be an F-manifold of dimension $n$.  The product $\circ$ is called regular at point $p$ if there exists an open neighborhood $U$ of $p$ and local coordinates (David-Hertling coordinates,  see \cite{DH}),
\[
        u^{i(\alpha)},
        \qquad
        \alpha=1,\dots,r,
        \qquad
        i=1,\dots,m_\alpha,
\]
with $\sum_{\alpha=1}^r m_{\alpha}=n$, 
in which the identity has components 
\[
        e^{i(\alpha)}=\delta^i_1
\]
and the product has the block form
\[
        \partial_{i(\alpha)}\circ \partial_{j(\beta)}=0
        \quad\text{if }\alpha\neq\beta,
\]
while, inside the block labelled by $\alpha$,
\begin{equation}\label{eq: product block}
        \partial_{i(\alpha)}\circ\partial_{j(\alpha)}
        =
        \begin{cases}
        \partial_{(i+j-1)(\alpha)}, & i+j-1\leq m_\alpha,\\
        0, & i+j-1>m_\alpha.
        \end{cases}
\end{equation}
Equivalently,
\[
        TM=\bigoplus_{\alpha=1}^r \mathcal A_\alpha,
        \qquad
        \mathcal A_\alpha
        :=
        \operatorname{span}\{\partial_{1(\alpha)},\dots,\partial_{m_\alpha(\alpha)}\},
\]
and the product is the direct sum of truncated polynomial algebras.
\end{definition}

\emph{2. The non-degeneracy of the vector field $X$}. 
\begin{definition}
 We call a vector field $X$ \emph{regularly non-degenerate} if in the David-Hertling coordinates 
  \beq\label{eq: regularly non-degenerate} X^{1(\alpha)}\neq X^{1(\beta)} \,  \text{ for all } \alpha\neq \beta,  \alpha, \beta \in \{1, \dots, r\} \, \text{  and } \, X^{2(\alpha)}\neq 0 \,  \text{ for every block with } m_{\alpha}\geq 2.\eeq
 \end{definition}

Under the above Assumptions $1$ and $2$,  there exists a unique torsionless connection $\tnabla$ satisfying 
\begin{equation}\label{eq: characterization of nable tilde}\tnabla e=0,\qquad d_{\tnabla}(X\circ)=0,\end{equation}
where the morphism $d_{\tnabla}:\T\otimes_{\O}\OO^k\rightarrow  \T\otimes_{\O}\OO^{k+1}$ is usually called exterior covariant derivative of vector-valued differential forms and is defined via:
\begin{eqnarray*}
(d_{\tnabla} \omega)(X_0, \dots, X_k)&=&\sum_{i=0}^k (-1)^i\tnabla _{X_i}(\omega(X_0, \dots, \hat{X}_i, \dots, X_k))+\\
&&\sum_{0\leq i<j\leq k}(-1)^{i+j}\omega([X_i, X_j], X_0, \dots, \hat{X}_i, \dots, \hat{X}_j, \dots X_k).
\end{eqnarray*}
 
 In particular, the following holds:
  \begin{theorem}\label{thm: LPvG1}\cite{LPVG1,LPVG3}
  Consider  a regular $n$-dimensional F-manifold, with product $\circ$, and an associated system of hydrodynamic type:
  \beq\label{SHS2}
u^i_{t}=(X\circ u_x)^i=c^i_{jk}X^ju^k_x.
\eeq
  If $X$ is regularly non-degenerate,
there exists a unique torsionless connection $\tnabla$ such that  equations \eqref{eq: characterization of nable tilde} hold. Moreover
\begin{equation}\label{nablaccomp2.eq}(\tnabla_X \circ)(Y, Z)=(\tnabla_Y \circ)(X,Z),
\end{equation}
for all local sections $X,Y, Z$ of $\T$. 
 \end{theorem}

 Let us now recall some important facts about the integrability of an F-system that will be reflected in the properties of the connection $\tnabla$. We refer for details to \cite{LPVG1,LPVG2,LPVG3}. Symmetries of the system \eqref{SHS2} are F-systems 
 \beq\label{Fsym}
u^i_\tau=(Y\circ u_x)^i=c^i_{jk}Y^ju^k_x,\qquad i\inw,
\eeq
compatible with  \eqref{SHS2} (i.e.\ $u_{t\tau}=u_{\tau t}$). It turns out that these are defined by vector fields $Y$ satisfying the condition
\beq\label{EqSym}
d_{\tnabla}(Y\circ)=0.
\eeq

 In the presence of symmetries, the connection $\tnabla$ associated to a system of the form \eqref{SHS2}, with $X$ regularly non-degenerate, not only satisfies the conditions \eqref{eq: characterization of nable tilde} and \eqref{nablaccomp2.eq}, but also a much deeper geometric constraint. This geometric constraint is singled out in the next Theorem (see equation \eqref{rc-intri}) and it is an incarnation of the notion of F-manifold with compatible connection as we will see in the next Section. 

 \begin{theorem}\label{LPVG_mainTh}(\cite{LPVG1})
	Let $\{X_{(0)},\dots,X_{(n-1)}\}$ be a set of linearly independent local vector fields on an $n$-dimensional regular F-manifold $(M,\circ,e)$. Assume that the corresponding flows
	\begin{align}
		{\bf  u}_{t_i}=X_{(i)}\circ {\bf u}_x,\qquad i\in\{0,\dots,n-1\},
		\notag
	\end{align}
	pairwise commute, and that there exists a local vector field $X\in\{X_{(0)},\dots,X_{(n-1)}\}$ which is regularly non-degenerate.  Then the connection $\tnabla$ associated to \eqref{SHS2} satisfies the conditions \eqref{eq: characterization of nable tilde},  \eqref{nablaccomp2.eq} and furthermore the following: 
	\beq\label{rc-intri}
Z\circ R^{\tnabla}(W,Y)(X)+W\circ R^{\tnabla}(Y,Z)(X)+Y\circ R^{\tnabla}(Z,W)(X)=0,
\eeq
for all local sections $X$, $Y$, $Z$, $W$ of $\T$, where $R^{\tnabla}(Z,W)(X)$ is the Riemann curvature of $\tnabla$.
\end{theorem}
Moreover the compatibility condition of the  system \eqref{Fsys} coincides with condition
 \eqref{rc-intri}.   
 
\begin{remark}	
Using Bianchi identity condition \eqref{rc-intri} can be written in the equivalent  form  
\begin{equation}\label{rc-intri-2}
R^{\tnabla}(Y,Z)(X\circ W)+R^{\tnabla}(X,Y)(Z\circ W)+R^{\tnabla}(Z,X)(Y\circ W)=0,
\end{equation}
for all local vector fields $X$, $Y$, $Z$, $W$ (see \cite{LPR} for  details).
\end{remark}

Let us point out that,  in local coordinates,  conditions \eqref{rc-intri} and \eqref{rc-intri-2} read \eqref{shc-intro}.

\subsection{Cyclic vector fields and David-Hertling coordinates}

For a vector field $X$ set \[
        A:=C_X=X\circ.
\]
\begin{definition}
We call $X$ \emph{cyclic} if $e$ is cyclic with respect to $A$, i.e. if
\[
        e,Ae,A^2e,\dots,A^{n-1}e
        \quad\text{is a local frame of }TM.
\]
Since $Ae=X$, this is the same as requiring
\[
        e,X,X^{\circ 2},\dots,X^{\circ(n-1)}
        \quad\text{to be a local frame.}
\]
We call an F-manifold \emph{cyclic} if it admits a \emph{cyclic} vector field.
\end{definition}

Regular F-manifolds are F-manifolds equipped with a cyclic Euler vector field. In this case cyclic vector fields  coincide with regularly non-degenerate vector fields (one direction of the following was already proved in \cite{LPVG1}):

\begin{proposition}
\label{prop:DH-cyclicity-equivalence}
In David-Hertling coordinates   a vector field $X$ is cyclic if and only if it is regularly non-degenerate.
\end{proposition}

\begin{proof} Fix one block $\mathcal A_\alpha:=\operatorname{span}\{\partial_{1(\alpha)},\dots,\partial_{m_\alpha(\alpha)}\},$ and write
\[
        e_\alpha:=\partial_{1(\alpha)}.
\]
For $m_\alpha>1$, define
\[
        N_\alpha:=C_{\partial_{2(\alpha)}}.
\]
Then, using \eqref{eq: product block}
\[
        N_\alpha e_\alpha=
        \partial_{2(\alpha)},
        \quad
        N_\alpha^2 e_\alpha=
        \partial_{3(\alpha)},
        \quad\dots,\quad
        N_\alpha^{m_\alpha-1}e_\alpha=
        \partial_{m_\alpha(\alpha)},
\]
and
\[
        N_\alpha^{m_\alpha}=0.
\]
The restriction of $A=C_X$ to $\mathcal A_\alpha$ is
\[
        A_\alpha
        =
        X^{1(\alpha)}\operatorname{Id}
        +X^{2(\alpha)}N_\alpha
        +X^{3(\alpha)}N_\alpha^2
        +\cdots
        +X^{m_\alpha(\alpha)}N_\alpha^{m_\alpha-1}.
\]
For $m_\alpha=1$, this simply reads
\[
        A_\alpha=X^{1(\alpha)}\operatorname{Id}.
\]

Assume $m_\alpha>1$. Then
\[
        A_\alpha-X^{1(\alpha)}\operatorname{Id}
        =N_\alpha q_\alpha(N_\alpha),
\]
where
\[
        q_\alpha(t)
        =X^{2(\alpha)}+X^{3(\alpha)}t+\cdots
        +X^{m_\alpha(\alpha)}t^{m_\alpha-2}.
\]
If $X^{2(\alpha)}\neq0$, then $q_\alpha(0)\neq0$, hence
$q_\alpha(N_\alpha)$ is invertible, since it can be written as a non-zero multiple of the identity plus a nilpotent term. Since $q_\alpha(N_\alpha)$ commutes with $N_\alpha$, the product $A_\alpha-X^{1(\alpha)}\operatorname{Id}=N_\alpha q_\alpha(N_\alpha)$ is nilpotent of the same order $m_\alpha$ as $N_\alpha$, and
\[
        (A_\alpha-X^{1(\alpha)}\operatorname{Id})^{m_\alpha-1}e_\alpha
        =\bigl(N_\alpha q_\alpha(N_\alpha)\bigr)^{m_\alpha-1}e_\alpha
        =N_\alpha^{m_\alpha-1}\,q_\alpha(N_\alpha)^{m_\alpha-1}e_\alpha,
\]
where we used that $N_\alpha$ commutes with $q_\alpha(N_\alpha)$.  Writing
$q_\alpha(N_\alpha)=q_\alpha(0)\operatorname{Id}+N_\alpha\,r(N_\alpha)$ for a suitable polynomial $r$, the expansion of $q_\alpha(N_\alpha)^{m_\alpha-1}$ is the sum of the constant term $q_\alpha(0)^{m_\alpha-1}\operatorname{Id}$ and terms each carrying at least one factor of $N_\alpha$.  Multiplying by $N_\alpha^{m_\alpha-1}$ and using $N_\alpha^{m_\alpha}=0$ annihilates all of the latter terms, whence
\[
        (A_\alpha-X^{1(\alpha)}\operatorname{Id})^{m_\alpha-1}e_\alpha
        =q_\alpha(0)^{m_\alpha-1}N_\alpha^{m_\alpha-1}e_\alpha
        =\bigl(X^{2(\alpha)}\bigr)^{m_\alpha-1}\partial_{m_\alpha(\alpha)}\neq 0,
\]
since $q_\alpha(0)=X^{2(\alpha)}\neq0$ and $N_\alpha^{m_\alpha-1}e_\alpha=\partial_{m_\alpha(\alpha)}$.
Thus $e_\alpha$ is a cyclic vector for the nilpotent operator $A_\alpha-X^{1(\alpha)}\operatorname{Id}$, of nilpotency order $m_\alpha$.  Since $A_\alpha$ and $A_\alpha-X^{1(\alpha)}\operatorname{Id}$ differ by the scalar operator $X^{1(\alpha)}\operatorname{Id}$, the two generate the same cyclic subspace from any vector.  In particular $e_\alpha$ is cyclic for $A_\alpha$ as well, so that
\[
        e_\alpha,A_\alpha e_\alpha,
        \dots,A_\alpha^{m_\alpha-1}e_\alpha
\]
are linearly independent.

Conversely, if $X^{2(\alpha)}=0$, then
\[
        A_\alpha-X^{1(\alpha)}\operatorname{Id}
        \in N_\alpha^2\mathbb{K}[N_\alpha].
\]
Thus the cyclic span of $e_\alpha$ misses the vector
$N_\alpha e_\alpha=\partial_{2(\alpha)}$. Hence $e_\alpha$
is not cyclic for $A_\alpha$. We have proved that, on a block,
$e_\alpha$ is cyclic for $A_\alpha$ if and only if either
$m_\alpha=1$, or $m_\alpha>1$ and $X^{2(\alpha)}\neq0$.

Now consider the direct sum
\[
        A=\bigoplus_{\alpha=1}^r A_\alpha,
        \qquad
        e=\sum_{\alpha=1}^r e_\alpha.
\]
The cyclic subspace generated by $e$ is the image of
\[
        \mathbb{K}[t]\longrightarrow TM,
        \qquad
        p(t)\longmapsto p(A)e.
\]
The dimension of the image of $p\mapsto p(A)e$ equals the degree of the minimal polynomial $\mu_{A,e}$ of the vector $e$, that is, of the monic generator of the ideal $\{p\in\mathbb{K}[t]:p(A)e=0\}$.  Since the blocks $\mathcal A_\alpha$ are independent $A$-invariant subspaces and $p(A)e_\alpha=p(A_\alpha)e_\alpha\in\mathcal A_\alpha$, one has $p(A)e=0$ if and only if $p(A_\alpha)e_\alpha=0$ for every $\alpha$; hence $\mu_{A,e}=\lcm_\alpha\mu_{A_\alpha,e_\alpha}$, where $\mu_{A_\alpha,e_\alpha}$ denotes the minimal polynomial of $e_\alpha$ relative to $A_\alpha$ and $\lcm$ stands for least common multiple.  In general, each $\mu_{A_\alpha,e_\alpha}$ is a power of $t-X^{1(\alpha)}$: the operator
$A_\alpha-X^{1(\alpha)}\operatorname{Id}=N_\alpha q_\alpha(N_\alpha)$ is nilpotent on the
block whether or not $X^{2(\alpha)}$ vanishes, since
$\bigl(N_\alpha q_\alpha(N_\alpha)\bigr)^{m_\alpha}
=N_\alpha^{m_\alpha}q_\alpha(N_\alpha)^{m_\alpha}=0$.
Hence
\[
        \mu_{A_\alpha,e_\alpha}(t)=(t-X^{1(\alpha)})^{k_\alpha},
        \qquad
        1\le k_\alpha\le m_\alpha,
\]
where $k_\alpha\geq 1$ because $e_\alpha\neq 0$. Moreover,
$k_\alpha=\deg\mu_{A_\alpha,e_\alpha}$ is the dimension of the cyclic subspace generated
by $e_\alpha$ in $\mathcal A_\alpha$, so the blockwise analysis above shows that
\[
        k_\alpha=m_\alpha
        \quad\Longleftrightarrow\quad
        m_\alpha=1 \ \text{ or } \ X^{2(\alpha)}\neq 0.
\]
Therefore
\[
        \mu_{A,e}=\lcm_{\alpha}\,(t-X^{1(\alpha)})^{k_\alpha},
\]
and since the least common multiple of powers of linear polynomials is the product, over
the distinct roots, of the highest occurring powers,
\[
        \deg\mu_{A,e}
        =\sum_{\lambda\in\{X^{1(1)},\dots,X^{1(r)}\}}
        \max\bigl\{k_\alpha : X^{1(\alpha)}=\lambda\bigr\}
        \;\leq\;\sum_{\alpha=1}^r k_\alpha
        \;\leq\;\sum_{\alpha=1}^r m_\alpha=n.
\]
The first inequality is an equality if and only if the values $X^{1(\alpha)}$ are pairwise
distinct: if a value were attained by $p\geq 2$ blocks then, since every $k_\alpha\geq 1$,
the corresponding maximum would fall short of the sum of those $k_\alpha$ by at least
$p-1$. The second inequality is an equality if and only if $k_\alpha=m_\alpha$ for every
$\alpha$. Consequently $e$ is cyclic for $A$, i.e.\ $\deg\mu_{A,e}=n$, if and only if
$X^{1(\alpha)}\neq X^{1(\beta)}$ for all $\alpha\neq\beta$ and $X^{2(\alpha)}\neq 0$ for
every block with $m_\alpha\geq 2$, which is precisely condition
\eqref{eq: regularly non-degenerate}.

\end{proof}

\subsection{Cyclicity is weaker than regularity}
Proposition~\ref{prop:DH-cyclicity-equivalence} says that, once
David--Hertling coordinates already exist, cyclicity of $X$ is exactly the
usual  non-degeneracy condition. The intrinsic cyclicity condition, however, is
weaker than the existence of David--Hertling coordinates. It only requires
\[
        e,X,\dots,X^{\circ(n-1)}
        \quad\text{to be a local frame.}
\]
It does not require an Euler vector field, the regularity of multiplication by an
Euler vector field, or a David--Hertling normal form.

We introduce a minimal two-dimensional example. Let $M=\mathbb R^2$
with coordinates $(t,s)$ and write
\[
        e_0:=\partial_t,
        \qquad
        v:=\partial_s.
\]
Define
\[
        e_0\circ e_0=e_0,
        \qquad
        e_0\circ v=v,
        \qquad
        v\circ v=s\,v.
\]
The unit is $e=e_0$. The product is commutative and associative. It is also
an F-manifold product. Indeed, using
\[
        P_U(Y,Z):=[U,Y\circ Z]-[U,Y]\circ Z-Y\circ[U,Z],
\]
the Hertling--Manin identity can be written as
\[
        P_{U\circ V}(Y,Z)=U\circ P_V(Y,Z)+V\circ P_U(Y,Z).
\]
All cases involving the unit $e_0$ are immediate. The only nontrivial check is
$U=V=Y=Z=v$. Since
\[
        P_v(v,v)=[v,s v]=v
\]
and
\[
        P_{v\circ v}(v,v)=P_{s v}(v,v)=2s v
        =2v\circ P_v(v,v),
\]
the Hertling--Manin identity holds.

Now take
\[
        X:=v.
\]
Since $n=2$, cyclicity means only that $e,X$ is a frame. In the ordered frame
$(e_0,v)$, the cyclicity matrix is the identity matrix:
\[
        [e,X]=[e_0,v].
\]
Thus $X$ is cyclic everywhere, including at points with $s=0$.

However, no David--Hertling coordinate system can exist on any neighbourhood
of a point with $s=0$. At $s=0$ the tangent algebra is
\[
        \mathbb K[\xi]/(\xi^2),
\]
because $v\circ v=0$. For $s\neq0$, the vectors
\[
        p_1:=\frac1s v,
        \qquad
        p_0:=e_0-\frac1s v
\]
are idempotents:
\[
        p_0\circ p_0=p_0,
        \qquad
        p_1\circ p_1=p_1,
        \qquad
        p_0\circ p_1=0.
\]
Hence the tangent algebra is semisimple\footnote{Throughout, \emph{semisimple} is understood in the split sense: the tangent algebra is isomorphic to $\mathbb K^n$, equivalently it admits a frame of idempotents $\epsilon_1,\dots,\epsilon_n$ with $\epsilon_i\circ\epsilon_j=\delta_{ij}\epsilon_i$ and $e=\sum_{i=1}^n\epsilon_i$. Over a field that is not algebraically closed this is strictly stronger than the vanishing of the nilradical (the absence of nonzero nilpotents). 
} for $s\neq0$, but is the algebra of
dual numbers at $s=0$. The block type therefore jumps from $(2)$ at $s=0$ to
$(1,1)$ for $s\neq0$. Since David--Hertling coordinates have a fixed block
form on their coordinate neighbourhood, such coordinates cannot exist around
$s=0$. This shows that cyclicity is weaker than the joint assumption
``David--Hertling coordinates exist and $X$ is regularly non-degenerate''. 
The same phenomenon persists in higher dimension: one can easily cook up much more complicated examples. 

Let us point out that on the semisimple locus of an F-manifold product the existence of cyclic vector fields is automatic locally.
Indeed,  let
\[
        \epsilon_1,\dots,\epsilon_n
\]
be the local idempotent frame, so that
\[
        \epsilon_i\circ\epsilon_j=\delta_{ij}\epsilon_i,
        \qquad
        e=\sum_{i=1}^n \epsilon_i.
\]
It is immediate to check that a vector field 
\[
        X=\sum_{i=1}^n X^i\epsilon_i,
\]
is cyclic if and only if $X^i\neq X^j\text{ for all }i\neq j.$

Thus cyclicity of $X$ allows one to pass from the semisimple locus to the non-semisimple locus seamlessly, even when David-Hertling coordinates do not exist. Saying otherwise: in the regular David–Hertling setting, regularity of the Euler field gives the normal form, while Proposition~\ref{prop:DH-cyclicity-equivalence} identifies regular non-degeneracy of $X$ with cyclicity of X in that normal form.

\section{Distinguished connections}
In \cite{LPVG1} it was proved that given a regularly non-degenerate vector field $X$ written in David-Hertling coordinates, there exists a unique torsionless connection $\tnabla$ satisfying the following conditions:
\[ \tnabla e=0,  \quad d_{\tnabla}(X\circ)=0, \quad (\tnabla_W \circ)(Y,Z)=(\tnabla_Y \circ)(W,Z),  \]
for all local sections $W,Y,Z$ of $\T$.

In this Section we vastly expand this result,  using only the assumption that $X$ is cyclic, which allows us to treat uniformly the semisimple and the non-semisimple locus,  even when David-Hertling coordinates do not exist. 
First a definition:
\begin{definition}
Let $(M,\circ,e)$ be an F-manifold and let $X$ be a cyclic vector field.  A torsionless connection $\nabla$ is called \emph{distinguished} with respect to $X$ if 
\[ d_{\nabla}(X\circ)=0,  \quad (\nabla_W\circ)(Y,Z)=(\nabla_Y\circ)(W,Z)
        \qquad\forall\,W,Y,Z \text{ local sections of } \T.\]
\end{definition}

We have the following 
\begin{theorem}\label{thm: existence and uniqueness of cyclic connections}
Let $(M,\circ,e)$ be an F-manifold and let $X$ be cyclic. 
Then locally there exists a unique  torsionless connection $\tnabla$ distinguished with respect to $X$ such that
    \[    \tnabla e=0.\]
We will call this connection the \emph{natural connection} associated with $X$.
\end{theorem}
Subsections 3.1, 3.2, 3.3 below are dedicated to the proof of this Theorem.

\subsection{Cyclic linear algebra and the unique flat-unit connection}
Let $V$ be an $n$-dimensional vector space, let $A\in\End(V)$, and assume that
$e\in V$ is cyclic for $A$, so that
\[
        e,Ae,\ldots,A^{n-1}e
\]
is a basis of $V$. For
\[
        S\in S^2V^*\otimes V
\]
define
\[
        (\mathcal M_A S)(Y,Z):=S(Y,AZ)-S(Z,AY).
\]
Then $\mathcal M_A S\in\Lambda^2V^*\otimes V$.

\begin{lemma}[Cyclic isomorphism]
\label{lem:cyclic-iso}
The restriction
\[
        \mathcal M_A:
        \{S\in S^2V^*\otimes V\mid S(e,\cdot)=0\}
        \longrightarrow
        \Lambda^2V^*\otimes V
\]
is an isomorphism.
\end{lemma}

\begin{proof}
We first prove the scalar statement. Define
\[
        B:=\{b\in S^2V^*\mid b(e,\cdot)=0\}.
\]
Consider the scalar map
\[
        \mathcal M_A^{\mathrm{scal}}:B\longrightarrow\Lambda^2V^*,
        \qquad
        (\mathcal M_A^{\mathrm{scal}}b)(Y,Z)
        :=b(Y,AZ)-b(Z,AY).
\]
Suppose $b\in B$ lies in the kernel. Then
\[
        b(Y,AZ)-b(Z,AY)=0
        \qquad\forall\,Y,Z\in V.
\]
Since $b$ is symmetric, this is equivalent to
\[
        b(AY,Z)=b(Y,AZ).
\]
Thus $A$ is self-adjoint with respect to $b$. Therefore, for all $i,j\geq0$,
\[
        b(A^ie,A^je)=b(e,A^{i+j}e).
\]
But $b(e,\cdot)=0$, so
\[
        b(A^ie,A^je)=0
        \qquad\forall\,i,j\geq0.
\]
Since
\[
        e,Ae,\ldots,A^{n-1}e
\]
is a basis of $V$, this implies $b=0$. Hence $\mathcal M_A^{\mathrm{scal}}$ is
injective.
Now define the map
\[\iota_e: S^2V^*\rightarrow V^*,  \quad \iota_e: b \mapsto b(e,\cdot),\] 
which is surjective since $e$ is non-zero. Thus $\dim B=\dim \ker(\iota_e)=\dim S^2V^*-\dim V^*$
and thus the dimensions of $B$ and $\Lambda^2V^*$ agree:
\[
        \dim B
        =\dim S^2V^*-n
        =\frac{n(n+1)}2-n
        =\frac{n(n-1)}2
        =\dim\Lambda^2V^*.
\]
Therefore
\[
        \mathcal M_A^{\mathrm{scal}}:B\longrightarrow\Lambda^2V^*
\]
is an isomorphism.

We now pass to the vector-valued statement. Fix a basis $f_1,\dots,f_n$ of $V$ and decompose
\[
        S=\sum_{k=1}^n S^k\otimes f_k,
        \qquad
        S^k\in S^2V^*,
\]
so that $S\in\{S\in S^2V^*\otimes V\mid S(e,\cdot)=0\}$ if and only if $S^k(e,\cdot)=0$ for every $k$. By definition of $\mathcal M_A$,
\[
        (\mathcal M_A S)(Y,Z)
        =\sum_{k=1}^n\bigl(\mathcal M_A^{\mathrm{scal}} S^k\bigr)(Y,Z)\otimes f_k.
\]
Hence $\mathcal M_A S=0$ if and only if $\mathcal M_A^{\mathrm{scal}} S^k=0$ for every $k$, which by the scalar statement above forces $S^k=0$ for every $k$, hence $S=0$. Thus $\mathcal M_A$ is injective. Since the source and target both have dimension $\tfrac{n^2(n-1)}{2}$, the map $\mathcal M_A$ is an isomorphism.
\end{proof}

\begin{proposition}[Existence and uniqueness with flat unit]
\label{prop:flat-unit}
There exists a unique torsionless connection $\tnabla$ such that
\[
        \tnabla e=0,
        \qquad
        d_{\tnabla}A=0.
\]
\end{proposition}

\begin{proof}
We work locally on the open set where $X$ is cyclic, i.e. where
\[
        e,Ae,\ldots,A^{n-1}e
\]
is a frame.

Choose an auxiliary torsionless connection $\bar\nabla$. Since both $\bar\nabla$ and $\tnabla$ are torsionless,
their difference is a symmetric tensor:
\begin{equation}\label{eq: prop1.2}
        \tnabla_YZ
        =\bar\nabla_YZ+S(Y,Z),
        \qquad
        S\in\Gamma(S^2T^*M\otimes TM).
\end{equation}

We are going to show that $\tnabla$ is uniquely fixed by the conditions $\tnabla e=0$ and $d_{\tnabla}A=0$. First we  impose the flat-unit condition
\[
        \tnabla e=0,
\]
which, due to \eqref{eq: prop1.2} and the symmetry of $S$, is equivalent to
\[
        S(e,Y)=-\bar\nabla_Ye.
\]
A symmetric tensor $S_0$ satisfying
\[
        S_0(e,Y)=-\bar\nabla_Ye.
\]
exists locally. Indeed, since $e$ is nowhere zero, choose a local one-form
$\theta$ such that
\[
        \theta(e)=1.
\]
Define
\[
        S_0(Y,Z)
        :=-\theta(Y)\bar\nabla_Ze-\theta(Z)\bar\nabla_Ye
        +\theta(Y)\theta(Z)\bar\nabla_ee.
\]
Then $S_0$ is symmetric and $S_0(e,Y)=-\bar\nabla_Ye$.

Every symmetric tensor $S$ satisfying the flat-unit condition can therefore be written
uniquely as
\[
        S=S_0+S_1,
        \qquad
        S_1(e,\cdot)=0.
\]
Let
\[
        \mathcal E
        :=\{S_1\in S^2T^*M\otimes TM\mid S_1(e,\cdot)=0\}.
\]
The unknown part of the connection is now precisely $S_1\in\Gamma(\mathcal E)$,  which controls the freedom to satisfy the additional equation $d_{\tnabla}A=0$.  To study the latter equation,  we  first compute how $d_{\bar\nabla} A$ changes when the connection $\bar\nabla$ is changed by $S$.
Since
\[
        (\tnabla_YA)(Z)
        =(\bar\nabla_YA)(Z)+S(Y,AZ)-A(S(Y,Z)),
\]
and for any torsionless connection $\nabla$
\begin{align*}
(d_{\nabla}A)(Y,Z)&=\nabla_Y(AZ)-\nabla_Z(AY)-A([Y,Z])\\
                               &=(\nabla_Y A)(Z)-(\nabla_Z A)(Y),
\end{align*}
we get
\begin{align*}
        (d_{\tnabla}A)(Y,Z)
        &=(d_{\bar\nabla}A)(Y,Z)
        +S(Y,AZ)-S(Z,AY)        \\
        &\quad
        -A(S(Y,Z)-S(Z,Y)).
\end{align*}
Because $S$ is symmetric, the final term vanishes. Hence
\[
        (d_{\tnabla}A)(Y,Z)
        =(d_{\bar\nabla}A)(Y,Z)+S(Y,AZ)-S(Z,AY),
\]
or
\begin{equation}\label{eq: prop1.1}
        d_{\tnabla}A=d_{\bar\nabla}A+\mathcal M_A S,
\end{equation}
where $\mathcal M_A$ is defined in Lemma \ref{lem:cyclic-iso}.

The equation $d_{\tnabla}A=0$ becomes (see \eqref{eq: prop1.1})
\[
        d_{\bar\nabla}A+\mathcal M_A S_0+\mathcal M_A S_1=0,
\]
or equivalently
\begin{equation}\label{eq: prop1.3}
        \mathcal M_A S_1
        =-d_{\bar\nabla}A-\mathcal M_A S_0.
\end{equation}
The right-hand side is a section of $\Lambda^2T^*M\otimes TM$,  which is already determined,  so we are left to prove that an $S_1$ satisfying \eqref{eq: prop1.3} exists and is unique. 

By Lemma~\ref{lem:cyclic-iso}, applied fibrewise, the map
\[
        \mathcal M_A:\mathcal E\longrightarrow\Lambda^2T^*M\otimes TM
\]
is a vector bundle isomorphism on the cyclic locus. Therefore the equation
\[
        \mathcal M_A S_1
        =-d_{\bar\nabla}A-\mathcal M_A S_0
\]
has a unique smooth solution, namely
\[
        S_1
        =\mathcal M_A^{-1}
        \bigl(-d_{\bar\nabla}A-\mathcal M_A S_0\bigr).
\]
Smoothness follows because the cyclicity condition is open and the inverse bundle map
$\mathcal M_A^{-1}$ depends smoothly on the base point.

Let us also record the analytic version, which will be used below (see Theorem \ref{thm:main-restated}). If $(M,\circ,e)$
and $X$ are real analytic , choose the auxiliary torsionless
connection $\bar\nabla$ analytically (for instance, the flat connection in analytic local
coordinates) and choose $\theta$ analytically with $\theta(e)=1$; locally this is possible
because some component of the nowhere-zero analytic vector field $e$ is non-zero. Then
$S_0$, $d_{\bar\nabla}A$, and $\mathcal M_A S_0$ are analytic. In analytic local
frames of $\mathcal E$ and $\Lambda^2T^*M\otimes TM$, the bundle map $\mathcal M_A$ is
represented by an analytic matrix whose determinant is nowhere zero on the cyclic locus,
by Lemma~\ref{lem:cyclic-iso}. Its inverse matrix is therefore analytic. Hence
\[
        S_1
        =\mathcal M_A^{-1}
        \bigl(-d_{\bar\nabla}A-\mathcal M_A S_0\bigr)
\]
is analytic, and so the connection constructed below is analytic as well.

Define
\[
        \tnabla_YZ:=\bar\nabla_YZ+S_0(Y,Z)+S_1(Y,Z).
\]
By construction, $\tnabla$ is torsionless, because $S_0+S_1$ is symmetric. Moreover,
$\tnabla e=0$, because
\[
        (S_0+S_1)(Y,e)=-\bar\nabla_Ye,
\]
and $d_{\tnabla}A=0$, because $S_1$ solves equation \eqref{eq: prop1.3} by construction. This proves
existence of $\tnabla.$

It remains to prove uniqueness. Suppose that $\nabla'$ is another torsionless connection
satisfying
\[
        \nabla' e=0,
        \qquad
        d_{\nabla'}A=0.
\]
Write
\[
        \nabla'=\tnabla+D,
        \qquad
        D\in\Gamma(S^2T^*M\otimes TM).
\]
Since both $\nabla'$ and $\tnabla$ have flat unit, $D(e,\cdot)=0$. Since both satisfy
$d_\nabla A=0$, their difference gives (see \eqref{eq: prop1.1})
\[
        \mathcal M_A D=0.
\]
But $D\in\Gamma(\mathcal E)$, and
\[
        \mathcal M_A:\mathcal E\longrightarrow\Lambda^2T^*M\otimes TM
\]
is an isomorphism. Hence $D=0$. Therefore $\nabla'=\tnabla$.
\end{proof}
Thus $\tnabla$ is the unique torsionless connection that satisfies $\tnabla e=0$ and $d_{\tnabla} X\circ =0$. 
Subsections 3.2 and 3.3 are dedicated to showing that $\tnabla$ satisfies the symmetry condition with respect to $\circ$ that makes it into a distinguished connection with respect to $X$. First we need a preliminary subsection on the Hertling-Manin condition.

\subsection{Covariant Hertling--Manin and associativity identities}
Let $\nabla$ be torsionless and define
\[
        T(U,Y,Z):=(\nabla_U\circ)(Y,Z).
\]
Since $\circ$ is commutative,
\begin{equation}\label{eq:Tcomm-two last entries}
        T(U,Y,Z)=T(U,Z,Y).
\end{equation}

\begin{lemma}[Covariant Hertling--Manin identity]
\label{lem:covHM}
For all vector fields $P,Q,R,S$,
\begin{align}
0={}&T(R\circ S,P,Q)-T(P\circ Q,R,S)\nonumber\\
&+T(P,R,S)\circ Q+T(Q,R,S)\circ P\nonumber\\
&-T(S,P,Q)\circ R-T(R,P,Q)\circ S.
\label{eq:covHM}
\end{align}
\end{lemma}

\begin{proof}
The proof is given in \cite{LPVG1} (Lemma 3.7). 
\end{proof}

\begin{lemma}[Covariant associativity identity]
\label{lem:covassoc}
For all vector fields $U,V,W,Z$,
\begin{equation}
        T(U,V\circ W,Z)+T(U,V,W)\circ Z
        =T(U,V,W\circ Z)+V\circ T(U,W,Z).
\label{eq:covassoc}
\end{equation}
\end{lemma}

\begin{proof}
Apply $\nabla_U$ to the associativity identity
\[
        (V\circ W)\circ Z=V\circ(W\circ Z).
\]
Using the definition of $T$, the derivative of the left-hand side is
\begin{align*}
\nabla_U((V\circ W)\circ Z)
={}&T(U,V\circ W,Z)+T(U,V,W)\circ Z\\
&+(\nabla_UV\circ W)\circ Z+(V\circ\nabla_UW)\circ Z
 +(V\circ W)\circ\nabla_UZ.
\end{align*}
The derivative of the right-hand side is
\begin{align*}
\nabla_U(V\circ(W\circ Z))
={}&T(U,V,W\circ Z)+V\circ T(U,W,Z)\\
&+\nabla_UV\circ(W\circ Z)+V\circ(\nabla_UW\circ Z)
  +V\circ(W\circ\nabla_UZ).
\end{align*}
The last three terms on each side agree by commutativity and associativity of $\circ$.
Cancelling them gives \eqref{eq:covassoc}.
\end{proof}

\begin{lemma}[Consequences of a flat unit]
\label{lem:Te}
If $\nabla e=0$, then for all local vector fields $Y,Z$:
\[
        T(e,Y,Z)=0,
        \qquad
        T(Y,e,Z)=0.
\]
\end{lemma}

\begin{proof}
The equality $T(Y,e,Z)=0$ follows directly from $e\circ Z=Z$ and $\nabla e=0$:
\[
        T(Y,e,Z)=(\nabla_Y\circ)(e,Z)=\nabla_Y(e\circ Z)-(\nabla_Y e)\circ Z-e\circ \nabla_YZ=0.
\]
Now put $P=e$ in \eqref{eq:covHM}. The terms $T(R\circ S,e,Q)$, $T(S,e,Q)$, and
$T(R,e,Q)$ vanish by the first part, while $-T(e\circ Q,R,S)$ and
$T(Q,R,S)\circ e$ cancel. Hence
\[
        T(e,R,S)\circ Q=0
        \qquad\forall\,Q,R,S.
\]
Taking $Q=e$ gives $T(e,R,S)=0$.
\end{proof}

Now we are ready to prove that $\tnabla$ satisfies the symmetry condition with respect to the product:

\subsection{Compatibility with the product of the natural connection}
In this subsection we use the connection $\tnabla$ and the cyclic vector field $X$.  Thus
\[
        \tnabla e=0,
        \qquad
        d_{\tnabla} A=0,
        \qquad
        A=C_X.
\]
We keep the notation
\[
        T(U,Y,Z)=(\tnabla_U\circ)(Y,Z).
\]

\begin{lemma}
\label{lem:nablaX}
There exists a vector field $W_0$ such that
\[
        \tnabla_YX=Y\circ W_0
        \qquad\forall\,Y.
\]
In fact, $W_0=\tnabla_eX$. Consequently,
\[
        T(Y,X,Z)=T(Z,X,Y)
        \qquad\forall\,Y,Z.
\]
\end{lemma}

\begin{proof}
The equation $d_{\tnabla} A=0$ means
\[
        (\tnabla_YA)(Z)=(\tnabla_ZA)(Y)
        \qquad\forall\,Y,Z.
\]
Since $A=C_X$, 
\begin{align*}
        (\tnabla_YA)(Z)
        &=
        \tnabla_Y(X\circ Z)-A(\tnabla_YZ)\\
        &=
        T(Y,X,Z)+(\tnabla_YX)\circ Z,
\end{align*}
and similarly
\[
        (\tnabla_ZA)(Y)
        =
        T(Z,X,Y)+(\tnabla_ZX)\circ Y.
\]
Thus $d_{\tnabla} A=0$ gives the identity
\begin{equation}
        T(Y,X,Z)+(\tnabla_YX)\circ Z
        =
        T(Z,X,Y)+(\tnabla_ZX)\circ Y.
\label{eq:dA-expanded}
\end{equation}

Now, setting $Z=e$, \eqref{eq:dA-expanded} becomes
\[
        \tnabla_YX=(\tnabla_eX)\circ Y,
\]
since $T(Y,X,e)\stackrel{\eqref{eq:Tcomm-two last entries}}{=}T(Y,e,X)=0$ and $T(e,X,Y)=0$ by Lemma \ref{lem:Te} .
Thus the required vector field is
\[
        W_0:=\tnabla_eX,
\]
and
\[
        \tnabla_YX=Y\circ W_0
        \qquad\forall\,Y.
\]

We now substitute this formula back into the identity \eqref{eq:dA-expanded}, this time
for arbitrary $Y$ and $Z$. Namely, we use
\[
        \tnabla_YX=(\tnabla_eX)\circ Y,
        \qquad
        \tnabla_ZX=(\tnabla_eX)\circ Z.
\]
Then \eqref{eq:dA-expanded} becomes
\[
        T(Y,X,Z)+((\tnabla_eX)\circ Y)\circ Z
        =
        T(Z,X,Y)+((\tnabla_eX)\circ Z)\circ Y,
\]
from which, using commutativity and associativity of $\circ$,  we get
\[
        T(Y,X,Z)=T(Z,X,Y).
\]
This proves the lemma.
\end{proof}

\begin{proposition}
\label{prop:oneleg}
The connection $\tnabla$ satisfies
\[
        T(X,Y,Z)=T(Y,X,Z)
        \qquad\forall\,Y,Z.
\]
Equivalently,
\[
        (\tnabla_X\circ)(Y,Z)=(\tnabla_Y\circ)(X,Z).
\]
\end{proposition}

\begin{proof}
Set
\[
        H(Y,Z):=T(X,Y,Z)-T(Y,X,Z).
\]
We prove that $H=0$.

First, $H$ is symmetric. Indeed, using the symmetry of $T$ in its last two arguments
and Lemma~\ref{lem:nablaX}, we get
\[H(Z,Y)=H(Y,Z),\]
and moreover,
\begin{equation}\label{eq: H(e,Z)=0}
        H(e,Z)=0.
\end{equation}
by Lemma~\ref{lem:Te}.

We now prove that $A$ is self-adjoint with respect to $H$:
\[
        H(AY,Z)=H(Y,AZ).
\]
Since $A=C_X$, we have
\[
        AY=X\circ Y,
        \qquad
        AZ=X\circ Z.
\]
Therefore, we get
\begin{align}
        H(AY,Z)-H(Y,AZ)
        &=
        \bigl[
        T(X,X\circ Y,Z)-T(X,Y,X\circ Z)
        \bigr]\nonumber\\
        &\quad
        -
        \bigl[
        T(X\circ Y,X,Z)-T(Y,X,X\circ Z)
        \bigr].
\label{eq:H-difference-two-brackets}
\end{align}
We shall prove that the two square brackets in
\eqref{eq:H-difference-two-brackets} are equal.

Let
\[
        B_1
        :=
        T(X,X\circ Y,Z)-T(X,Y,X\circ Z).
\]
We compute $B_1$ using the covariant associativity identity
\eqref{eq:covassoc},
\[
        T(U,V\circ W,Z)+T(U,V,W)\circ Z
        =
        T(U,V,W\circ Z)+V\circ T(U,W,Z).
\]
Apply it with
\[
        U=X,
        \qquad
        V=Y,
        \qquad
        W=X.
\]
Since the product is commutative equation \eqref{eq:covassoc} gives
\[
        T(X,X\circ Y,Z)+T(X,Y,X)\circ Z
        =
        T(X,Y,X\circ Z)+Y\circ T(X,X,Z).
\]
Rearranging, we obtain
\[
        T(X,X\circ Y,Z)-T(X,Y,X\circ Z)
        =
        Y\circ T(X,X,Z)-T(X,Y,X)\circ Z.
\]
Since $T$ is symmetric in its last two arguments (see \eqref{eq:Tcomm-two last entries}) then
\begin{equation}
        B_1
        =
        Y\circ T(X,X,Z)-Z\circ T(X,X,Y).
\label{eq:B1-expression}
\end{equation}

Now let
\[
        B_2
        :=
        T(X\circ Y,X,Z)-T(Y,X,X\circ Z).
\]
We compute $B_2$ using the covariant Hertling--Manin identity
\eqref{eq:covHM},
\begin{align*}
0={}&T(R\circ S,P,Q)-T(P\circ Q,R,S)\\
&+T(P,R,S)\circ Q+T(Q,R,S)\circ P\\
&-T(S,P,Q)\circ R-T(R,P,Q)\circ S.
\end{align*}
Apply this identity with
\[
        P=X,
        \qquad
        Q=Y,
        \qquad
        R=X,
        \qquad
        S=Z.
\]
This gives
\begin{align}
0={}&
T(X\circ Z,X,Y)-T(X\circ Y,X,Z)\nonumber\\
&+T(X,X,Z)\circ Y+T(Y,X,Z)\circ X\nonumber\\
&-T(Z,X,Y)\circ X-T(X,X,Y)\circ Z.
\label{eq:HM-XXYZ}
\end{align}
By Lemma~\ref{lem:nablaX}, we have $ T(Y,X,Z)=T(Z,X,Y),$
hence the two middle terms $T(Y,X,Z)\circ X,  -T(Z,X,Y)\circ X$
cancel in \eqref{eq:HM-XXYZ}. Thus
\begin{equation}
0=
T(X\circ Z,X,Y)-T(X\circ Y,X,Z)
+Y\circ T(X,X,Z)-Z\circ T(X,X,Y).
\label{eq:HM-after-middle-cancel}
\end{equation}
Again by Lemma~\ref{lem:nablaX}, now applied with the pair of arguments
$X\circ Z$ and $Y$, we have
\[
        T(X\circ Z,X,Y)=T(Y,X,X\circ Z).
\]
Substituting this into \eqref{eq:HM-after-middle-cancel}, we obtain
\[
0=
T(Y,X,X\circ Z)-T(X\circ Y,X,Z)
+Y\circ T(X,X,Z)-Z\circ T(X,X,Y).
\]
Rearranging gives
\begin{equation}
        T(X\circ Y,X,Z)-T(Y,X,X\circ Z)
        =
        Y\circ T(X,X,Z)-Z\circ T(X,X,Y).
\label{eq:B2-expression}
\end{equation}
In other words,
\[
        B_2
        =
        Y\circ T(X,X,Z)-Z\circ T(X,X,Y).
\]

Comparing \eqref{eq:B1-expression} and \eqref{eq:B2-expression}, we get
\[
        B_1=B_2.
\]
Therefore \eqref{eq:H-difference-two-brackets} gives
\begin{equation}
        H(AY,Z)=H(Y,AZ).
\label{eq:H-A-invariance}
\end{equation}

Finally,  we use cyclicity. Since
\[
        e,Ae,\ldots,A^{n-1}e
\]
is a local frame, it is enough to evaluate $H$ on pairs of vectors of this frame.
Repeatedly using \eqref{eq:H-A-invariance}, we obtain
\[
        H(A^ie,A^je)=H(e,A^{i+j}e).
\]
But $H(e,\cdot)=0$ by \eqref{eq: H(e,Z)=0}, hence
\[
        H(A^ie,A^je)=0
        \qquad
        \forall\,i,j.
\]
Therefore $H=0$. 
\end{proof}

For a vector field $U$, denote by $\mathcal S(U)$ the property
\[
        T(U,Y,Z)=T(Y,U,Z)
        \qquad\forall\,Y,Z.
\]

\begin{lemma}
\label{lem:propagate}
Let $U_0$ be a vector field. If $\mathcal S(X)$ and $\mathcal S(U_0)$ hold,
then $\mathcal S(X\circ U_0)$ holds.
\end{lemma}

\begin{proof}
We assume
\[
        \mathcal S(U_0):\qquad
        T(U_0,Y,Z)=T(Y,U_0,Z)
        \qquad\forall\,Y,Z.
\]
We also know
\[
        \mathcal S(X):\qquad
        T(X,Y,Z)=T(Y,X,Z)
        \qquad\forall\,Y,Z
\]
by Proposition~\ref{prop:oneleg}. We want to prove
\[
        T(X\circ U_0,Y,Z)=T(Y,X\circ U_0,Z)
        \qquad\forall\,Y,Z.
\]

We first compute $T(X\circ U_0,Y,Z)$. Start from the covariant
Hertling--Manin identity
\begin{align*}
0={}&T(R\circ S,P,Q)-T(P\circ Q,R,S)\\
&+T(P,R,S)\circ Q+T(Q,R,S)\circ P\\
&-T(S,P,Q)\circ R-T(R,P,Q)\circ S.
\end{align*}
Apply it with
\[
        P=X,
        \qquad
        Q=U_0,
        \qquad
        R=Y,
        \qquad
        S=Z.
\]
This gives
\begin{align}
0={}&T(Y\circ Z,X,U_0)-T(X\circ U_0,Y,Z)\nonumber\\
&+T(X,Y,Z)\circ U_0+T(U_0,Y,Z)\circ X\nonumber\\
&-T(Z,X,U_0)\circ Y-T(Y,X,U_0)\circ Z.
\label{eq:prop-HM-start}
\end{align}

We now rewrite the terms in \eqref{eq:prop-HM-start} using $\mathcal S(X)$.
Since $\mathcal S(X)$ says that
\[
        T(X,A,B)=T(A,X,B)
        \qquad\forall\,A,B,
\]
we have
\[
        T(Y\circ Z,X,U_0)=T(X,Y\circ Z,U_0),
\]
and also
\[
        T(Z,X,U_0)=T(X,Z,U_0),
        \qquad
        T(Y,X,U_0)=T(X,Y,U_0).
\]
Hence \eqref{eq:prop-HM-start} becomes
\begin{align*}
0={}&T(X,Y\circ Z,U_0)-T(X\circ U_0,Y,Z)\\
&+T(X,Y,Z)\circ U_0+X\circ T(U_0,Y,Z)\\
&-Y\circ T(X,Z,U_0)-Z\circ T(X,Y,U_0).
\end{align*}
Solving for $T(X\circ U_0,Y,Z)$ gives
\begin{align}
T(X\circ U_0,Y,Z)
={}&T(X,Y\circ Z,U_0)+T(X,Y,Z)\circ U_0
+X\circ T(U_0,Y,Z)\nonumber\\
&-Y\circ T(X,Z,U_0)-Z\circ T(X,Y,U_0).
\label{eq:prop-left-first}
\end{align}

We now simplify the first two terms in \eqref{eq:prop-left-first} by the covariant
associativity identity \eqref{eq:covassoc}:
\begin{equation}\label{eq: covassaux}
        T(A,B\circ C,D)+T(A,B,C)\circ D
        =
        T(A,B,C\circ D)+B\circ T(A,C,D).
\end{equation}
Apply this identity with
\[
        A=X,
        \qquad
        B=Y,
        \qquad
        C=Z,
        \qquad
        D=U_0.
\]
We obtain
\begin{equation}
        T(X,Y\circ Z,U_0)+T(X,Y,Z)\circ U_0
        =
        T(X,Y,Z\circ U_0)+Y\circ T(X,Z,U_0).
\label{eq:prop-assoc-left}
\end{equation}
Substituting \eqref{eq:prop-assoc-left} into \eqref{eq:prop-left-first}, the two terms
$Y\circ T(X,Z,U_0)$ and $-Y\circ T(X,Z,U_0)$ cancel. Therefore
\begin{equation}
        T(X\circ U_0,Y,Z)
        =
        T(X,Y,Z\circ U_0)
        +X\circ T(U_0,Y,Z)
        -Z\circ T(X,Y,U_0).
\label{eq:left-prop-expanded}
\end{equation}

We now compute $T(Y,X\circ U_0,Z)$. Apply the covariant associativity identity \eqref{eq: covassaux} with
\[
        A=Y,
        \qquad
        B=X,
        \qquad
        C=U_0,
        \qquad
        D=Z.
\]
This gives
\begin{equation}
        T(Y,X\circ U_0,Z)+T(Y,X,U_0)\circ Z
        =
        T(Y,X,U_0\circ Z)+X\circ T(Y,U_0,Z).
\label{eq:prop-right-assoc-start}
\end{equation}

We now use both symmetry assumptions. From $\mathcal S(X)$,
\[
        T(Y,X,U_0)=T(X,Y,U_0),
        \qquad
        T(Y,X,U_0\circ Z)=T(X,Y,U_0\circ Z).
\]
From $\mathcal S(U_0)$,
\[
        T(Y,U_0,Z)=T(U_0,Y,Z).
\]
Therefore \eqref{eq:prop-right-assoc-start} becomes
\[
        T(Y,X\circ U_0,Z)+T(X,Y,U_0)\circ Z
        =
        T(X,Y,U_0\circ Z)+X\circ T(U_0,Y,Z).
\]
Solving for $T(Y,X\circ U_0,Z)$ gives
\begin{equation}
        T(Y,X\circ U_0,Z)
        =
        T(X,Y,U_0\circ Z)
        +X\circ T(U_0,Y,Z)
        -Z\circ T(X,Y,U_0).
\label{eq:right-prop-expanded}
\end{equation}

Finally, by commutativity of the product $\circ$ we get that the right-hand sides of \eqref{eq:left-prop-expanded} and \eqref{eq:right-prop-expanded} are equal,  and therefore
\[
        T(X\circ U_0,Y,Z)=T(Y,X\circ U_0,Z).
\]
This is exactly $\mathcal S(X\circ U_0)$.
\end{proof}

\begin{proposition}
\label{prop:fullsym}
The connection $\tnabla$ satisfies
\begin{equation}\label{eq: full symmetry}
        (\tnabla_W\circ)(Y,Z)=(\tnabla_Y\circ)(W,Z)
        \qquad\forall\,W,Y,Z.
\end{equation}
\end{proposition}

\begin{proof}
We have that $\mathcal S(e)$ holds  by Lemma~\ref{lem:Te}, and $\mathcal S(X)$ holds by
Proposition~\ref{prop:oneleg}. By Lemma~\ref{lem:propagate},
\[
        \mathcal S(X^{\circ k}) \text{ holds }
        \qquad k=0,1,\ldots,n-1.
\]
Since $e,X,\ldots,X^{\circ(n-1)}$ is a frame, 
$\mathcal S(W)$ holds for every vector field $W$,  since $T$ is a tensor,  thus $\mathcal{O}_M$-linear in each of its entries. 
\end{proof}

\vspace{.2 cm}
This concludes the proof of Theorem \ref{thm: existence and uniqueness of cyclic connections}. 

\subsection{Characterization of the symmetries}

 The following is a vast generalization of  Lemma 3.5 in \cite{LPVG1},  providing also a much simpler proof of the latter:
\begin{theorem}[Equation for the symmetries]
\label{thm:coordinate-free-prop34-lem35}
Let $(M,\circ,e)$ be an F-manifold, let $X$ be cyclic, and set
\[
        A:=C_X=X\circ.
\]
Let $\tnabla$ be the unique torsionless connection which is distinguished with respect to  $X$ constructed in Theorem \ref{thm: existence and uniqueness of cyclic connections}. Then the two hydrodynamic systems
\[
        u_t=A(u)u_x,
        \qquad
        u_\tau=B(u)u_x
\]
have commuting flows if and only if $B=C_Y=Y\circ$  for some vector field  $Y$ and
\[
        d_{\tnabla}(Y\circ)=0.
\]
\end{theorem}

\begin{proof}
First we recall Lemma 3.3 and Proposition 3.4 from \cite{LPVG1}.
Let $A$ and $B$ be arbitrary $(1,1)$-tensor fields on a manifold $M$, and let $\nabla$
be a torsionless connection. Consider the two hydrodynamic systems
\[
        u_t=A(u)u_x,
        \qquad
        u_\tau=B(u)u_x.
\]
Then the two hydrodynamic systems have commuting flows if and only if
\begin{equation}
        AB=BA, \text{ (commutation as endomorphisms)}
\label{eq:AB-BA}
\end{equation}
and 
\begin{equation}
        (d_\nabla A)(U,BV)+(d_\nabla A)(V,BU)
        =
        (d_\nabla B)(U,AV)+(d_\nabla B)(V,AU) \quad \forall U,V
\label{eq:coordinate-free-prop34}
\end{equation}
hold.  This is just Lemma 3.3 and Proposition 3.4 of \cite{LPVG1} that were derived in full generality there. 

We now specialize to the theorem. We take
\[
        \nabla=\tnabla,
        \qquad
        A=C_X=X\circ.
\]
Since $X$ is cyclic, at each point of the cyclic locus $C_X$ is a cyclic endomorphism, and its commutant in $\End(TM)$ coincides with the polynomial algebra in $C_X$ over $\O$. So there exist smooth functions $a_k$, uniquely determined on the cyclic locus, with
\[
B=\sum_{k=0}^{n-1} a_k\,C_X^k.
\]
Using $C_X^k=C_{X^{\circ k}}$ (by associativity) and the linearity of $Z\mapsto C_Z$, set $Y:=\sum_{k=0}^{n-1} a_kX^{\circ k}=Be\in\Gamma(\T)$; then $B=C_{Y}=Y\circ$.

We need to impose the remaining condition, that is
\eqref{eq:coordinate-free-prop34}. Since $\tnabla$ was constructed so that
\[
        d_{\tnabla}A=0,
\]
the left-hand side of \eqref{eq:coordinate-free-prop34} vanishes. Hence
\begin{equation}
        (d_{\tnabla}B)(U,AV)+(d_{\tnabla}B)(V,AU)=0
        \qquad
        \forall\,U,V.
\label{eq:cyclic-reduced-equation}
\end{equation}

We now prove that \eqref{eq:cyclic-reduced-equation} implies
$d_{\tnabla}B=0$. This is the generalization of Lemma $3.5$ in \cite{LPVG1}.

Fix a point $p\in M$ and set
\[
        F_p:=T_pM.
\]
Since $X$ is cyclic,  $ e_p,A_pe_p,A_p^2e_p,\ldots,A_p^{n-1}e_p$
is a basis of $F_p$.  Define $H$ as 
\[
        H:=(d_{\tnabla}B)_p\in\Lambda^2F_p^*\otimes F_p.
\]
Equation \eqref{eq:cyclic-reduced-equation} says
\begin{equation}
        H(u,A_pv)+H(v,A_pu)=0
        \qquad
        \forall\,u,v\in F_p.
\label{eq:H-cyclic-reduced}
\end{equation}
Let $\lambda\in F_p^*$ be arbitrary and define the scalar skew-form $ h_\lambda(u,v):=\lambda(H(u,v)).$
Pairing \eqref{eq:H-cyclic-reduced} with $\lambda$ gives
\begin{equation}
        h_\lambda(u,A_pv)+h_\lambda(v,A_pu)=0.
\label{eq:h-cyclic-reduced}
\end{equation}
Since $h_\lambda$ is skew-symmetric,  $h_\lambda(v,A_pu)=-h_\lambda(A_pu,v).$ Therefore \eqref{eq:h-cyclic-reduced} is equivalent to
\begin{equation}
        h_\lambda(A_pu,v)=h_\lambda(u,A_pv)
        \qquad
        \forall\,u,v\in F_p.
\label{eq:h-self-adjoint}
\end{equation}
Thus $A_p$ is self-adjoint with respect to the skew-form $h_\lambda$. Using \eqref{eq:h-self-adjoint} repeatedly, for all $i,j\geq0$ we obtain
\[
        h_\lambda(A_p^ie_p,A_p^je_p)
        =
        h_\lambda(e_p,A_p^{i+j}e_p).
\]
On the other hand, applying \eqref{eq:h-self-adjoint} once more gives
\[
        h_\lambda(A_p^{i+j}e_p,e_p)
        =
        h_\lambda(e_p,A_p^{i+j}e_p).
\]
Since $h_\lambda$ is skew-symmetric,
\[
        h_\lambda(A_p^{i+j}e_p,e_p)
        =
        -h_\lambda(e_p,A_p^{i+j}e_p).
\]
Hence we get $h_\lambda(e_p,A_p^{i+j}e_p)=0.$
Therefore
\[
        h_\lambda(A_p^ie_p,A_p^je_p)=0
        \qquad
        \forall\,i,j\]
        and since  $e_p,A_pe_p,\ldots,A_p^{n-1}e_p$
is a basis of $F_p$, we get
\[
        h_\lambda=0.
\]
Since $\lambda\in F_p^*$ was arbitrary, this implies
     $   H=0.$ Thus $
        (d_{\tnabla}B)_p=0. $
As $p$ is arbitrary, $ d_{\tnabla}B=0,$
 and since $B=Y\circ$, this is exactly
\[
        d_{\tnabla}(Y\circ)=0.
\]
The converse is clear. 
\end{proof}

\section{Cauchy-Kowalevski-Cartan-K\"ahler study of the symmetries}\label{SectionCKsymm}

\subsection{Non-characteristic Cauchy theorem for analytic and holomorphic symmetries}
In this subsection $X$ is a cyclic vector field and $\tnabla$ is the unique distinguished torsionless connection associated with $X$ such that $\tnabla e=0$.  A vector field $Y$ will be called a symmetry when it satisfies $d_{\tnabla}(Y\circ)=0$.

The analytic arguments below use the classical Cauchy--Kowalevski theorem for first-order analytic systems solved for a non-characteristic derivative; see the original source \cite{Kowalevsky1875} and the standard references \cite{CK-CourantHilbert,CK-John,CK-Hormander}.  The compatibility calculation in the proof is the corresponding local analytic involutivity check in the spirit of the Cartan--K\"ahler theorem; for the exterior differential systems perspective see \cite{BCGGG,IveyLandsberg}.

\begin{lemma}\label{lem:symmetry-equiv-PDE}
The symmetry equation
\begin{equation}\label{eq:symmetry-basic}
        d_{\tnabla}C_Y=0
\end{equation}
is equivalent to
\begin{equation}\label{eq:nablaY-reduced-main}
        \tnabla_UY=U\circ\tnabla_eY
        \qquad\forall U\in\Gamma(TM).
\end{equation}
\end{lemma}
\proof
The first implication is just the first statement of Lemma \ref{lem:nablaX} with $Y$ in place of $X$ and $U$ in place of $Y$.  Conversely, $d_{\tnabla}(Y\circ)=0$ is equivalent to
\[
        (\tnabla_UY\circ)(V)=(\tnabla_VY\circ)(U).
\]
Expanding by the Leibniz rule and using the full symmetry of $\tnabla$ in \eqref{eq: full symmetry}, this reduces to
\[
        (\tnabla_UY)\circ V=(\tnabla_VY)\circ U.
\]
The identity \eqref{eq:nablaY-reduced-main} implies this equality immediately.  \endproof

A tangent vector $T\in T_pM$ will be called \emph{product-invertible} if multiplication by $T$ is an invertible endomorphism of $T_pM$, equivalently if there exists $T^{-1}\in T_pM$ such that
\[
        T^{-1}\circ T=e_p .
\]
This condition is open.  Moreover, on the open set where it holds, the inverse $T^{-1}$ depends analytically on $T$ and on the product, because it is obtained by inverting the matrix of the multiplication operator $C_T$.  An analytic embedded curve whose tangent is product-invertible will be called product-non-characteristic.

\begin{theorem}\label{thm:main-restated}\label{thm:noncharacteristic-CK-symmetries}
Let $(M,\circ,e)$ be a real analytic\protect\footnote{\label{fn:holomorphic-case}The same statement holds verbatim in the holomorphic category, with the obvious substitutions: open intervals $(-\varepsilon,\varepsilon)\subset\mathbb R$ are replaced by open discs $\Delta_\varepsilon\subset\mathbb C$, and ``real analytic'' is replaced by ``holomorphic'' throughout. Product-invertibility is understood over $\mathbb C$. The proof below carries over without modification: one uses holomorphic coordinates adapted to the embedded curve and the holomorphic Cauchy--Kowalevski theorem with uniqueness; for the latter see \cite[Thm.~9.4.5]{CK-Hormander}.}  F-manifold of dimension $n$, let $X$ be analytic and cyclic on the open set $U_{\mathrm{cyc}}\subset M$, and let $\tnabla$ be the unique flat-unit torsionless connection of Theorem~\ref{thm: existence and uniqueness of cyclic connections}, characterised on $U_{\mathrm{cyc}}$ by $\tnabla e=0$ and $d_{\tnabla}(X\circ)=0$.  Assume the curvature identity
\begin{equation}\label{eq:3RC-Zinside}
        R^{\tnabla}(U,V)(W\circ Z)+R^{\tnabla}(V,W)(U\circ Z)+R^{\tnabla}(W,U)(V\circ Z)=0,
\end{equation}
or equivalently\footnote{The two conditions \eqref{eq:3RC-Zinside} and \eqref{eq:3RC-outside} are equivalent: as $\tnabla$ is torsion-free and satisfies the symmetry $(\tnabla_X\circ)(Y,Z)=(\tnabla_Y\circ)(X,Z)$ of Proposition~\ref{prop:fullsym}, applying the first Bianchi identity to one of them, with the cyclicity sum carried out in the appropriate slots, produces the other (see also \cite{LPR}).}
\begin{equation}\label{eq:3RC-outside}
        W\circ R^{\tnabla}(U,V)Z+U\circ R^{\tnabla}(V,W)Z+V\circ R^{\tnabla}(W,U)Z=0.
\end{equation}
Let
\[
        C:(-\varepsilon,\varepsilon)\longrightarrow U_{\mathrm{cyc}}
\]
be an analytic embedded product-non-characteristic curve, i.e.
\[
        T(s):=\dot C(s)
\]
is product-invertible for all $s$, after possibly shrinking the interval.  Let $Y_C$ be an analytic vector field along $C$.  Then, after shrinking the interval and a neighbourhood of $C$, there exists a unique analytic vector field $Y$ such that
\begin{equation}\label{eq:noncharacteristic-main-conclusion}
        Y|_C=Y_C,
        \qquad
        d_{\tnabla}(Y\circ)=0.
\end{equation}
Furthermore, the local analytic solution space of the equation $d_{\tnabla}(Y\circ)=0$ near any point of $U_{\mathrm{cyc}}$ is parametrised by $n$ arbitrary analytic functions of one variable, namely the components of the Cauchy datum $Y_C$ in an analytic frame of $TM|_C$.
\end{theorem}

\begin{proof}
\noindent\emph{Part 1: choice of non-characteristic coordinates.}
The purpose of this first step is only to put the Cauchy curve in normal position and to identify the derivative with respect to which the Cauchy--Kowalevski theorem will be applied.  Choose analytic coordinates $(s,z^2,\ldots,z^n)$ in a neighbourhood of the curve such that
\[
        C(s)=(s,0,\ldots,0),
        \qquad
        \partial_s|_C=\dot C(s).
\]
After shrinking the neighbourhood we may assume that
\[
        T:=\partial_s
\]
is product-invertible everywhere.  Write
\[
        Z_\alpha:=\partial_{z^\alpha},\qquad
        P_\alpha:=Z_\alpha\circ T^{-1},
        \qquad \alpha=2,\ldots,n.
\]
Thus
\[
        Z_\alpha=P_\alpha\circ T.
\]
The vector field $T$ is the distinguished non-characteristic direction.  The vector fields $Z_\alpha$ are the transverse coordinate directions, and the product-invertibility of $T$ is exactly what allows us to solve the symmetry equation for the transverse derivatives in terms of the $T$-derivative.

\smallskip
\noindent\emph{Part 2: reduction to the non-characteristic system.}
By Lemma~\ref{lem:symmetry-equiv-PDE}, the symmetry equation is equivalent to
\begin{equation}\label{eq:nonchar-full-symmetry-equation}
        \tnabla_UY=U\circ\tnabla_eY
        \qquad\text{for every vector field }U.
\end{equation}
We first solve the following non-characteristic system:
\begin{equation}\label{eq:nonchar-Dalpha-system}
        \mathcal D_\alpha Y:=\tnabla_{Z_\alpha}Y
        -P_\alpha\circ\tnabla_TY=0,
        \qquad \alpha=2,\ldots,n.
\end{equation}
This is exactly the system obtained from \eqref{eq:nonchar-full-symmetry-equation} by solving the equation in the $T$-direction for $\tnabla_eY.$ Now we prove that \eqref{eq:nonchar-full-symmetry-equation} is equivalent to \eqref{eq:nonchar-Dalpha-system},  in the sense that every solution of the former is a solution of the latter and vice versa.   Indeed, if \eqref{eq:nonchar-full-symmetry-equation} holds, then in particular
\[
        \tnabla_TY=T\circ\tnabla_eY,
        \qquad
        \tnabla_eY=T^{-1}\circ\tnabla_TY,
\]
and substituting this back into \eqref{eq:nonchar-full-symmetry-equation} gives
\[
        \tnabla_UY=U\circ T^{-1}\circ\tnabla_TY.
\]
Taking $U=Z_\alpha, \, \alpha=2,\dots n$ gives \eqref{eq:nonchar-Dalpha-system}.  Conversely, if \eqref{eq:nonchar-Dalpha-system} holds and
\[
        V:=T^{-1}\circ\tnabla_TY,
\]
then $\tnabla_TY=T\circ V$ and $\tnabla_{Z_\alpha}Y=Z_\alpha\circ V$ for every $\alpha$.  Since $\tnabla_{\bullet}Y$ is $\mathcal{O}_M$-linear and since $T,Z_2,\ldots,Z_n$ form a local frame,   $\tnabla_UY=U\circ V$ for every $U$.  Taking $U=e$ gives $V=\tnabla_eY$, hence \eqref{eq:nonchar-full-symmetry-equation} follows.

\smallskip
\noindent\emph{Part 3: Cauchy--Kowalevski normal form.}
This step records explicitly that each equation in \eqref{eq:nonchar-Dalpha-system} is a first-order analytic system solved for one transverse derivative.  If $M_\alpha=C_{P_\alpha}$, i.e.
\[
        M_\alpha{}^k{}_r
\]
are the components of multiplication by $P_\alpha=Z_\alpha\circ T^{-1}$, then \eqref{eq:nonchar-Dalpha-system} reads
\begin{equation}\label{eq:nonchar-CK-normal-form}
        \partial_{z^\alpha}Y^k
        =M_\alpha{}^k{}_r\bigl(\partial_sY^r+\Gamma^r_{1q}Y^q\bigr)
         -\Gamma^k_{\alpha r}Y^r,
        \qquad \alpha=2,\ldots,n.
\end{equation}
Here the index $1$ denotes the coordinate direction $s$.  The coefficients are analytic because $T^{-1}$ depends analytically on $T$ and on the product.  Thus, when all variables except $z^\alpha$ are regarded as tangential variables, \eqref{eq:nonchar-CK-normal-form} is precisely the standard Cauchy--Kowalevski normal form for a first-order analytic system, as in \cite{CK-CourantHilbert,CK-John,CK-Hormander}.

\smallskip
\noindent\emph{Part 4: formal compatibility of the non-characteristic operators.}
The purpose of this step is to prove the involutivity condition needed in the iterative Cauchy--Kowalevski construction.  Let
\[
        \mathcal T(A,B,W):=(\tnabla_A\circ)(B,W)
\]
(this is the tensor $T$ defined in Section~3, but here we use $\mathcal T$ since $T$ is reserved for the derivative along the Cauchy curve).  By commutativity of $\circ$ and Proposition~\ref{prop:fullsym}, $\mathcal T$ is totally symmetric.  Moreover, Lemma~\ref{lem:covassoc} implies the following identity: for all vector fields $P,Q,T,W$,
\begin{equation}\label{eq:nonchar-product-T-identity}
\begin{split}
        &\mathcal T(P\circ T,Q,W)-P\circ\mathcal T(T,Q,W)\\
        &\hspace{3cm}=\mathcal T(Q\circ T,P,W)-Q\circ\mathcal T(T,P,W).
\end{split}
\end{equation}
Indeed, apply \eqref{eq:covassoc} with $(U,V,W,Z)=(W,P,T,Q)$ and then use the total symmetry of $\mathcal T$.

We claim that the operators $\mathcal D_\alpha$ defined in \eqref{eq:nonchar-Dalpha-system} commute:
\begin{equation}\label{eq:nonchar-D-commute}
        [\mathcal D_\alpha,\mathcal D_\beta]Y=0
        \qquad\text{for every vector field }Y,
        \quad 2\le \alpha,\beta\le n.
\end{equation}
To prove this, set temporarily
\begin{equation}\label{eq:position}
        U:=Z_\alpha=P\circ T,
        \qquad
        V:=Z_\beta=Q\circ T,
        \qquad
        P:=P_\alpha,
        \qquad
        Q:=P_\beta,
\end{equation}
where $T$ is the tangent vector field to the curve $C(s)$ appearing already in Part 1,  and the rest is consistent with the notation introduced in Part 1.
The coordinate vector fields commute, so $[U,V]=[U,T]=[V,T]=0$.  A direct expansion of
\[
        \mathcal D_U:=\tnabla_U-C_P\tnabla_T,
        \qquad
        \mathcal D_V:=\tnabla_V-C_Q\tnabla_T
\]
gives
\begin{equation}\label{eq:nonchar-commutator-expansion}
\begin{split}
[\mathcal D_U,\mathcal D_V]Y
={}&\Bigl(R^{\tnabla}(U,V)+C_P R^{\tnabla}(V,T)-C_Q R^{\tnabla}(U,T)\Bigr)Y\\
&+\Bigl(-\tnabla_UC_Q+\tnabla_VC_P+C_P\tnabla_TC_Q-C_Q\tnabla_TC_P\Bigr)(\tnabla_TY).
\end{split}
\end{equation}
Here $\tnabla_AC_B$ denotes the covariant derivative of the endomorphism $C_B$.  We show that the two parentheses in \eqref{eq:nonchar-commutator-expansion} vanish.

First consider the curvature parenthesis.  Multiply it by $C_T$.  Since $C_T$ is invertible and $C_TC_P=C_U$, $C_TC_Q=C_V$, we have, for every $Y$,
\[
\begin{split}
&T\circ\Bigl(R^{\tnabla}(U,V)Y+P\circ R^{\tnabla}(V,T)Y-Q\circ R^{\tnabla}(U,T)Y\Bigr)\\
&\qquad =T\circ R^{\tnabla}(U,V)Y
        +U\circ R^{\tnabla}(V,T)Y
        +V\circ R^{\tnabla}(T,U)Y,
\end{split}
\]
which is zero by the curvature identity \eqref{eq:3RC-outside}, applied with $W=T$ and $Z=Y$.  Hence the curvature parenthesis in \eqref{eq:nonchar-commutator-expansion} vanishes.

We now prove that the coefficient of $\tnabla_TY$ in \eqref{eq:nonchar-commutator-expansion} also vanishes.  First observe that
\begin{equation}\label{eq:nonchar-nabla-C-formula}
        (\tnabla_AC_B)(W)=(\tnabla_AB)\circ W+\mathcal T(A,B,W).
\end{equation}
We also need the following consequence of the commutativity of the coordinate frame.  Since $[U,V]=0$ and the connection is torsionless,
\begin{equation}\label{eq:position2}
        0=\tnabla_UV-\tnabla_VU\stackrel{\eqref{eq:position}}{=}\tnabla_U(Q\circ T)-\tnabla_V(P\circ T).
\end{equation}
 By the Leibniz rule for the covariant derivative of the product,
\[
\begin{split}
\tnabla_U(Q\circ T)
   &= (\tnabla_UQ)\circ T+Q\circ \tnabla_UT+\mathcal T(U,Q,T),\\
\tnabla_V(P\circ T)
   &= (\tnabla_VP)\circ T+P\circ \tnabla_VT+\mathcal T(V,P,T).
\end{split}
\]
Therefore \eqref{eq:position2} gives
\begin{equation}\label{eq:position3}
\begin{split}
0={}&(\tnabla_UQ-\tnabla_VP)\circ T
       +Q\circ \tnabla_UT-P\circ \tnabla_VT  \\
&\quad +\mathcal T(U,Q,T)-\mathcal T(V,P,T).
\end{split}
\end{equation}
Now use the remaining commutation relations.  Since $[U,T]=0$ and $[V,T]=0$, torsion-freeness gives
\[
        \tnabla_UT=\tnabla_TU,
        \qquad
        \tnabla_VT=\tnabla_TV .
\]
Using again \eqref{eq:position}, namely $U=P\circ T$ and $V=Q\circ T$, and expanding the two products, we obtain
\[
\begin{split}
\tnabla_UT
   &=\tnabla_T(P\circ T)
     =(\tnabla_TP)\circ T+P\circ\tnabla_TT+\mathcal T(T,P,T),\\
\tnabla_VT
   &=\tnabla_T(Q\circ T)
     =(\tnabla_TQ)\circ T+Q\circ\tnabla_TT+\mathcal T(T,Q,T).
\end{split}
\]
Substitution in \eqref{eq:position3} gives
\begin{align*}
0={}&(\tnabla_UQ-\tnabla_VP)\circ T
     +Q\circ\bigl((\tnabla_TP)\circ T+P\circ\tnabla_TT+
        \mathcal T(T,P,T)\bigr)\\
&\quad -P\circ\bigl((\tnabla_TQ)\circ T+Q\circ\tnabla_TT+
        \mathcal T(T,Q,T)\bigr)\\
&\quad +\mathcal T(U,Q,T)-\mathcal T(V,P,T).
\end{align*}
The two terms containing $\tnabla_TT$ cancel, because the product is commutative and associative.

Using commutativity and associativity once more to factor out the final $T$, the preceding identity becomes
\begin{equation}\label{eq:position4}
\begin{split}
0={}&\bigl(\tnabla_UQ-\tnabla_VP+Q\circ\tnabla_TP-P\circ\tnabla_TQ\bigr)\circ T\\
&+\mathcal T(U,Q,T)-\mathcal T(V,P,T)
  +Q\circ\mathcal T(T,P,T)-P\circ\mathcal T(T,Q,T).
 \end{split}
\end{equation}
The second line is zero by \eqref{eq:nonchar-product-T-identity} with $W=T$: indeed, since $U=P\circ T$ and $V=Q\circ T$, that identity reads
\[
        \mathcal T(U,Q,T)-P\circ\mathcal T(T,Q,T)
        =\mathcal T(V,P,T)-Q\circ\mathcal T(T,P,T).
\]
Since $C_T$ is invertible, we obtain from \eqref{eq:position4}:
\begin{equation}\label{eq:nonchar-PQ-derivative-identity}
        \tnabla_UQ-P\circ\tnabla_TQ
        =\tnabla_VP-Q\circ\tnabla_TP.
\end{equation}
Now evaluate the second parenthesis on the right hand side of \eqref{eq:nonchar-commutator-expansion} on an arbitrary vector field $W$.  By \eqref{eq:nonchar-nabla-C-formula}, it equals
\begin{align*}
&\bigl(-\tnabla_UQ+\tnabla_VP+P\circ\tnabla_TQ-Q\circ\tnabla_TP\bigr)\circ W\\
&\quad -\mathcal T(U,Q,W)+\mathcal T(V,P,W)
       +P\circ\mathcal T(T,Q,W)-Q\circ\mathcal T(T,P,W).
\end{align*}
The first line vanishes by \eqref{eq:nonchar-PQ-derivative-identity}, and the second line vanishes by \eqref{eq:nonchar-product-T-identity}.  Thus the second parenthesis in \eqref{eq:nonchar-commutator-expansion} is zero.  This proves \eqref{eq:nonchar-D-commute}.

\smallskip
\noindent\emph{Part 5: iterative Cauchy--Kowalevski construction and propagation of the old equations.}
We now construct the solution one transverse variable at a time.  The preceding commutator identity is precisely what guarantees that the equations already imposed remain true after the next Cauchy--Kowalevski step.  Let
\[
        M_k:=\{z^{k+1}=\cdots=z^n=0\},
        \qquad k=1,\ldots,n,
\]
so that $M_1=C((-\varepsilon,\varepsilon))$ after shrinking.  Start with $Y_1=Y_C$ on $M_1$.  Suppose that, for some $1\le k<n$, an analytic vector field $Y_k$ has been constructed on a neighbourhood of the curve in $M_k$, agrees with $Y_C$ on $M_1$, and satisfies
\[
        \mathcal D_\alpha Y_k=0,
        \qquad 2\le \alpha\le k.
\]
On $M_{k+1}$ solve the single analytic Cauchy--Kowalevski equation
\[
        \mathcal D_{k+1}Y=0
\]
with Cauchy datum $Y|_{M_k}=Y_k$.  In the coordinate form \eqref{eq:nonchar-CK-normal-form}, this is a first-order \emph{linear homogeneous} analytic system solved for $\partial_{z^{k+1}}Y$; hence the Cauchy--Kowalevski theorem gives a unique analytic solution $\widetilde Y_{k+1}$ near the curve, after shrinking.

We must show that the old equations remain true for $\widetilde Y_{k+1}$.  For $2\le\alpha\le k$, set
\[
        E_\alpha:=\mathcal D_\alpha \widetilde Y_{k+1}.
\]
The Cauchy datum gives $E_\alpha|_{M_k}=0$.  Since $\mathcal D_{k+1}\widetilde Y_{k+1}=0$ and the operators commute by \eqref{eq:nonchar-D-commute},
\[
        \mathcal D_{k+1}E_\alpha
        =\mathcal D_{k+1}\mathcal D_\alpha\widetilde Y_{k+1}
        =\mathcal D_\alpha\mathcal D_{k+1}\widetilde Y_{k+1}
        =0.
\]
In components this is the homogeneous Cauchy--Kowalevski system
\[
        \partial_{z^{k+1}}E_\alpha^a
        =M_{k+1}{}^a{}_r\bigl(\partial_sE_\alpha^r+\Gamma^r_{1q}E_\alpha^q\bigr)
         -\Gamma^a_{k+1,r}E_\alpha^r.
\]
With zero Cauchy datum on $M_k$, the uniqueness part of the Cauchy--Kowalevski theorem gives $E_\alpha=0$.  Hence $\widetilde Y_{k+1}$ satisfies all equations $\mathcal D_\alpha Y=0$ for $2\le\alpha\le k+1$; set $Y_{k+1}:=\widetilde Y_{k+1}$.

\smallskip
\noindent\emph{Part 6: conclusion, uniqueness, and parameter count.}
Iterating Part~5 yields an analytic vector field $Y$ on a neighbourhood of the curve satisfying \eqref{eq:nonchar-Dalpha-system} for every $\alpha=2,\ldots,n$ and $Y|_C=Y_C$.  By the equivalence proved in Part~2, $Y$ satisfies \eqref{eq:nonchar-full-symmetry-equation}; Lemma~\ref{lem:symmetry-equiv-PDE} then gives $d_{\tnabla}(Y\circ)=0$.

Since the system is linear homogeneous,  uniqueness follows from the same Cauchy--Kowalevski uniqueness argument applied to the difference of two solutions with the same Cauchy datum.  Finally, the datum $Y_C$ is an analytic section of $TM|_C$ and is therefore specified by $n$ arbitrary analytic functions of one variable, its components in any analytic frame along $C$.  The existence-and-uniqueness result above identifies these data with the corresponding local analytic solutions.
\end{proof}

\begin{corollary}\label{cor:product-noncharacteristic-symmetries}\label{cor:unit-curve-symmetries}
Under the geometric hypotheses of Theorem~\ref{thm:main-restated}, let
\[
        C:(-\varepsilon,\varepsilon)\longrightarrow U_{\mathrm{cyc}}
\]
be any analytic embedded product-non-characteristic curve, and let $Y_C$ be an analytic vector field along $C$.  Then, after possibly shrinking the interval and a neighbourhood of $C$, there exists a unique analytic vector field $Y$ such that
\[
        Y|_C=Y_C,
        \qquad
        d_{\tnabla}(Y\circ)=0.
\]
Consequently, the local analytic solution space of $d_{\tnabla}(Y\circ)=0$ may be parametrised by $n$ arbitrary analytic functions of one variable along any product-non-characteristic curve $C$.

Moreover, if $E_{(1)},\ldots,E_{(n)}$ is an analytic frame of $TM|_C$ and $Y_{(a)}$ denotes the corresponding analytic solution with $Y_{(a)}|_C=E_{(a)}$, then, after further shrinking the interval and the neighbourhood, $Y_{(1)},\ldots,Y_{(n)}$ are pointwise linearly independent.  The hydrodynamic flows
\[
        u_{t_a}=Y_{(a)}\circ u_x
\]
pairwise commute and commute with the original F-system $u_t=X\circ u_x$.  In particular, a curve $C$ that is an integral curve of the unit $e$ always satisfies the hypothesis of Theorem~\ref{thm:main-restated}.
\end{corollary}

\begin{proof}
The first part follows directly from Theorem~\ref{thm:noncharacteristic-CK-symmetries}.

Let $E_{(1)},\ldots,E_{(n)}$ be an analytic frame of $TM|_C$, and let $Y_{(1)},\ldots,Y_{(n)}$ be the corresponding analytic extensions given by Theorem~\ref{thm:noncharacteristic-CK-symmetries}.  Since $Y_{(a)}|_C=E_{(a)}$, the determinant of the $n$ vector fields $Y_{(1)},\ldots,Y_{(n)}$ in any analytic local frame is non-zero along $C$.  After shrinking the interval and the neighbourhood of $C$, this determinant remains non-zero, so the extended fields are pointwise linearly independent.

It remains to check commutativity of the corresponding hydrodynamic flows.  We use the general commutation criterion for arbitrary $(1,1)$-tensors recalled in the proof of Theorem~\ref{thm:coordinate-free-prop34-lem35}, namely \eqref{eq:AB-BA}--\eqref{eq:coordinate-free-prop34}.  Put
\[
        Y_{(0)}:=X,
        \qquad
        A_{(a)}:=C_{Y_{(a)}}=Y_{(a)}\circ,
        \qquad a=0,1,\ldots,n.
\]
For $a=0$ we have $d_{\tnabla}A_{(0)}=d_{\tnabla}(X\circ)=0$ by the defining property of $\tnabla$, while for $a\ge1$ we have $d_{\tnabla}A_{(a)}=d_{\tnabla}(Y_{(a)}\circ)=0$ by construction.  For any $a,b\in\{0,1,\ldots,n\}$,
\[
        A_{(a)}A_{(b)}(U)
        =Y_{(a)}\circ\bigl(Y_{(b)}\circ U\bigr)
        =Y_{(b)}\circ\bigl(Y_{(a)}\circ U\bigr)
        =A_{(b)}A_{(a)}(U)
\]
for every vector field $U$, by commutativity and associativity of $\circ$.  Thus \eqref{eq:AB-BA} holds.  Since both exterior covariant derivatives vanish, \eqref{eq:coordinate-free-prop34} holds identically for the pair $A_{(a)},A_{(b)}$.  The general criterion therefore implies that the flows $u_{t_a}=Y_{(a)}\circ u_x$, $a=1,\ldots,n$, pairwise commute; taking one index equal to $0$ gives that each of them commutes with the original F-system $u_t=X\circ u_x$.
\end{proof}

\begin{remark}\label{coframe}
The linear system for densities of conservation laws 
\begin{equation}\label{dcl}
\tnabla_j\partial_ih=c^s_{ij}\tnabla_e\partial_sh,
	 \end{equation} 
can be studied
 in this much more general setting, using Cauchy--Kowalevski theorem combined with formal compatibility conditions, obtaining similar results. In particular, in the analytic setting, it is possible to prove the existence of a coframe field defined by the differentials of a set of functionally independent densities of conservation laws.
 \end{remark}

\subsection{Comparison with Darboux-Tsarev regular case}\label{SubsectDarbouxTsarevComparison} In the non-semisimple setting the situation becomes much more involved since the system for the symmetries of an F-system in general does not have the form studied by Darboux even in the regular case (see \cite{LPVG2}). 
According to the generalised hodograph method, any symmetry defines a solution in implicit form. In Tsarev's framework also the converse statement is true. Indeed, given the initial data, Tsarev proved that, under some transversality conditions, it is possible to determine the symmetry defining the solution of the Cauchy problem (see Proposition 1 and Theorem 10 in \cite{Tsarev}). In the non-semisimple case, in order to fully extend Tsarev's theory one needs some further assumptions leading, in the regular case, to the notion of \emph{Darboux-Tsarev} F-system (see \cite{LPVG2,LPVG3}). It is easy to see that as soon as the Darboux-Tsarev conditions are not satisfied, higher-order linear PDEs (like the heat equation or similar equations \cite{KK,KO,XF,LPVG3}) come into play and the application of the generalised hodograph method becomes more involved. For instance, it is possible to write the general solution of the heat equation in terms of two arbitrary functions of a single variable, but the solution formula for analytic initial data is a series containing the derivatives of these functions of arbitrary order.

\subsection{The hodograph method in the analytic setting}\label{SubsectIntrinsicHodograph}\footnote{All results of this subsection admit holomorphic counterparts. In the holomorphic category, replace real analytic manifolds, vector fields, curves, functions, and solutions by holomorphic ones; replace intervals by discs; and regard $x,t$ as complex variables. Product-invertibility is understood over $\mathbb{C}.$ The proofs are unchanged: Theorem~\ref{thm:intrinsic-hodograph-safe} uses the holomorphic implicit function theorem, Theorem~\ref{thm:main-restated} uses holomorphic Cauchy--Kowalevski with uniqueness in a product-non-characteristic normal form, and Theorem~\ref{thm:intrinsic-hodograph-completeness} combines these two ingredients with holomorphic Cauchy--Kowalevski uniqueness for the evolutionary F-system.} We now present the hodograph method in this generalized context which does not depend on semisimple coordinates, Darboux axes, or David--Hertling block coordinates.  

We shall use the product-invertibility and product-non-characteristic terminology introduced before Theorem~\ref{thm:main-restated}.

The point is to separate two logically distinct statements.  First, a single symmetry $Y$ produces a solution through an implicit algebraic equation (Theorem~\ref{thm:intrinsic-hodograph-safe}).  Second, a product-non-characteristic analytic Cauchy problem for the F-system is solved by a uniquely determined symmetry and the same intrinsic hodograph formula, and the map so obtained is the unique real analytic solution of the Cauchy problem (Theorem~\ref{thm:intrinsic-hodograph-completeness}).  The construction of this symmetry uses the non-characteristic Cauchy theorem for the symmetry equation, Theorem~\ref{thm:main-restated}. 

The following result was proved in  \cite{LPVG2}.  We recast it here in the analytic setting,  for its use in the results that follow below.

\begin{theorem}[Hodograph formula from one analytic symmetry]\label{thm:intrinsic-hodograph-safe}
Let $(M,\circ,e)$, $X$ and $\tnabla$ satisfy the geometric hypotheses of Theorem~\ref{thm:main-restated} on an analytic open set, and let $Y$ be a local analytic symmetry,
\[
        d_{\tnabla}(Y\circ)=0.
\]
Consider the intrinsic hodograph equation
\begin{equation}\label{eq:intrinsic-hodograph-formula}
        Y(u)=x\,e(u)+t\,X(u),
\end{equation}
for an unknown map $u=u(x,t)$ with values in $M$.  Suppose that at a point $(u_*,x_*,t_*)$ equation \eqref{eq:intrinsic-hodograph-formula} holds and that multiplication by
\begin{equation}\label{eq:hodograph-Z}
        Z_*:=\tnabla_eY\big|_{u_*}-t_*\,\tnabla_eX\big|_{u_*}
\end{equation}
is invertible on $T_{u_*}M$.  Then, after shrinking neighbourhoods, \eqref{eq:intrinsic-hodograph-formula} determines a unique analytic map $u(x,t)$ with $u(x_*,t_*)=u_*$, and this map solves the F-system
\[
        u_t=X\circ u_x .
\]
\end{theorem}

\begin{proof}
 Let us recall the main steps of the proof. Let
\[
        F(u,x,t):=Y(u)-x e(u)-tX(u).
\]
Using Lemma~\ref{lem:symmetry-equiv-PDE} it turns out that the $u$-Jacobian of $F$  evaluated at $(u_*,x_*,t_*)$ is
 \begin{equation}\label{eq:hodograph-DuF-multiplication}
        (D_uF)_{(u_*,x_*,t_*)}
        =\bigl(\tnabla_eY-t_*\tnabla_eX\bigr)\big|_{u_*}\circ
        =Z_*\circ.
\end{equation}
By hypothesis this operator is invertible.  The analytic implicit function theorem therefore gives, after shrinking neighbourhoods of $(x_*,t_*)$ and $u_*$, a unique analytic map
\[
        u=u(x,t),\qquad u(x_*,t_*)=u_*,
\]
satisfying the hodograph equation \eqref{eq:intrinsic-hodograph-formula}.  Since invertibility is an open condition, we may shrink once more so that multiplication by
\begin{equation}\label{eq:hodograph-Z-along-solution}
        Z(x,t):=\tnabla_eY\big|_{u(x,t)}-t\,\tnabla_eX\big|_{u(x,t)}
\end{equation}
is invertible for all $(x,t)$ under consideration.

It remains to prove that this implicit solution satisfies the F-system. Differentiating the identity $F=0$ with respect to $x$ and $t$ and 
using again Lemma~\ref{lem:symmetry-equiv-PDE} for $Y$ and the identity $d_{\tnabla}(X\circ)=0$ for $X$, one gets the identities
\begin{equation}\label{eq:hodograph-ux-ut-relations}
        u_x\circ Z=e,\qquad u_t\circ Z=X.
\end{equation}

By associativity,
\[
        (X\circ u_x)\circ Z
        =X\circ(u_x\circ Z)
        =X=u_t\circ Z.
\]
Since multiplication by $Z$ is invertible on the neighbourhood under consideration, we can conclude that the analytic map obtained from the hodograph equation is a solution of the F-system. 
\end{proof}

The non-characteristic Cauchy theorem for symmetries, Theorem~\ref{thm:main-restated}, is the substitute for Darboux completeness in the present generalized context.  In general it does not necessarily provide an explicit closed formula for the required symmetry.  However, when the physical Cauchy curve is sufficiently simple, the same theorem can still be effective and can lead to closed-form solutions beyond the currently available Darboux--Tsarev methods; see Example~\ref{ex:DH-hodograph-not-GT} below.  The product-invertibility of the physical tangent is precisely the intrinsic version of the non-characteristic assumptions appearing in Tsarev's diagonal theory, where one requires $\partial_x u^i(x,t)|_{(x_0,t_0)}\ne0$, and in the David--Hertling regular setting of \cite{LPVG2}, where one requires $\partial_xu^{1(\alpha)}(x,t)|_{(x_0,t_0)}\ne0$ for every block.

\begin{theorem}[Completeness of the hodograph representation]\label{thm:intrinsic-hodograph-completeness}
Assume the geometric hypotheses of Theorem~\ref{thm:main-restated}.  Let $t_0\in\mathbb R$, let $I\subset\mathbb R$ be an interval containing $x_0$, and let
\[
        \varphi:I\longrightarrow U_{\mathrm{cyc}}
\]
be a real analytic initial datum.  Consider the Cauchy problem
\begin{equation}\label{eq:intrinsic-Cauchy-problem}
        \begin{cases}
        u_t=X(u)\circ u_x,\\[2mm]
        u(x,t_0)=\varphi(x).
        \end{cases}
\end{equation}
Set
\[
        u_0:=\varphi(x_0),
\]
and assume that the initial tangent
\[
        u_x(x_0,t_0):=\varphi'(x_0)\in T_{u_0}M
\]
is product-invertible.

Then, after possibly shrinking $I$ around $x_0$, there exists a unique local real analytic symmetry $Y$, satisfying
\[
        d_{\tnabla}(Y\circ)=0,
\]
which is characterized by the matching condition along the physical Cauchy curve $C(x):=\varphi(x)$:
\begin{equation}\label{eq:hodograph-matching-condition-cauchy}
        Y(C(x))=x\,e(C(x))+t_0X(C(x)).
\end{equation}
For this symmetry the hodograph non-degeneracy condition at $(u_0,x_0,t_0)$ is automatically satisfied: multiplication by
\begin{equation}\label{eq:hodograph-Z-Cauchy}
        Z_0:=\tnabla_eY\big|_{u_0}-t_0\,\tnabla_eX\big|_{u_0}
\end{equation}
is invertible on $T_{u_0}M$.  More precisely,
\[
        Z_0=\varphi'(x_0)^{-1}
\]
with respect to the product $\circ$.

Consequently, after shrinking neighbourhoods of $x_0$, $t_0$, and $u_0$, the intrinsic hodograph equation
\begin{equation}\label{eq:complete-hodograph-formula}
        Y(u(x,t))=x\,e(u(x,t))+t\,X(u(x,t))
\end{equation}
determines a real analytic map $u(x,t)$ with $u(x_0,t_0)=u_0$.  This map is a real analytic solution of the Cauchy problem \eqref{eq:intrinsic-Cauchy-problem} near $(x_0,t_0)$.

Moreover, $u$ is the unique real analytic solution of the Cauchy problem \eqref{eq:intrinsic-Cauchy-problem} near $(x_0,t_0)$: any real analytic solution of \eqref{eq:intrinsic-Cauchy-problem}, defined on a connected open neighbourhood of $(x_0,t_0)$, coincides with $u$ on a neighbourhood of $(x_0,t_0)$.  
\end{theorem}

\begin{proof}
Since product-invertibility is open and $\varphi'(x_0)$ is product-invertible, after shrinking $I$ around $x_0$ the tangent vector
\[
        C'(x)=\varphi'(x)
\]
is product-invertible for every $x\in I$, where $C(x):=\varphi(x)$.  Since $C'(x_0)\neq 0$, we may shrink once more and assume that $C$ is embedded.

Define an analytic vector field along $C$ by
\begin{equation}\label{eq:hodograph-matching-data}
        Y_C(C(x)):=x\,e(C(x))+t_0X(C(x)).
\end{equation}
By Theorem~\ref{thm:main-restated}, there exists, after shrinking a neighbourhood of $C(I)$ in $U_{\mathrm{cyc}}$, a unique local analytic vector field $Y$ such that
\[
        Y|_C=Y_C,
        \qquad
        d_{\tnabla}(Y\circ)=0.
\]
This proves the asserted construction and uniqueness of the matching symmetry.

We now verify that the non-degeneracy condition required in Theorem~\ref{thm:intrinsic-hodograph-safe} is not an additional assumption here.  Differentiating \eqref{eq:hodograph-matching-data} with respect to $x$ and using $\tnabla e=0$ gives
\[
        \tnabla_{C'}Y
        =e+t_0\tnabla_{C'}X.
\]
By Lemma~\ref{lem:symmetry-equiv-PDE}, applied to $Y$ and also to $X$ since $d_{\tnabla}(X\circ)=0$, this becomes
\[
        C'\circ\tnabla_eY
        =e+t_0C'\circ\tnabla_eX.
\]
Equivalently,
\begin{equation}\label{eq:initial-inverse-Z}
        C'\circ(\tnabla_eY-t_0\tnabla_eX)=e.
\end{equation}
Thus, along the initial curve,
\[
        \tnabla_eY\big|_{C(x)}-t_0\tnabla_eX\big|_{C(x)}
        =C'(x)^{-1}
\]
with respect to the product $\circ$.  In particular, at $x=x_0$ the vector $Z_0$ defined in \eqref{eq:hodograph-Z-Cauchy} is product-invertible, and the hodograph non-degeneracy condition is automatically satisfied.

The hypotheses of Theorem~\ref{thm:intrinsic-hodograph-safe} are therefore satisfied at $(u_0,x_0,t_0)$.  Hence, after shrinking neighbourhoods, the implicit equation
\[
        Y(\widetilde u)=x\,e(\widetilde u)+t\,X(\widetilde u)
\]
defines a real analytic map $\widetilde u(x,t)$ with $\widetilde u(x_0,t_0)=u_0$, and this map solves
\[
        \widetilde u_t=X(\widetilde u)\circ \widetilde u_x.
\]

It remains to check the Cauchy datum.  At $t=t_0$, the curve $C(x)=\varphi(x)$ satisfies the same implicit equation by construction:
\[
        Y(C(x))=x\,e(C(x))+t_0X(C(x)).
\]
For $x$ sufficiently close to $x_0$, the point $C(x)$ lies in the neighbourhood where the implicit-function branch is unique.  Hence
\[
        \widetilde u(x,t_0)=C(x)=\varphi(x)
\]
near $x_0$.  Therefore $\widetilde u$ is a real analytic solution of the Cauchy problem \eqref{eq:intrinsic-Cauchy-problem}.

It remains to prove uniqueness within the real analytic class.  This is the uniqueness part of the Cauchy--Kowalevski theorem applied to the F-system itself, which is an analytic evolutionary system in Cauchy--Kowalevski normal form \cite{CK-CourantHilbert,CK-John,CK-Hormander}.  After shrinking, every solution close to $u_0$ takes its values in a fixed analytic coordinate chart around $u_0$, in which the F-system reads
\begin{equation}\label{eq:F-system-CK-normal-form}
        \partial_tu^i=V^i_k(u)\,\partial_xu^k,
        \qquad
        V^i_k(u):=c^i_{jk}(u)\,X^j(u),
\end{equation}
with analytic coefficients $V^i_k$: the system is solved for the $t$-derivative, so the initial line $\{t=t_0\}$ is non-characteristic for it.  The standard uniqueness clause of the Cauchy-Kowalevski theorem gives the uniqueness of the solution. 
Note that this last part of the argument uses only the evolutionary form \eqref{eq:F-system-CK-normal-form} of the F-system, and not the product-invertibility of $\varphi'$.

\end{proof}

Due to Theorem~\ref{thm:intrinsic-hodograph-completeness}  every product-non-characteristic real analytic solution of the F-system is locally representable in intrinsic hodograph form.  The proof above replaces the Darboux completeness theorem used in the Darboux--Tsarev setting by the non-characteristic Cauchy--Kowalevski Theorem~\ref{thm:main-restated}. As we already remarked,  Theorem~\ref{thm:main-restated} does not always produce a closed formula for $Y$.  Nevertheless, in some cases it allows one to go beyond the set-up of the Darboux--Tsarev class, as the following example clearly shows.

\begin{example}[A David--Hertling example outside the Darboux--Tsarev completeness class]\label{ex:DH-hodograph-not-GT}
\normalfont In this example,  David--Hertling coordinates exist globally, but the additional Darboux--Tsarev completeness assumptions used in \cite{LPVG2} do not apply.  The intrinsic hodograph theorem nevertheless applies directly.

Let $M=\mathbb R^2$ with coordinates $(u^1,u^2)$ and frame
\[
        e_1:=\partial_{u^1},
        \qquad
        e_2:=\partial_{u^2}.
\]
Consider the one-block David--Hertling product of size two,
\[
        e_1\circ e_1=e_1,
        \qquad
        e_1\circ e_2=e_2,
        \qquad
        e_2\circ e_2=0.
\]
The unit is $e=e_1$.  This is the algebra of dual numbers in every tangent space, written in the standard David--Hertling coordinates.  Take
\begin{equation}\label{eq:DH-nonGT-X}
        X:=u^2e_1+e_2 .
\end{equation}
Then
\[
        e=e_1,
        \qquad
        X=u^2e_1+e_2
\]
are linearly independent everywhere, hence $X$ is cyclic.  The associated multiplication operator is
\[
        C_X=
        \begin{pmatrix}
        u^2 & 0\\
        1 & u^2
        \end{pmatrix}
\]
in the frame $(e_1,e_2)$.

The corresponding F-system
\[
        u_t=X(u)\circ u_x,
        \qquad
        u=(u^1,u^2),
\]
is
\begin{equation}\label{eq:DH-nonGT-hydro-system}
        u^1_t=u^2u^1_x,
        \qquad
        u^2_t=u^1_x+u^2u^2_x .
\end{equation}

This example is not in the Darboux--Tsarev completeness class singled out in \cite{LPVG2}.  For a single Jordan block of size two, the triangular dependence condition used there requires the first component of the velocity to depend only on the first coordinate, namely $X^1=X^1(u^1)$.  
Thus the Darboux completeness theorem of \cite{LPVG2} cannot be invoked, even though the coordinates are David--Hertling coordinates and $X$ is regularly non-degenerate in the sense of Proposition~\ref{prop:DH-cyclicity-equivalence}.

It is easy to check that for the natural connection $\tnabla$ in the current cooordinate system
\begin{equation}\label{eq:DH-nonGT-connection}
        \Gamma^1_{22}=-1,
\end{equation}
and all other Christoffel symbols vanish.  It is immediate to check that this connection is flat.  In particular the  identity \eqref{eq:3RC-Zinside} (or equivalently \eqref{eq:3RC-outside}) required in Theorem~\ref{thm:main-restated} holds.

Let
\[
        Y=A(u^1,u^2)e_1+B(u^1,u^2)e_2.
\]
By Lemma~\ref{lem:symmetry-equiv-PDE}, the symmetry equation $d_{\tnabla}(Y\circ)=0$ is equivalent to
\[
        \tnabla_UY=U\circ\tnabla_eY
        \qquad\text{for every }U.
\]
The equation for $U=e_2$ is the only nontrivial one.  Since
\[
        \tnabla_{e_1}Y=A_{u^1}e_1+B_{u^1}e_2
\]
and
\[
        e_2\circ\tnabla_{e_1}Y=A_{u^1}e_2,
\]
while, using \eqref{eq:DH-nonGT-connection},
\[
        \tnabla_{e_2}Y=(A_{u^2}-B)e_1+B_{u^2}e_2,
\]
we obtain the symmetry system
\begin{equation}\label{eq:DH-nonGT-symmetry-system}
        A_{u^2}=B,
        \qquad
        B_{u^2}=A_{u^1}.
\end{equation}
This system is not in the Darboux form used in the Darboux--Tsarev completeness argument: the second equation contains the tangential derivative $A_{u^1}$.

We now use the intrinsic hodograph method to solve a two-parameter family of Cauchy problems with linear-in-$x$ initial data for
\eqref{eq:DH-nonGT-hydro-system}.  Let
\begin{equation}\label{eq:DH-nonGT-cauchy-data}
        u^1(x,0)=\alpha x,
        \qquad
        u^2(x,0)=\beta x,
        \qquad
        (\alpha,\beta)\ne(0,0).
\end{equation}
The physical initial curve in the target manifold is
\[
        C_{\alpha,\beta}(x)=(\alpha x,\beta x).
\]
Since \(t_0=0\), the matching condition in Theorem~\ref{thm:intrinsic-hodograph-completeness} is
\begin{equation}\label{eq:DH-nonGT-matching-for-param-curve}
        Y(C_{\alpha,\beta}(x))=x e(C_{\alpha,\beta}(x)).
\end{equation}
Writing \(Y=Ae_1+Be_2\), and recalling that \(e=e_1\), this becomes
\begin{equation}\label{eq:DH-nonGT-matching-components}
        A(\alpha x,\beta x)=x,
        \qquad
        B(\alpha x,\beta x)=0.
\end{equation}
The non-characteristic condition is also transparent.  We have
\[
        C'_{\alpha,\beta}(x)=\alpha e_1+\beta e_2.
\]
In the dual-number algebra an element \(a e_1+b e_2\) is product-invertible if and only if \(a\ne0\), and then
\[
        (a e_1+b e_2)^{-1}=\frac1a e_1-\frac{b}{a^2}e_2.
\]
Thus \(C'_{\alpha,\beta}\) is product-invertible exactly when
\[
        \alpha\ne0.
\]
We first treat this non-characteristic case, where the intrinsic Cauchy--Kowalevski/hodograph theorem applies directly.  The characteristic case \(\alpha=0,\beta\ne0\) is discussed at the end, in order to show why the non-characteristic hypothesis is necessary.

\medskip
\noindent\textbf{The matching symmetry for \(\alpha\ne0\).}
Since \(B=A_{u^2}\), the symmetry system \eqref{eq:DH-nonGT-symmetry-system} is equivalent to
\begin{equation}\label{eq:DH-nonGT-heat-equation-A}
        A_{u^2u^2}=A_{u^1},
        \qquad
        B=A_{u^2}.
\end{equation}
Introduce coordinates adapted to the initial curve:
\[
        \lambda:=\frac{u^1}{\alpha},
        \qquad
        \eta:=u^2-\frac{\beta}{\alpha}u^1.
\]
Equivalently,
\[
        u^1=\alpha\lambda,
        \qquad
        u^2=\beta\lambda+\eta.
\]
Then the curve \(C_{\alpha,\beta}\) is given by \(\eta=0\) and \(\lambda=x\).  Set
\[
        \mathcal A(\lambda,\eta):=A(\alpha\lambda,\beta\lambda+\eta),
        \qquad
        \mathcal B(\lambda,\eta):=B(\alpha\lambda,\beta\lambda+\eta).
\]
Since
\[
        \partial_{u^2}=\partial_\eta,
        \qquad
        \partial_{u^1}=\frac1{\alpha}\partial_\lambda
        -\frac{\beta}{\alpha}\partial_\eta,
\]
the symmetry system becomes
\begin{equation}\label{eq:DH-nonGT-adapted-symmetry-system}
        \mathcal A_\eta=\mathcal B,
        \qquad
        \mathcal B_\eta=\frac1{\alpha}\mathcal A_\lambda
        -\frac{\beta}{\alpha}\mathcal B.
\end{equation}
The matching conditions become
\begin{equation}\label{eq:DH-nonGT-adapted-cauchy-data}
        \mathcal A(\lambda,0)=\lambda,
        \qquad
        \mathcal B(\lambda,0)=0.
\end{equation}
Equivalently, using \(\mathcal B=\mathcal A_\eta\), one has
\[
        \mathcal A(\lambda,0)=\lambda,
        \qquad
        \mathcal A_\eta(\lambda,0)=0.
\]
Because the Cauchy datum is affine in \(\lambda\), we look for
\[
        \mathcal A(\lambda,\eta)=\lambda+f(\eta).
\]
Then \(\mathcal B=f'(\eta)\), and the second equation in
\eqref{eq:DH-nonGT-adapted-symmetry-system} gives
\begin{equation}\label{eq:DH-nonGT-f-ode}
        f''(\eta)+\frac{\beta}{\alpha}f'(\eta)=\frac1{\alpha},
        \qquad
        f(0)=f'(0)=0.
\end{equation}

If \(\beta\ne0\), the solution of \eqref{eq:DH-nonGT-f-ode} is
\begin{equation}\label{eq:DH-nonGT-f-beta-nonzero}
        f'(\eta)=\frac1{\beta}
        \left(1-e^{-\frac{\beta}{\alpha}\eta}\right),
        \qquad
        f(\eta)=\frac{\eta}{\beta}
        +\frac{\alpha}{\beta^2}
        \left(e^{-\frac{\beta}{\alpha}\eta}-1\right).
\end{equation}
Therefore, putting
\[
        \eta=u^2-\frac{\beta}{\alpha}u^1,
\]
the matching symmetry is
\begin{equation}\label{eq:DH-nonGT-symmetry-beta-nonzero}
        A(u^1,u^2)=\frac{u^1}{\alpha}
        +\frac{\eta}{\beta}
        +\frac{\alpha}{\beta^2}
        \left(e^{-\frac{\beta}{\alpha}\eta}-1\right),
        \qquad
        B(u^1,u^2)=\frac1{\beta}
        \left(1-e^{-\frac{\beta}{\alpha}\eta}\right).
\end{equation}
Indeed, \(A_{u^2}=B\), and
\[
        B_{u^2}=\frac1{\alpha}e^{-\frac{\beta}{\alpha}\eta}
        =\frac1{\alpha}-\frac{\beta}{\alpha}B=A_{u^1},
\]
so \eqref{eq:DH-nonGT-symmetry-system} holds.  Along the physical initial curve, \(\eta=0\) and \(u^1/\alpha=x\), hence
\[
        A(\alpha x,\beta x)=x,
        \qquad
        B(\alpha x,\beta x)=0,
\]
as required.

If \(\beta=0\), equation \eqref{eq:DH-nonGT-f-ode} becomes \(f''=1/\alpha\). Thus
\begin{equation}\label{eq:DH-nonGT-symmetry-beta-zero}
        f(\eta)=\frac{\eta^2}{2\alpha},
        \qquad
        f'(\eta)=\frac{\eta}{\alpha},
\end{equation}
and therefore
\begin{equation}\label{eq:DH-nonGT-Y-beta-zero}
        A(u^1,u^2)=\frac{u^1}{\alpha}+\frac{(u^2)^2}{2\alpha},
        \qquad
        B(u^1,u^2)=\frac{u^2}{\alpha}.
\end{equation}
Again \(A_{u^2}=B\), \(B_{u^2}=A_{u^1}=1/\alpha\), and the matching conditions along \(u^1=\alpha x,u^2=0\) are satisfied.

\medskip
\noindent\textbf{Hodograph solution for \(\alpha\ne0,\beta\ne0\).}
The intrinsic hodograph formula is
\[
        Y(u)=x e(u)+tX(u).
\]
Since \(e=e_1\) and \(X=u^2e_1+e_2\), this is the scalar system
\begin{equation}\label{eq:DH-nonGT-hodograph-equations-param}
        A(u^1,u^2)=x+t u^2,
        \qquad
        B(u^1,u^2)=t.
\end{equation}
Using \eqref{eq:DH-nonGT-symmetry-beta-nonzero}, the second equation gives
\[
        t=\frac1{\beta}\left(1-e^{-\frac{\beta}{\alpha}\eta}\right).
\]
Set
\[
        L:=1-\beta t.
\]
Then
\[
        e^{-\frac{\beta}{\alpha}\eta}=L,
        \qquad
        \eta=-\frac{\alpha}{\beta}\log L.
\]
Write \(u^1=\alpha\lambda\), \(u^2=\beta\lambda+\eta\).  The first hodograph equation is
\[
        \lambda+f(\eta)=x+t(\beta\lambda+\eta),
\]
or
\[
        L\lambda=x+t\eta-f(\eta).
\]
Since, by \eqref{eq:DH-nonGT-f-beta-nonzero},
\[
        f(\eta)=\frac{\eta}{\beta}+\frac{\alpha}{\beta^2}(L-1)
        =\frac{\eta}{\beta}-\frac{\alpha}{\beta}t,
\]
we obtain
\[
        \lambda=\frac{x}{L}+\frac{\alpha}{\beta^2}\log L
        +\frac{\alpha t}{\beta L}.
\]
Consequently
\begin{equation}\label{eq:DH-nonGT-solution-beta-nonzero}
        u^2(x,t)=\frac{\beta x+\alpha t}{1-\beta t},
        \qquad
        u^1(x,t)=
        \frac{\alpha(\beta x+\alpha t)}{\beta(1-\beta t)}
        +\frac{\alpha^2}{\beta^2}\log(1-\beta t).
\end{equation}
This is the local branch near \(t=0\), with \(\log(1)=0\).  The initial data are immediate:
\[
        u^2(x,0)=\beta x,
        \qquad
        u^1(x,0)=\alpha x.
\]
One can easily check that the equations are verified directly. 

\medskip
\noindent\textbf{The limiting non-characteristic case \(\beta=0\).}
Using \eqref{eq:DH-nonGT-Y-beta-zero}, the hodograph equations are
\[
        \frac{u^1}{\alpha}+\frac{(u^2)^2}{2\alpha}=x+t u^2,
        \qquad
        \frac{u^2}{\alpha}=t.
\]
Hence
\begin{equation}\label{eq:DH-nonGT-solution-beta-zero}
        u^2(x,t)=\alpha t,
        \qquad
        u^1(x,t)=\alpha x+\frac{\alpha^2}{2}t^2.
\end{equation}
This is also the limit of \eqref{eq:DH-nonGT-solution-beta-nonzero} as \(\beta\to0\).  A direct check gives that these satisfy the system.

\medskip
\noindent\textbf{Nondegeneracy and the unit-curve data.}
For \(\alpha\ne0\), the physical initial tangent is
\[
        C'_{\alpha,\beta}(x)=\alpha e_1+\beta e_2,
\]
and its product inverse is
\begin{equation}\label{eq:DH-nonGT-physical-tangent-inverse}
        C'_{\alpha,\beta}(x)^{-1}=\frac1{\alpha}e_1-\frac{\beta}{\alpha^2}e_2.
\end{equation}
For the matching symmetries above, \(\tnabla_eX=\tnabla_{e_1}X=0\).  Along the initial curve, one computes
\[
        \tnabla_eY|_{C_{\alpha,\beta}(x)}
        =\frac1{\alpha}e_1-\frac{\beta}{\alpha^2}e_2
        =C'_{\alpha,\beta}(x)^{-1},
\]
in agreement with the identity in the proof of Theorem~\ref{thm:intrinsic-hodograph-completeness}.  For \(\beta\ne0\), along the full hodograph solution,
\[
        Z=\tnabla_eY-t\tnabla_eX
        =\frac{1-\beta t}{\alpha}
        \left(e_1-\frac{\beta}{\alpha}e_2\right),
\]
which is product-invertible for \(t\) sufficiently close to zero.  For \(\beta=0\), one simply has \(Z=(1/\alpha)e_1\).  Hence the intrinsic hodograph nondegeneracy condition is satisfied in all non-characteristic cases \(\alpha\ne0\).

\noindent\textbf{The characteristic case \(\alpha=0,\beta\ne0\).}
Here
\[
        C_{0,\beta}'(x)=\beta e_2,
\]
which is not product-invertible.  The stronger intrinsic hodograph theorem does not apply.  In fact, the matching condition is incompatible with the symmetry equation.  It would require
\[
        A(0,\beta x)=x,
        \qquad
        B(0,\beta x)=0.
\]
Differentiating the first equation with respect to \(x\) gives
\[
        \beta A_{u^2}(0,\beta x)=1.
\]
Since \(A_{u^2}=B\), this says \(B(0,\beta x)=1/\beta\), contradicting \(B(0,\beta x)=0\).  Hence no symmetry can satisfy the hodograph matching condition along this characteristic curve.

The hydrodynamic Cauchy problem itself is still elementary in this characteristic case: one has
\[
        u^1(x,t)=0,
        \qquad
        u^2(x,t)=\frac{\beta x}{1-\beta t}.
\]
Indeed, since \(u^1\equiv 0\) implies \(u^1_x=0\), the second equation \(u^2_t=u^1_x+u^2u^2_x\) of the system reduces to \(u^2_t=u^2u^2_x\), which is satisfied by the formula above, and \(u^2(x,0)=\beta x\).  However, it is not obtained from the nondegenerate intrinsic hodograph completeness theorem, precisely because the product-non-characteristic hypothesis fails.

This example shows that the intrinsic Cauchy--Kowalevski/hodograph method can solve a whole family of non-characteristic Cauchy problems in David--Hertling coordinates even when the Darboux--Tsarev completeness assumptions of \cite{LPVG2} fail.  The obstruction to the method of \cite{LPVG2} is the lack of triangular dependence \(X^1=X^1(u^1)\), whereas the intrinsic method uses the product-non-characteristic physical curve to construct the matching symmetry and then applies the coordinate-free hodograph formula.

One can easily construct an example where the product is semisimple on an open dense subset and it degenerates to a non-semisimple product on a closed hypersurface where David-Hertling coordinates do not exist, and still shows that the method described in this section can be carried over.  We leave the details to the reader. 

\end{example}

\section{Generalized Novikov Conjecture: Hamiltonian F-systems are integrable}

\subsection{From F-manifolds with compatible connection to Riemannian F-manifolds} Consider a manifold $M$ equipped with a commutative associative product $\circ$ on sections of $\T$. By a metric $g$ on $M$ we mean a symmetric non-degenerate bilinear pairing on sections of $\T$; in particular we do not require this pairing to be positive definite.  In this subsection,  $X$ stands for any local vector field, not neccessarily a cyclic one. 

\begin{definition}
We have the following:
\begin{itemize}
\item A metric $g$ on $M$ is called {\em invariant} or {\em compatible} with $\circ$ if\begin{equation}\label{gccomp.eq} g(X\circ Y, Z)=g(X, Y\circ Z),\end{equation}
for all local sections $X,Y, Z$ of $\T$.  It is immediate to see that such metrics are in bijective correspondence with sections $\theta$ of $\T^*$ of the form $\theta=e^{\flat}=:g(e, \cdot),$ via the formula $g(X,Y)=\theta(X\circ Y)$,  
\item A torsionless connection $\nabla$ on $\T$ is called {\em compatible} with $\circ$ if $(\nabla_X \circ)(Y,Z)$ is symmetric in all its entries, i.e. 
\begin{equation}\label{nablaccomp.eq}(\nabla_X \circ)(Y, Z)=(\nabla_Y \circ)(X,Z).\end{equation}
This is equivalent to requiring that the connection $\nabla^{\lambda}$ given by 
\[
\nabla^{\lambda}_X Y:=\nabla_X Y+\lambda X\circ Y,  \quad X, Y \in \T
\]
has Riemann tensor that does not depend on $\lambda$, i.e.\ $R^{\nabla^\lambda}=R^{\na}$. By a result of Hertling \cite{HertlingBook}, this automatically implies the condition \eqref{HMeq1free}, thus $(M, \circ)$ is an F-manifold in this case. 
\end{itemize}
\end{definition}

\begin{definition}\label{Fmancon.def}
Let $(M, \circ, e,  \na)$ be the datum of a manifold,  with a commutative associative product $\circ$ on local sections of $\T$,  with identity $e$ and a torsionless connection $\na$ on $\T$.  Then $(M, \circ, e, \na)$ is an {\em F-manifold with compatible connection} $\na$ if 
\begin{itemize}
\item $\nabla$ is compatible with $\circ$ in the sense of \eqref{nablaccomp.eq}.
\item the Riemann tensor $R^{\na}$ of $\nabla$ satisfies the following condition
\begin{equation}\label{Rc1.eq}
R^{\na}(X,Y)(Z\circ W)+R^{\na}(Y,Z)(X\circ W)+R^{\na}(Z,X)(Y\circ W)=0,
\end{equation}
where $X,Y,Z,W$ are arbitrary local sections of $\T$. 
\end{itemize}

It is an F-manifold with compatible connection and flat unit if additionally
\beq\label{eq: flat unit}\nabla e=0.\eeq 

\end{definition}
 
 \begin{remark}
 Given an invariant metric $g$, it is not true in general that the unique torsionless metric connection $\nabla^g$ associated to it is also compatible with the product. Vice versa, given a torsionless connection $\nabla$ compatible with the product, it is not true in general that one can find an invariant metric $g$ among all the metrics that are compatible with $\nabla$ (i.e.\ for which $\nabla g=0$).
 \end{remark}
 
 \begin{remark}\label{Fmnfwcompatconn_flatunit}
More concretely an  F-manifold with compatible connection and flat unit can be described as follows.  
Let  $c^i_{jk}$ be the structure constants of the product $\circ$,   let  
$\Gamma^i_{jk}$ be  the Christoffel symbols  of the connection $\nabla$, then 
\begin{itemize}
\item[(i)] the one-parameter family of  connections $\{\nabla^{\lambda}\}_{\lambda}$ defined by
$$\nabla^{\lambda}=\nabla+\lambda\circ,$$
gives a torsionless connection for any choice of $\lambda$ and moreover the Riemann tensor satisfies
\beq\label{curvlambda}
R^{\nabla^{\lambda}}(X,Y)(Z)=R^{\nabla}(X,Y)(Z),
\eeq
and fulfills the condition \eqref{rc-intri};
\item[(ii)] $e$ is the unit of the product;
\item[(iii)] $e$ is flat: $\nabla e=0$.
\end{itemize}
\hspace{-1em}For a given $\lambda$, the torsion and curvature are respectively given by
\begin{eqnarray*}
T^{(\lambda)k}_{ij}&=&\Gamma^k_{ij}-\Gamma^k_{ji}+\lambda(c^k_{ij}-c^k_{ji}),\\
R^{(\lambda)k}_{ijl}&=&R^k_{ijl}+\lambda(\nabla_i c^k_{jl}-\nabla_j c^k_{il})+\lambda^2(c^k_{im}c^m_{jl}-c^k_{jm}c^m_{il}),
\end{eqnarray*}
where $R^k_{ijl}$ is the Riemann tensor of $\nabla$. Thus, condition $\mathrm{(i)}$ without \eqref{rc-intri} is equivalent to
\begin{enumerate}
\item the vanishing of the torsion of $\nabla$;
\item the commutativity of the product  $\circ$;
\item the symmetry  in the lower indices of the tensor field $\nabla_l c^k_{ij}$;
\item the associativity of the product $\circ$.
\end{enumerate}

 \end{remark}
 
\begin{definition}\label{FmanR.def}(\cite{ABLR})
A Riemannian F-manifold is the datum of an F-manifold $(M, \circ, e)$ with a metric $g$ such that:
\begin{itemize}
\item the metric $g$ is invariant according to \eqref{gccomp.eq};
\item The following relation holds:
\begin{equation}\label{Rcg.eq}
R^{g}(X,Y)(Z\circ W)+R^{g}(Y,Z)(X\circ W)+R^{g}(Z,X)(Y\circ W)=0,
\end{equation}
where $R^{g}$ is the Riemann tensor of $\na^g$,  the Levi-Civita connection of $g$ and $X,Y,Z,W$ are arbitrary local sections of $\T$.
\end{itemize}
Furthermore, if the condition 
\begin{equation}\label{killing1.eq}
\mathcal{L}_e g=0
\end{equation}
holds, where $\mathcal{L}_e$ is the Lie derivative with respect to $e$, then $(M, \circ, e, g)$ is called a Riemannian F-manifold with Killing unit vector field. 
\end{definition}

Our next result shows that it is possible to reconstruct a family of Riemannian F-manifolds (not a unique one in general) starting from an F-manifold with compatible connection provided that the system of PDEs \eqref{gfromnabla.eq} below admit as a solution a metric $g$ invariant for $\circ$.  This generalizes Theorem 1.13  from \cite{ABLR} in which a family of Riemannian F-manifolds was constructed starting from a flat F-manifold. 

\begin{theorem}\label{Riemannfromcompatibleth}
Let $(M,  \circ,  e, \nabla)$ be an F-manifold with compatible connection $\na$  
and let $g$ be an invariant metric for $\circ$ such that $g$ satisfies the following system of PDEs:
\begin{equation}\label{gfromnabla.eq}
(\na_X g)(Y,Z)=\frac{1}{2}d\theta(X\circ Y, Z)+\frac{1}{2}d\theta(X\circ Z,Y)+\frac{1}{2}(\mathcal{L}_eg)(X\circ Y,Z)+\frac{1}{2}(\mathcal{L}_eg)(X\circ Z,Y),
\end{equation}
where $\theta$ is the counit and $X, Y, Z$ are local sections of $\T$.  
Then $(M, \circ, e,  g)$ defines a Riemannian F-manifold.
\end{theorem}

\begin{proof}

{\it Part 1}
Define a connection $\nabla^{\mathrm{def}}$ by 
\begin{equation}\label{gconn.eq}
\na^{\mathrm{def}}_XY:=\na_XY +\frac{1}{2}(\iota_{X\circ Y}d\theta)^{\sharp}+\frac{1}{2}((\mathcal{L}_e g)(X\circ Y,\cdot))^{\sharp},
\end{equation}
where for any local section $\alpha \in \T^*$,  $\alpha^{\sharp}$ is the local section of $\T$ obtained via the isomorphism $\sharp: \T^* \rightarrow \T$ induced by the cometric $g^*$. We claim that $\nabla^{\mathrm{def}}$ is torsionless and metric with respect to $g$; by the uniqueness of the Levi-Civita connection of $g$, it then coincides with $\nabla^g$, the Levi-Civita connection of $g$, justifying its name. The fact that $\na^{\mathrm{def}}$ is torsionless follows from \eqref{gconn.eq}, since $\na^{\mathrm{def}}_XY-\na^{\mathrm{def}}_YX=\na_XY-\na_YX=[X,Y]$, due to the fact that the additional terms in \eqref{gconn.eq} are symmetric under the exchange of $X$ and $Y$.  As for the metric property:
\begin{equation} 
\begin{split}
(\nabla^{\mathrm{def}}_X g)(Y,Z)& = X(g(Y,Z))-g(\na^{\mathrm{def}}_X Y, Z) -g(Y, \na^{\mathrm{def}}_X Z)=\\
 & = (\na_X g)(Y,Z)+g(\na_X Y, Z)+g(Y, \na_X Z)-g(\na^{\mathrm{def}}_X Y, Z) -g(Y, \na^{\mathrm{def}}_X Z)\stackrel{\eqref{gconn.eq}}{=}\\
 &=(\na_Xg)(Y,Z)-\frac{1}{2}d\theta(X\circ Y, Z)\\
 &-\frac{1}{2}d\theta(X\circ Z,Y)-\frac{1}{2}(\mathcal{L}_eg)(X\circ Y,Z)-\frac{1}{2}(\mathcal{L}_eg)(X\circ Z,Y)\stackrel{\eqref{gfromnabla.eq}}{=}0.
\end{split}
\end{equation}
Hence $\nabla^{\mathrm{def}}=\nabla^g$, and \eqref{gconn.eq} shows that $\nabla$ and $\nabla^g$ are related by an explicit $(1,1)$-tensor correction.

{\it Part 2} The connection $\nabla$ is compatible with the product and $\nabla$ and $\nabla^g$ are related via \eqref{gconn.eq},  which, notice,  can be written as  
\[ \nabla_XY =\nabla^g_X Y +W(X\circ Y),\]
for a suitable $(1,1)$-tensor field $W$.  Proposition 1.12 of \cite{ABLR} proves that in this case
\[R^{\nabla}(Y, Z)(X\circ W)+R^{\nabla}(X,Y)(Z\circ W)+R^{\nabla}(Z, X)(Y\circ W)=\]
\[R^{g}(Y, Z)(X\circ W)+R^{g}(X,Y)(Z\circ W)+R^{g}(Z, X)(Y\circ W),\]
where $R^{\nabla}$ is the curvature operator associated to $\nabla$ and $R^g$ is the one associated to $\nabla^g$.  The result follows immediately. 
\end{proof}

\begin{remark}
Theorem \ref{Riemannfromcompatibleth} depends on the existence of nontrivial solutions to equation \eqref{gfromnabla.eq}. In the case in which $\circ$ is semisimple, equation \eqref{gfromnabla.eq} always admits a solution, while in the nonsemisimple case, this is not always the case. In particular, if there is a solution to \eqref{gfromnabla.eq}, in general there are many, which is why from an F-manifold with compatible connection one reconstructs a family of Riemannian F-manifolds.
\end{remark}

The following Theorem generalizes Theorem 1.8 of  \cite{ABLR}.  It constructs an F-manifold with compatible connection (not flat in general) starting from a Riemannian F-manifold (not necessarily with Killing unit vector field in general).  In particular,  there is no obstruction to build an F-manifold with compatible connection from a Riemannian F-manifold,  unlike the reverse direction.  Furthermore,  the process provides only one F-manifold with compatible connection,  starting from a Riemannian F-manifold.

\begin{theorem}\label{Compatiblefromriemann.th}
Let $(M, \circ, e, g)$ be a Riemannian F-manifold.  Consider the connection given by 
\begin{equation}\label{conncomp1.eq}
  \na_XY:= \na^g_XY -\frac{1}{2}(\iota_{X\circ Y}d\theta)^{\sharp}-\frac{1}{2}((\mathcal{L}_e g)(X\circ Y,\cdot))^{\sharp},
\end{equation}
where $X, Y$ are arbitrary sections of $\T$, $\sharp$ is the isomorphism from $\T^*$ to $\T$ induced by the inverse metric and $\nabla^g$ is the Levi-Civita connection of $g$.
Then the data $(M,  \circ, e, \nabla)$ define an F-manifold with compatible connection.  Furthermore,  the connection $\nabla$ and the metric $g$ fulfill the system of PDEs \eqref{gfromnabla.eq}.
\end{theorem}
\proof
It is clear that $\nabla$ is torsionless since $\na_XY-\nabla_Y X=\na^g_XY-\na^g_YX=[X,Y]$.  It is also clear that $\na$ and $g$ satisfies the system of PDEs \eqref{gfromnabla.eq},  using the same proof in part 1 of Theorem \ref{Riemannfromcompatibleth}  and observing that in this case $\nabla^g g=0$ by assumption. 

Now we prove that $\na$ is compatible with $\circ$.  We compute

\begin{equation}\label{Delta-full-expansion.eq}
\begin{split}
&(\nabla_X \circ)(Y,Z)-(\nabla_Y\circ)(X,Z)=\\
&\nabla_X(Y\circ Z)-\nabla_X Y\circ Z-Y\circ \nabla_X Z-\nabla_Y(X\circ Z)+\nabla_YX\circ Z+ X\circ \nabla_Y Z\\
&\stackrel{\eqref{conncomp1.eq}}{=}\nabla_X^g(Y\circ Z)-\cancel{\frac{1}{2}(\iota_{X\circ Y\circ Z} d\theta)^{\sharp}}-\bcancel{\frac{1}{2}\left((\mathcal{L}_e g)(X\circ Y\circ Z, \cdot)\right)^{\sharp}}-[X,Y]\circ Z\\
&-Y\circ \nabla^g_X Z+Y\circ \frac{1}{2}(\iota_{X\circ Z}d\theta)^{\sharp}+Y\circ \frac{1}{2}\left((\mathcal{L}_e g)(X\circ Z, \cdot) \right)^{\sharp}-\nabla^g_Y(X\circ Z)\\
&+\cancel{\frac{1}{2}(\iota_{Y\circ X\circ Z} d\theta)^{\sharp}}+\bcancel{\frac{1}{2}\left((\mathcal{L}_e g)(Y\circ X\circ Z, \cdot) \right)^{\sharp}}\\
&+X\circ \nabla^g_Y Z-X\circ \frac{1}{2}(\iota_{Y\circ Z} d\theta)^{\sharp}-X\circ \frac{1}{2}\left((\mathcal{L}_e g)(Y\circ Z, \cdot) \right)^{\sharp},
\end{split}
\end{equation}
where the crossed out terms cancel in pairs due to commutativity and associativity of $\circ$.  .

For the convenience of the reader, we now spell out in full detail the steps that lead to the cancellation of the terms that do not contain $\mathcal{L}_eg$. Denote by $\nabla^0$ the connection obtained from \eqref{conncomp1.eq} by dropping the $\mathcal{L}_eg$-correction, namely
\begin{equation}\label{nablazerocomp.eq}
\nabla^0_XY:=\nabla^g_XY-\frac{1}{2}(\iota_{X\circ Y}d\theta)^\sharp .
\end{equation}
First, we expand $(\nabla^0_X\circ)(Y,Z)=\nabla^0_X(Y\circ Z)-(\nabla^0_XY)\circ Z-Y\circ(\nabla^0_XZ)$, substituting \eqref{nablazerocomp.eq} in each of the three terms:
\[
(\nabla^0_X\circ)(Y,Z)=(\nabla^g_X\circ)(Y,Z)
-\frac{1}{2}\bigl(\iota_{X\circ Y\circ Z}\,d\theta\bigr)^\sharp
+\frac{1}{2}\bigl(\iota_{X\circ Y}\,d\theta\bigr)^\sharp\circ Z
+\frac{1}{2}\,Y\circ\bigl(\iota_{X\circ Z}\,d\theta\bigr)^\sharp .
\]
Subtracting from this the same expression with $X$ and $Y$ exchanged, the second summands on the right-hand side cancel,   by commutativity and associativity of $\circ$, and the third summands cancel, because $X\circ Y=Y\circ X$.  Hence the $\mathcal{L}_eg$-free part of \eqref{Delta-full-expansion.eq} is precisely
\[
\begin{split}
\Delta_0(X,Y,Z):=\;&(\nabla^0_X\circ)(Y,Z)-(\nabla^0_Y\circ)(X,Z)\\
=\;&(\nabla^g_X\circ)(Y,Z)-(\nabla^g_Y\circ)(X,Z)\\
&+\frac{1}{2}\,Y\circ\bigl(\iota_{X\circ Z}\,d\theta\bigr)^\sharp
-\frac{1}{2}\,X\circ\bigl(\iota_{Y\circ Z}\,d\theta\bigr)^\sharp .
\end{split}
\]
Now we prove that $\Delta_0=0$.  Since $g$ is non-degenerate, $\Delta_0=0$ if and only if
\[
g(\Delta_0(X,Y,Z),W)=0
\]
for every local vector field $W$.  By the definition of the isomorphism $\sharp$, the vector field $(\iota_{X\circ Z}d\theta)^\sharp$ is characterised by the property
\[
g\bigl((\iota_{X\circ Z}d\theta)^\sharp,\,U\bigr)=d\theta(X\circ Z,U)
\qquad\text{for every local vector field $U$.}
\]
Using first the invariance of $g$ and then this property with $U=Y\circ W$, we get
\[
g\bigl(Y\circ(\iota_{X\circ Z}d\theta)^\sharp,\,W\bigr)
=g\bigl((\iota_{X\circ Z}d\theta)^\sharp,\,Y\circ W\bigr)
=d\theta(X\circ Z,Y\circ W),
\]
and, exchanging $X$ and $Y$,
\[
g\bigl(X\circ(\iota_{Y\circ Z}d\theta)^\sharp,\,W\bigr)=d\theta(Y\circ Z,X\circ W).
\]
Pairing the above expression for $\Delta_0(X,Y,Z)$ with $W$ therefore yields
\begin{equation}\label{DeltazeroABLR.eq}
\begin{split}
g(\Delta_0(X,Y,Z),W)=&\,g((\nabla^g_X\circ)(Y,Z)-(\nabla^g_Y\circ)(X,Z),W)\\
&+\frac{1}{2}d\theta(X\circ Z,Y\circ W)-\frac{1}{2}d\theta(Y\circ Z,X\circ W).
\end{split}
\end{equation}
We now rewrite the first summand on the right-hand side of \eqref{DeltazeroABLR.eq}.  Let
\[
\mathcal{C}(U,V,S):=g(U\circ V,S)=\theta(U\circ V\circ S),
\]
where the second equality follows from the unit property and the invariance of $g$:
\[
g(U\circ V,S)=g\bigl((U\circ V)\circ e,\,S\bigr)=g\bigl(e,\,(U\circ V)\circ S\bigr)=\theta(U\circ V\circ S).
\]
In particular $\mathcal{C}$ is completely symmetric.  Moreover, since $\nabla^gg=0$,
\[
X\,\mathcal{C}(Y,Z,W)=X\,g(Y\circ Z,W)
=g\bigl(\nabla^g_X(Y\circ Z),W\bigr)+g\bigl(Y\circ Z,\nabla^g_XW\bigr),
\]
so that, expanding the covariant derivative of the tensor $\mathcal{C}$ and using $\mathcal{C}(Y,Z,\nabla^g_XW)=g(Y\circ Z,\nabla^g_XW)$,
\[
\begin{split}
(\nabla^g_X\mathcal{C})(Y,Z,W)
&=X\,\mathcal{C}(Y,Z,W)-\mathcal{C}(\nabla^g_XY,Z,W)-\mathcal{C}(Y,\nabla^g_XZ,W)-\mathcal{C}(Y,Z,\nabla^g_XW)\\
&=g\bigl(\nabla^g_X(Y\circ Z)-(\nabla^g_XY)\circ Z-Y\circ(\nabla^g_XZ),\,W\bigr)\\
&=g\bigl((\nabla^g_X\circ)(Y,Z),\,W\bigr).
\end{split}
\]
Thus the desired cancellation is equivalent to the intrinsic identity
\begin{equation}\label{ABLRkeyidentity.eq}
(\nabla^g_X\mathcal{C})(Y,Z,W)-(\nabla^g_Y\mathcal{C})(X,Z,W)
=-\frac{1}{2}d\theta(X\circ Z,Y\circ W)+\frac{1}{2}d\theta(Y\circ Z,X\circ W).
\end{equation}
This identity is established in the compatibility part of Theorem 1.8 of \cite{ABLR}, where, however, the unit vector field is Killing as part of the standing assumptions.  Since here no such assumption is made, we give a short intrinsic proof of \eqref{ABLRkeyidentity.eq}; it is the coordinate-free counterpart of the computation carried out in \cite{ABLR}, and it makes apparent that the Killing condition $\mathcal{L}_eg=0$ is never used.

Both sides of \eqref{ABLRkeyidentity.eq} are tensorial in $X,Y,Z,W$, so it is enough to prove the identity for pairwise commuting vector fields, for instance for coordinate vector fields.  Throughout, commutativity and associativity of $\circ$ are used through the complete symmetry of $\mathcal{C}$, while the invariance of $g$ and the unit property are used through
\[
g(U,V)=\theta(U\circ V),\qquad g(U,V\circ S)=\mathcal{C}(U,V,S)=\theta(U\circ V\circ S).
\]
Recall that $({\mathcal L}_V\circ)(U,S)=[V,U\circ S]-[V,U]\circ S-U\circ[V,S]$, so that, when $U$ and $S$ commute with $V$,
\begin{equation}\label{LVbracket.eq}
[V,U\circ S]=({\mathcal L}_V\circ)(U,S),
\qquad
[U\circ S,V]=-({\mathcal L}_V\circ)(U,S).
\end{equation}

Step 1: Koszul formula.  For commuting vector fields $X,Z$ the Koszul formula for $\nabla^g$ reads
\[
2g(\nabla^g_XZ,V)=Xg(Z,V)+Zg(X,V)-Vg(X,Z)-g([Z,V],X)+g([V,X],Z).
\]
Choosing $V=Y\circ W$ with $Y$ and $W$ commuting with both $X$ and $Z$ and using \eqref{LVbracket.eq}, we obtain
\begin{equation}\label{KoszulC.eq}
\begin{split}
2\,\mathcal{C}(Y,W,\nabla^g_XZ)=\;&X\,\mathcal{C}(Y,Z,W)+Z\,\mathcal{C}(X,Y,W)-(Y\circ W)\,\theta(X\circ Z)\\
&-g\bigl(X,({\mathcal L}_Z\circ)(Y,W)\bigr)-g\bigl(Z,({\mathcal L}_X\circ)(Y,W)\bigr).
\end{split}
\end{equation}
Exchanging $X$ and $Y$ in \eqref{KoszulC.eq} and subtracting, the terms $Z\,\mathcal{C}(X,Y,W)$ cancel by the symmetry of $\mathcal{C}$, and we get:
\begin{equation}\label{KoszulCdiff.eq}
\begin{split}
2\bigl[\mathcal{C}(Y,W,\nabla^g_XZ)-\mathcal{C}(X,W,\nabla^g_YZ)\bigr]
=\;&X\,\mathcal{C}(Y,Z,W)-Y\,\mathcal{C}(X,Z,W)\\
&-(Y\circ W)\,\theta(X\circ Z)+(X\circ W)\,\theta(Y\circ Z)\\
&-g\bigl(X,({\mathcal L}_Z\circ)(Y,W)\bigr)+g\bigl(Y,({\mathcal L}_Z\circ)(X,W)\bigr)\\
&-g\bigl(Z,({\mathcal L}_X\circ)(Y,W)-({\mathcal L}_Y\circ)(X,W)\bigr).
\end{split}
\end{equation}

Step 2: splitting off the $d\theta$-terms.  Since $[X,Y]=0$ and $\nabla^g$ is torsionless, $\nabla^g_XY=\nabla^g_YX$, so that, expanding the covariant derivatives of $\mathcal{C}$,
\[
\begin{split}
(\nabla^g_X\mathcal{C})(Y,Z,W)-(\nabla^g_Y\mathcal{C})(X,Z,W)
=\;&X\,\mathcal{C}(Y,Z,W)-Y\,\mathcal{C}(X,Z,W)\\
&-\bigl[\mathcal{C}(Y,W,\nabla^g_XZ)-\mathcal{C}(X,W,\nabla^g_YZ)\bigr]\\
&-\bigl[\mathcal{C}(Y,Z,\nabla^g_XW)-\mathcal{C}(X,Z,\nabla^g_YW)\bigr].
\end{split}
\]
Substituting \eqref{KoszulCdiff.eq} twice, the second time with $Z$ and $W$ exchanged, the derivatives of $\mathcal{C}$ cancel, and we are left with
\begin{equation}\label{intrinsicsplit.eq}
\begin{split}
2\bigl[(\nabla^g_X\mathcal{C})(Y,Z,W)-(\nabla^g_Y\mathcal{C})(X,Z,W)\bigr]
=\;&(Y\circ W)\,\theta(X\circ Z)-(X\circ Z)\,\theta(Y\circ W)\\
&+(Y\circ Z)\,\theta(X\circ W)-(X\circ W)\,\theta(Y\circ Z)\\
&+\Sigma(X,Y,Z,W),
\end{split}
\end{equation}
where
\begin{equation}\label{Sigmadef.eq}
\begin{split}
\Sigma(X,Y,Z,W):=\;
&g\bigl(X,({\mathcal L}_Z\circ)(Y,W)\bigr)-g\bigl(Y,({\mathcal L}_Z\circ)(X,W)\bigr)\\
+\;&g\bigl(X,({\mathcal L}_W\circ)(Y,Z)\bigr)-g\bigl(Y,({\mathcal L}_W\circ)(X,Z)\bigr)\\
+\;&g\bigl(Z,({\mathcal L}_X\circ)(Y,W)\bigr)-g\bigl(Z,({\mathcal L}_Y\circ)(X,W)\bigr)\\
+\;&g\bigl(W,({\mathcal L}_X\circ)(Y,Z)\bigr)-g\bigl(W,({\mathcal L}_Y\circ)(X,Z)\bigr).
\end{split}
\end{equation}
Since $d\theta(U,V)=U\theta(V)-V\theta(U)-\theta([U,V])$ for arbitrary vector fields $U$, $V$, the first four terms on the right-hand side of \eqref{intrinsicsplit.eq} equal
\[
d\theta(Y\circ W,X\circ Z)+d\theta(Y\circ Z,X\circ W)
+\theta\bigl([Y\circ W,X\circ Z]\bigr)+\theta\bigl([Y\circ Z,X\circ W]\bigr),
\]
and, by the antisymmetry of $d\theta$, the sum $d\theta(Y\circ W,X\circ Z)+d\theta(Y\circ Z,X\circ W)$ is exactly twice the right-hand side of \eqref{ABLRkeyidentity.eq}.  Hence \eqref{ABLRkeyidentity.eq} is equivalent to the vanishing of the residue
\begin{equation}\label{intrinsicresidue.eq}
\theta\bigl([Y\circ W,X\circ Z]\bigr)+\theta\bigl([Y\circ Z,X\circ W]\bigr)+\Sigma(X,Y,Z,W)=0.
\end{equation}

Step 3: the Hertling--Manin condition kills the residue.  This is the only step where the Hertling--Manin condition enters, in the form ${\mathcal L}_{U\circ V}\circ=U\circ{\mathcal L}_V\circ+V\circ{\mathcal L}_U\circ$ of \eqref{HMeq1free}.  Expanding the bracket by the Leibniz rule for ${\mathcal L}_{Y\circ W}$ and using \eqref{LVbracket.eq}, we get
\[
[Y\circ W,X\circ Z]
=({\mathcal L}_{Y\circ W}\circ)(X,Z)-({\mathcal L}_X\circ)(Y,W)\circ Z-X\circ({\mathcal L}_Z\circ)(Y,W).
\]
Applying $\theta$ and using $\theta(U\circ V)=g(U,V)$ together with the Hertling--Manin condition, which gives
\[
\theta\bigl(({\mathcal L}_{Y\circ W}\circ)(X,Z)\bigr)
=g\bigl(Y,({\mathcal L}_W\circ)(X,Z)\bigr)+g\bigl(W,({\mathcal L}_Y\circ)(X,Z)\bigr),
\]
we obtain
\begin{equation}\label{thetabracket.eq}
\begin{split}
\theta\bigl([Y\circ W,X\circ Z]\bigr)
=\;&g\bigl(Y,({\mathcal L}_W\circ)(X,Z)\bigr)+g\bigl(W,({\mathcal L}_Y\circ)(X,Z)\bigr)\\
&-g\bigl(Z,({\mathcal L}_X\circ)(Y,W)\bigr)-g\bigl(X,({\mathcal L}_Z\circ)(Y,W)\bigr),
\end{split}
\end{equation}
and, exchanging $Z$ and $W$,
\[
\begin{split}
\theta\bigl([Y\circ Z,X\circ W]\bigr)
=\;&g\bigl(Y,({\mathcal L}_Z\circ)(X,W)\bigr)+g\bigl(Z,({\mathcal L}_Y\circ)(X,W)\bigr)\\
&-g\bigl(W,({\mathcal L}_X\circ)(Y,Z)\bigr)-g\bigl(X,({\mathcal L}_W\circ)(Y,Z)\bigr).
\end{split}
\]
Substituting these two expressions into \eqref{intrinsicresidue.eq}, every term of $\Sigma(X,Y,Z,W)$ in \eqref{Sigmadef.eq} cancels against a term of the expanded brackets, term by term.  This proves \eqref{intrinsicresidue.eq}, hence \eqref{ABLRkeyidentity.eq} and, equivalently, $\Delta_0(X,Y,Z)=0$.

The point of the preceding computation is that it uses only the invariance of $g$ with respect to $\circ$, the fact that $e$ is the unit of $\circ$, associativity and commutativity of $\circ$, and the Hertling-Manin condition \eqref{HMeq1free}.  It does not use $\mathcal{L}_eg=0$.  Consequently, in the expression above all terms independent of $\mathcal{L}_eg$ cancel.  Therefore, to prove that $\nabla$ is compatible with $\circ$ it remains only to show that the terms containing $\mathcal{L}_eg$ cancel out.  For this, using the fact that $g$ is non-degenerate it is enough to show that for each local section $W$ of $\T$
\[g\left( Y\circ \frac{1}{2}\left((\mathcal{L}_e g)(X\circ Z, \cdot) \right)^{\sharp}-X\circ \frac{1}{2}\left((\mathcal{L}_e g)(Y\circ Z, \cdot) \right)^{\sharp}, W\right)= 0. \]
Using the invariance of $g$ with respect to $\circ$ and the fact that $\sharp$ is the isomorphism induced by the inverse metric between $\T^*$ and $\T$, we have that this is equivalent to showing that 
\begin{equation}\label{keycomp.eq}(\mathcal{L}_e g)(X\circ Z, Y\circ W)=(\mathcal{L}_e g)(Y\circ Z, X\circ W).\end{equation}
Now observe that because of the Hertling-Manin condition ${\mathcal L}_{X\circ Y}\circ =X\circ {\mathcal L}_Y \circ +Y\circ {\mathcal L}_X \circ$ with $X=Y=e$ one has ${\mathcal L}_e \circ=0$. 
Due to invariance one has $g(X\circ Z, Y\circ W)=g(Y\circ Z, X\circ W)$ and from this we get $e(g(X\circ Z, Y\circ W))=e(g(Y\circ Z, X\circ W)).$  Expanding both sides and using ${\mathcal L}_e \circ=0$ we get:
\begin{equation}
\begin{split}
\cancel{g(\Le X \circ Z, Y\circ W)}+\bcancel{g(X\circ \Le Z, Y\circ W)}+\xcancel{g(X\circ Z, \Le Y \circ W)}\\+\cancel{g(X\circ Z, Y\circ \Le W)}+(\Le g)(X\circ Z, Y\circ W)=\\
\xcancel{g(\Le Y\circ Z, X\circ W)}+\bcancel{g(Y \circ \Le Z, X\circ W)}+\cancel{g(Y\circ Z, \Le X\circ W)}\\+ \cancel{g(Y\circ Z, X\circ \Le W)}+(\Le g)(Y\circ Z, X\circ W).
\end{split}
\end{equation}
where the terms cancel in pairs due to invariance of $g$ and commutativity and associativity of $\circ.$ From this,  one gets \eqref{keycomp.eq}.

Finally,  using the second part of the proof of Theorem \ref{Riemannfromcompatibleth} (equivalently Proposition 1.12 in \cite{ABLR}), since  $\nabla$ is compatible with $\circ$,  we get condition \eqref{Rc1.eq} for the curvature $R^{\nabla}$ associated to $\nabla$,  since \eqref{Rc1.eq} holds by assumption for $R^g$,  the curvature of the metric $g$. 
\endproof

\subsection{The generalized Novikov's conjecture and its proof}\hfill\\

\noindent
{\bf Standing assumption: } We consider systems of the form \eqref{SHS2} on an F-manifold $(M,\circ,e)$ and assume that the vector field $X$ defining the system is cyclic, namely $e,X,X^{\circ 2},$
$\ldots,X^{\circ(n-1)}$ is a local frame.  When David--Hertling coordinates exist, this is equivalent to the usual regular non-degeneracy condition by Proposition \ref{prop:DH-cyclicity-equivalence}.  

\medskip 

{\bf Generalized Novikov's conjecture}: \emph{If a system of the form \eqref{SHS2} is Hamiltonian with respect to a local Hamiltonian structure, or with respect to a non-local Hamiltonian structure of hydrodynamic type with admissible Hamiltonian density (see below), then it satisfies the integrability condition \eqref{Rc1.eq}}.

\medskip 

The goal of this Section is to prove this statement.  The proof is based on the previous Section and on the following steps.  

First, from the assumption that system \eqref{SHS2} is Hamiltonian with respect to a local or non-local Hamiltonian structure of hydrodynamic type with admissible density,  we associate to it a canonical Riemannian F-manifold $(M, g, e, \circ)$ satisfying the additional condition $d_{\nabla^g}(X\circ)=0$,  where $X$ is the vector field defining \eqref{SHS2}.  In this Riemannian F-manifold, the unit vector field is not necessarily a Killing vector field. 

Leveraging in part the results of the previous Section, we will associate to this Riemannian F-manifold a canonical F-manifold with compatible connection $(M, \nabla, e, \circ)$ satisfying the further condition that $d_{\nabla}(X\circ)=0$ and such that $\nabla e=0$.
Finally, using the cyclic uniqueness theorem, Theorem \ref{thm: existence and uniqueness of cyclic connections}, we can identify the connection $\tnabla$ originating from the system \eqref{SHS2} with the connection $\nabla$ of the F-manifold with compatible connection, whenever both are characterized by $\nabla e=0$ and $d_\nabla(X\circ)=0$.  This will give the equation \eqref{Rc1.eq} for $\tnabla$. 

\medskip 

We first introduce general Hamiltonian operators following \cite{F} (see also \cite{CLV}).

\begin{theorem}\label{thm: Hamiltonian operator}\cite{F} Suppose $\det(g)\neq 0$,  where $g$ is a pseudo-Riemannian metric.  
Consider the operator 
\begin{equation}\label{eq: Hamiltonian operator}
P^{ij}:=g^{ij}\frac{d}{dx}-g^{is}\Gamma^j_{sk}u^k_x+\sum_{\alpha=1}^N \epsilon_{\alpha}(W_{(\alpha)})^i_ku^k_x  \left(\frac{d}{dx}\right)^{-1}(W_{(\alpha)})^j_h u^h_x, \quad \epsilon_{\alpha}=\pm 1,\end{equation}
where $\Gamma^i_{sk}$ are the Christoffel symbols of the Levi-Civita connection $\nabla^g$ associated to $g.$ Then  $P^{ij}$ is a Hamiltonian operator if and only if the $(1,1)$-tensor fields $W_{(\alpha)}$ satisfy the following conditions:
\begin{enumerate}
\item \begin{equation}\label{eq: 2ndDN} d_{\nabla^g}(W_{(\alpha)})=0;
\end{equation}
\item \begin{equation}\label{eq: 1stDN} g_{ik}(W_{(\alpha)})^k_j=g_{jk} (W_{(\alpha)})^k_i;
\end{equation}
\item \begin{equation}\label{eq: commuting affinor} [W_{(\alpha)}, W_{(\alpha')}]=0,  \quad \forall \alpha, \alpha'
\end{equation}
where $[\cdot, \cdot]$ is the commutator;
\item \begin{equation}\label{eq: R quadratic affinor} R^{ij}_{kh}=\sum_{\alpha=1}^N \epsilon_{\alpha}\left((W_{(\alpha)})^j_k (W_{(\alpha)})^i_h- (W_{(\alpha)})^i_k (W_{(\alpha)})^j_h\right),
\end{equation}
where $ R^{ij}_{kh}:=g^{is}R^j_{skh}$ and $R^j_{skh}$ are the components of the Riemann tensor of $g$.
\end{enumerate}
\end{theorem}

\begin{remark}	
The $(1,1)$-tensor fields $W_{(\alpha)}$ can be identified with the Weingarten operators of a submanifold (of codimension $N$) with flat normal bundle in a pseudo-Euclidean space. The metric $g$ coincides with the restriction of the pseudo-Euclidean metric on the submanifold (see \cite{F}).
\end{remark}	

\medskip

\noindent
{\bf Admissibility convention for Hamiltonian densities.}
In the standard formal calculus of non-local differential functions, the operator $D_x^{-1}$ applied to a hydrodynamic density $\eta_i(u) u^i_x$ produces an element of the algebra $\mathcal A_{\rm loc}(U)$ of local differential functions if and only if the one-form $\eta = \eta_i(u)\,du^i$ is exact:
\begin{equation}\label{eq:hydrodynamic-density-exactness}
        \eta_i(u) u^i_x\in D_x\mathcal A_{\rm loc}(U)
        \quad\Longleftrightarrow\quad
        \eta\in d\mathcal O(U),
\end{equation}
where $\mathcal O(U)$ denotes the chosen local category of functions (smooth, analytic, or holomorphic).  Indeed, suppose $D_xF=\eta_i(u) u^i_x$ for some $F\in\mathcal A_{\rm loc}(U)$.  If $F$ depended on a jet variable $u^i_{(m)}$ with $m\ge1$, then $D_xF$ would contain the independent variable $u^i_{(m+1)}$, a contradiction.  Hence $F=F(u)$ and $\eta=dF$; the converse is immediate.

Applying \eqref{eq:hydrodynamic-density-exactness} to each non-local term in $P^{ij}\partial_jh$, locality of the right-hand side of $u^i_t = P^{ij}\partial_jh$ separately for each $\alpha$ amounts to requiring that the density $(W_{(\alpha)})^j_l u^l_x \partial_jh$ is a $D_x$-derivative of a local function, equivalently that the one-form $W_{(\alpha)}^* dh$ is exact: there exists $F_\alpha\in\mathcal O(U)$ with
\begin{equation}\label{eq:admissibility-exact}
        W_{(\alpha)}^*\,dh = dF_\alpha,
        \qquad\text{i.e.,}\qquad
        (W_{(\alpha)})^j_l\,\partial_jh = \partial_l F_\alpha .
\end{equation}
A density $h$ satisfying \eqref{eq:admissibility-exact} for every $\alpha=1,\ldots,N$ will be called \emph{$P$-admissible} (or simply \emph{admissible} when $P$ is fixed).  This is a locality convention introduced here: it ensures that each non-local summand of $P^{ij}\partial_jh$ reduces to a local expression.  The need for such a restriction is already present in Ferapontov's formalism: for a general non-local tail, not every local functional produces a local Hamiltonian flow, and in the one-tail case the locality condition is precisely the exactness of the corresponding one-form $W^*dh$; see \cite{F}.   The local case $N=0$ is trivially admissible.

A system of the form \eqref{SHS2} is called  \emph{Hamiltonian} with respect to a local or a non-local Poisson structure of the form \eqref{eq: Hamiltonian operator} if it can be written as
\begin{equation}\label{eq: Hamiltonian system}
u^i_t=P^{ij}\frac{\delta H[u]}{\delta u^j},
\end{equation}
where $H[u]=\int h(u)\,dx$ is a local functional of hydrodynamic type with admissible density $h$, so that $\frac{\delta H}{\delta u^j}=\partial_jh$ and each non-local term in $P^{ij}\partial_jh$ is local.

\begin{rmk}\label{rmk:curvature-after-W-Yprod}
In the non-local case the automatic curvature identity for the Levi-Civita connection of $g$ should be understood after substituting the product form for the non-local affinors. Namely, Lemma~\ref{lem:W-is-Yprod} below gives $W_{(\alpha)}=Y_{(\alpha)}\circ$; substituting this into the quadratic curvature formula \eqref{eq: R quadratic affinor}, one obtains
 \beq\label{shc-LC}
R^s_{lmi}c^j_{ks}+R^s_{lik}c^j_{ms}+R^s_{lkm}c^j_{is}=0,\qquad 
R^j_{skl}c^s_{mi}+R^j_{smk}c^s_{li}+R^j_{slm}c^s_{ki}=0
\eeq
by associativity and commutativity of $\circ$.  Thus the point requiring proof is precisely the product form of the non-local affinors.
\end{rmk}

\medskip
\noindent
Recall, in addition, that the metric defining local (see \cite{DN84}) and non-local (\cite{F}) Hamiltonian operators for a system of hydrodynamic type 
\begin{equation}
u^i_t=V^i_j(u)u^j_x,
\end{equation}
satisfies the system  (see \cite{ts86,DN})
\begin{eqnarray}
\label{DN1}
g_{ik}V^k_j&=&g_{jk}V^k_i,\\
\label{DN2}
d_{\nabla^g}V&=&0.
\end{eqnarray}

The following lemma establishes the structural form of the non-local tails $W_{(\alpha)}$ for a Hamiltonian F-system and identifies the vector fields $Y_{(\alpha)}$ behind them as symmetries.

\begin{lemma}[Form of the non-local tail of a Hamiltonian F-system]\label{lem:W-is-Yprod}
Under the standing assumptions, suppose the F-system \eqref{SHS2} is Hamiltonian with respect to an operator of the form \eqref{eq: Hamiltonian operator}, with admissible Hamiltonian density $h$, Hamiltonian functional $H=\int h(u)\,dx$, and metric $g$.  Assume that the affinors $\{W_{(\alpha)}\}_{\alpha=1}^N$ satisfy conditions \eqref{eq: 2ndDN}--\eqref{eq: R quadratic affinor}. Then:
\begin{enumerate}
\item[(i)] $[C_X,W_{(\alpha)}]=0$ for every $\alpha=1,\dots,N$;
\item[(ii)] there exist (unique on the cyclic locus) vector fields $Y_{(\alpha)}\in\Gamma(\T)$ such that
\[
Y_{(\alpha)}\circ=C_{Y_{(\alpha)}}=W_{(\alpha)};
\]
\end{enumerate}
\end{lemma}

\proof
\emph{(i)} Let $h(u)$ be the Hamiltonian density, so that $H=\int h(u)\,dx$. By the admissibility convention adopted above, for every $\alpha=1,\ldots,N$ there exists a smooth function $F_\alpha(u)$ such that
\begin{equation}\label{eq:locality-W-h}
(W_{(\alpha)})^j_l\,\partial_j h = \partial_l F_\alpha,
\end{equation}
or equivalently $W_{(\alpha)}^*dh=dF_\alpha$ (this is \eqref{eq:admissibility-exact}).  Hamiltonicity of the F-system with respect to $P$ reads
\[
V^i_j u^j_x = P^{ij}\,\frac{\partial h}{\partial u^j},
\]
with $V:=C_X$. Expanding \eqref{eq: Hamiltonian operator}, using $\left(\frac{d}{dx}\right)^{-1}\!(\partial_x F_\alpha) = F_\alpha$ on each non-local summand, and matching coefficients of $u^l_x$, we get
\begin{equation}\label{eq:V-decomposition}
V^i_l = g^{ij}\,(\nabla^g_l \partial_j h) + \sum_\alpha \epsilon_\alpha\, F_\alpha\, (W_{(\alpha)})^i_l,
\end{equation}
where we have absorbed the connection term into the covariant Hessian $H_{jl} := \nabla^g_l \partial_j h = \partial_l\partial_j h - \Gamma^k_{lj}\partial_k h$, which is symmetric in $(j,l)$ by torsion-freeness of $\nabla^g$. Setting $H_g := g^{-1}H$, with components $(H_g)^i_l := g^{ij} H_{jl}$, equation \eqref{eq:V-decomposition} reads as the matrix relation
$V = H_g + \sum_\alpha \epsilon_\alpha F_\alpha W_{(\alpha)}$. We need to show that $[V, W_{(\alpha)}]=0$
for every $\alpha$.

By the commutation relations \eqref{eq: commuting affinor}, $[W_{(\alpha)}, W_{(\beta)}] = 0$ for every $\alpha$, hence
\begin{equation}\label{eq:V-Wbeta-commutator}
[V, W_{(\beta)}] = [H_g, W_{(\beta)}].
\end{equation}
It remains to show $[H_g, W_{(\beta)}] = 0$. Differentiate \eqref{eq:locality-W-h} covariantly with $\nabla^g_k$:
\[
(\nabla^g_k (W_{(\beta)})^j_l)\, \partial_j h + (W_{(\beta)})^j_l\, H_{jk} = \nabla^g_k \partial_l F_\beta.
\]
The right-hand side is the symmetric covariant Hessian of $F_\beta$, hence symmetric in $(k,l)$. By \eqref{eq: 2ndDN}, $d_{\nabla^g} W_{(\beta)} = 0$, equivalently $\nabla^g_k (W_{(\beta)})^j_l = \nabla^g_l (W_{(\beta)})^j_k$, so the first summand on the left is also symmetric in $(k,l)$. Consequently the second summand is symmetric in $(k,l)$:
\begin{equation}\label{eq:WbetaT-H-symmetric}
(W_{(\beta)})^j_l\, H_{jk} = (W_{(\beta)})^j_k\, H_{jl}.
\end{equation}
Let us spell out the matrix notation in \eqref{eq:WbetaT-H-symmetric}. We regard $W_{(\beta)}$ as the matrix whose $(i,j)$-entry is $(W_{(\beta)})^i_j$, while $H$ is the matrix with entries $H_{ij}$. Then
\[
\bigl(W_{(\beta)}^T H\bigr)_{lk}
= (W_{(\beta)})^j_l H_{jk},
\qquad
\bigl(W_{(\beta)}^T H\bigr)_{kl}
= (W_{(\beta)})^j_k H_{jl}.
\]
Thus \eqref{eq:WbetaT-H-symmetric} is precisely the equality of the $(l,k)$- and $(k,l)$-entries of $W_{(\beta)}^T H$, i.e. it says that $W_{(\beta)}^T H$ is symmetric. The transpose appears because in \eqref{eq:WbetaT-H-symmetric} the index of $W_{(\beta)}$ contracted with the first covariant index of $H$ is its upper index. Equivalently, $W_{(\beta)}$ is acting on a covariant slot, or on the one-form $dh$ in \eqref{eq:locality-W-h}, and hence through the dual map $W_{(\beta)}^*$, whose matrix is $W_{(\beta)}^T$. 
Since $H$ is itself symmetric, the symmetry of $W_{(\beta)}^T H$ is equivalent to $W_{(\beta)}^T H = H W_{(\beta)}$. The $g$-self-adjointness \eqref{eq: 1stDN}, in matrix form $W_{(\beta)}^T g = g W_{(\beta)}$, can be rewritten as $W_{(\beta)}^T = g\, W_{(\beta)}\, g^{-1}$. Substituting,
\[
g\, W_{(\beta)}\, g^{-1}\, H = H\, W_{(\beta)}
\quad\Longleftrightarrow\quad
W_{(\beta)}\, (g^{-1} H) = (g^{-1} H)\, W_{(\beta)},
\]
that is, $[H_g, W_{(\beta)}] = 0$. Combined with \eqref{eq:V-Wbeta-commutator}, this gives $[C_X, W_{(\beta)}] = 0$.

\emph{(ii)} This follows directly from the proof of Theorem \ref{thm:coordinate-free-prop34-lem35} (the same argument in the proof). 

\endproof

In the case of F-systems the condition  \eqref{DN1} reads:
\begin{equation}\label{eq: partial invariance}
g_{ij}c^j_{kl}X^l=g_{kj}c^j_{il}X^l
\end{equation}
where $X$ is the vector field defining the system \eqref{SHS2}.  Equations \eqref{eq: partial invariance} provide a \emph{partial invariance} of the metric with respect to the product, see equation \eqref{gccomp.eq}.  Under the cyclicity assumption on $X$, this partial invariance is enough to obtain the full invariance of $g$ without passing through David--Hertling coordinates.

\begin{lemma}
\label{lem:cyclic-self-adjointness}
Let $(M,\circ,e)$ be a manifold with a commutative associative product with unit $e$.  Let $X$ be a cyclic vector field and set
\[
        A:=C_X=X\circ.
\]

Let $g$ be a symmetric bilinear form on $TM$.  Assume that $A$ is self-adjoint with respect to $g$, namely
\begin{equation}\label{eq:CX-self-adjoint-g}
        g(AY,Z)=g(Y,AZ)
        \qquad \forall\,Y,Z.
\end{equation}
Then $g$ is invariant with respect to $\circ$:
\begin{equation}\label{eq:cyclic-g-invariance}
        g(W \circ Y,Z)=g(Y,W\circ Z)
        \qquad \forall\,Y,W,Z.
\end{equation}
\end{lemma}
\proof
Since $A=C_X$, associativity gives
\[
        A^k e=X^{\circ k}.
\]
By cyclicity, every local vector field $W$ can be written uniquely as
\[
        W=\sum_{k=0}^{n-1} a_k A^k e
        =\sum_{k=0}^{n-1} a_k X^{\circ k},
        \qquad a_k\in C^\infty(M).
\]
Consequently, the multiplication operator $C_W$ is
\begin{equation}\label{eq:CW-polynomial-in-A}
        C_W=\sum_{k=0}^{n-1} a_k C_{X^{\circ k}}=\sum_{k=0}^{n-1} a_k A^k.
\end{equation}

Since $A$ is self-adjoint with respect to $g$, every power $A^k$ is self-adjoint with respect to $g$.  
Thus every polynomial in $A$ with functional coefficients is self-adjoint with respect to $g$.  Applying this to \eqref{eq:CW-polynomial-in-A}, we obtain
\[
        g(Y\circ W,Z)=g(Y,W\circ Z).
\]
\endproof

\begin{proposition}\label{thm: system to Riemannian}
Under the standing assumptions,  suppose the system \eqref{SHS2} is Hamiltonian with respect to an operator of the form \eqref{eq: Hamiltonian operator}, with admissible Hamiltonian density $h$.  Associated to these data,  there is a canonical\footnote{The result is canonical relative to a chosen Hamiltonian operator and a chosen admissible density, not canonical for the F-system alone.} Riemannian F-manifold $(M, g, e, \circ)$ satisfying the additional condition \beq\label{eq: dnablagXprod} d_{\nabla^g}(X\circ)=0,\eeq
where $X$ is the vector field defining the system \eqref{SHS2}.
In general this Riemannian F-manifold is not equipped with a unit Killing vector field. 
\end{proposition}
\proof Firstly, Hamiltonianity gives
\[
        g_{ik}(X\circ)^k_j=g_{jk}(X\circ)^k_i
\]
(see \eqref{DN1}).  This is precisely the self-adjointness of $C_X=X\circ$ with respect to $g$.  Since $X$ is cyclic by the standing assumption,  Lemma ~\ref{lem:cyclic-self-adjointness} gives the full invariance of $g$ with respect to $\circ$:
\[
        g(Y\circ W,Z)=g(Y,W\circ Z)
        \qquad \forall\,Y,W,Z.
\]
Secondly, again Hamiltonianity gives $d_{\nabla^g}(X\circ)=0$ (see \eqref{DN2}).
Finally, in the local case $N=0$ the quadratic curvature formula \eqref{eq: R quadratic affinor} gives $R^g=0$, hence \eqref{Rcg.eq}.  In the non-local case, Lemma~\ref{lem:W-is-Yprod} gives $W_{(\alpha)}=Y_{(\alpha)}\circ$ for every non-local affinor; substituting this product form into \eqref{eq: R quadratic affinor} and using associativity and commutativity of $\circ$ gives \eqref{Rcg.eq}.
\endproof

\begin{theorem}\label{thm: key}
Under the standing assumptions, suppose the system \eqref{SHS2} is Hamiltonian with respect to an operator of the form \eqref{eq: Hamiltonian operator}, with admissible Hamiltonian density $h$. Associated to these data, there is a canonical (in the sense of the footnote attached to Proposition \ref{thm: system to Riemannian}) F-manifold with compatible connection $(M, \nabla, e, \circ)$, with flat unit $\nabla e=0$, satisfying the additional condition \beq\label{eq: dnablaXprod} d_{\nabla}(X\circ)=0,\eeq
where $X$ is the vector field defining \eqref{SHS2}.
\end{theorem}
\proof By Proposition \ref{thm: system to Riemannian} associated to \eqref{SHS2} there is a Riemannian F-manifold satisfying the additional condition \eqref{eq: dnablagXprod}.  By Theorem \ref{Compatiblefromriemann.th}, there is a unique F-manifold associated to this Riemannian F-manifold, and we just need to check that it satisfies \eqref{eq: dnablaXprod}.
This holds because:
\[d_{\nabla}(X\circ)(Y,Z)=\nabla_Y(X\circ Z)-\nabla_Z(X\circ Y)-X\circ[Y,Z]=\]
\[\cancel{\nabla^g_{Y}(X\circ Z)-\nabla^g_{Z}(X\circ Y)-X\circ [Y, Z]}\]
\[-\frac{1}{2}\left( i_{Y\circ (X\circ Z)}d\theta\right)^{\sharp}-\frac{1}{2}\left((\mathcal{L}_eg)(Y\circ (X\circ Z), \cdot) \right)^{\sharp}\]
\[ +\frac{1}{2}\left( i_{Z\circ (X\circ Y)}d\theta\right)^{\sharp}+\frac{1}{2}\left((\mathcal{L}_eg)(Z\circ (X\circ Y), \cdot) \right)^{\sharp}=0 \]
where the last two rows sum to zero since $\circ$ is associative and the first row is zero because \eqref{eq: dnablagXprod} holds. 

\noindent To prove $\nabla e=0,$ observe that for all sections $X, Y \in \T$
\[ g(Y, \nabla_Xe)=g(Y, \nabla^g_X e)-\frac{1}{2}g(Y, (\iota_{X}d\theta)^{\sharp})-\frac{1}{2}g(Y, ((\mathcal{L}_e g)(X,\cdot))^{\sharp})=\]
\[=g(Y, \nabla^g_X e)-\frac{1}{2}(d\theta)(X,Y)-\frac{1}{2}(\mathcal{L}_e g)(X,Y).\]
Using the Koszul formula for the Levi-Civita connection
\[ g(Y,  \nabla^g_Xe)=\frac{1}{2}\left[X(\theta(Y))+e(g(X,Y))-Y\theta(X)+g([X,e],Y)-g([e, Y], X)-\theta([X,Y])\right]\]
and the formula for $d\theta$ and substituting in the previous expression,  all the terms involving $\theta$ cancel out and we are left with
\[ g(Y, \nabla_Xe)=\cancel{\frac{1}{2}\left[e(g(X,Y))+g([X,e],Y)-g([e,Y],X)\right]}-\cancel{\frac{1}{2}(\mathcal{L}_e g)(X,Y)}=0\]
since 
\[ e(g(X,Y))=\mathcal{L}_e(g(X,Y))=(\mathcal{L}_e g)(X,Y)+g([e,X],Y)+g(X,[e,Y]).\]
Since $g$ is non-degenerate and $X,Y$ are arbitrary we get the claim. 
\endproof

\begin{theorem}
Under the standing assumptions,  if the system \eqref{SHS2} is Hamiltonian with respect to an operator of the form \eqref{eq: Hamiltonian operator}, with admissible Hamiltonian density, then it is integrable. Moreover, each vector field $Y_{(\alpha)}$ appearing in the non-local part of the operator defines a symmetry of \eqref{SHS2}: $d_{\tnabla}(Y_{(\alpha)}\circ)=0.$
\end{theorem}
\begin{proof}
Via Theorem \ref{thm: key}, associated to the system of hydrodynamic type \eqref{SHS2} there is an F-manifold with compatible connection $(M,\nabla, e, \circ)$, with the torsionless connection $\nabla$ satisfying the additional conditions $d_\nabla(X\circ)=0$ and $\nabla e=0$, where $X$ is the vector field defining \eqref{SHS2}. Since $\nabla$ is compatible with $\circ$ (Definition \ref{Fmancon.def}), it also satisfies the symmetry $(\nabla_W\circ)(Y,Z)=(\nabla_Y\circ)(W,Z)$. According to Theorem \ref{thm: existence and uniqueness of cyclic connections}, associated to \eqref{SHS2} there is a unique torsionless connection $\tnabla$ satisfying $d_{\tnabla}(X\circ)=0$, $(\tnabla_W\circ)(Y,Z)=(\tnabla_Y\circ)(W,Z)$, and $\tnabla e=0$. Therefore, the connection of the F-manifold $(M,\nabla, e, \circ)$ is equal to $\tnabla$. Since $\nabla$ satisfies the additional condition \eqref{Rc1.eq}, then so does $\tnabla$. But this is exactly the integrability condition for the system \eqref{SHS2}. The last statement of the theorem follows from \eqref{eq: 2ndDN} and from the identity
\[d_{\nabla^g}(Y\circ)=d_{\nabla}(Y\circ),\qquad \forall Y,\]
that follows immediately from the definition of $\nabla$ in terms of $\nabla^g$.
\end{proof}

\section{Conservative F-systems are integrable and integrable F-systems are conservative}

\begin{definition}\label{def:conservative}
The F-system 
\beq\label{Fsys1}
u^i_t=(X_0\circ u_x)^i=(C_{X_0} u_x)^i=c^i_{jk}X^j_0u^k_x,\qquad i\inw,
\eeq
is called \emph{locally conservative} if there exist
$n$ local density-flux pairs of conservation laws
\[
        (h^1,F^1),\ldots,(h^n,F^n),
\]
with independent differentials
\[
        \dd h^1\wedge\cdots\wedge \dd h^n\neq 0,
\]
such that every smooth solution of \eqref{Fsys1} satisfies
\begin{equation}\label{eq:conservation-laws}
        \partial_t h^a(u)=\partial_x F^a(u),
        \qquad a=1,\ldots,n.
\end{equation}
In particular,  and this is the form used below, the one-forms
$C_{X_0}^*\dd h^a$ are exact:
\begin{equation}\label{eq:flux-differential}
        C_{X_0}^*\dd h^a=\dd F^a,
        \qquad a=1,\ldots,n.
\end{equation}
\end{definition}

\subsection{Two elementary lemmas}
Assuming that $X_0$ is cyclic, we prove two elementary lemmas.
The first one says that a conservation law produces, by covariant differentiation with respect to the connection associated with $X_0$, a symmetric bilinear form which is invariant with respect to the F-manifold product (think of $\omega$ as $dh^a$ for some $a=1,\dots, n$ where $h^a$ is one of the densities of conservation laws).

\begin{lemma}\label{lem:invariant-Hessian}
Let $\omega$ be a closed one-form such that $C_{X_0}^*\omega$ is closed and $\tnabla$ is the connection associated with the cyclic vector field $X_0$. Put
\begin{equation}\label{eq:Hessian}
        \calH(X,Y):=(\tnabla_X\omega)(Y).
\end{equation}
Then $\calH$ is symmetric and invariant with respect to $\circ$, namely
\begin{equation}\label{eq:H-invariant}
        \calH(X\circ Y,Z)=\calH(X,Y\circ Z)
\end{equation}
for all local vector fields $X,Y,Z$.
\end{lemma}

\begin{proof}
Since $\tnabla$ is torsionless and $\dd\omega=0$, we have
\[
        \calH(X,Y)-\calH(Y,X)
        =\dd\omega(X,Y)=0.
\]
Thus $\calH$ is symmetric.

Now use the closedness of $C_{X_0}^*\omega$ to prove invariance with respect to $\circ$.  For arbitrary vector fields $X,Y$,
\begin{align}
0
&=\dd(C_{X_0}^*\omega)(X,Y) = X((C_{X_0}^*\omega)(Y))-  Y((C_{X_0}^*\omega)(X)) -  (C_{X_0}^*\omega)([X,Y])                                          \\
&=X(\omega(C_{X_0}Y))-  Y(\omega(C_{X_0}X)) -  \omega(C_{X_0}[X,Y]) \\\label{eq:3.5}
&=(\tnabla_X\omega)(C_{X_0}Y)-(\tnabla_Y\omega)(C_{X_0}X)
  +\omega\bigl((\tnabla_X C_{X_0})Y-(\tnabla_Y C_{X_0})X\bigr),
\end{align}
where in the last equality we used $X(\omega(C_{X_0}Y))=(\tnabla_X \omega)(C_{X_0}Y)+\omega((\tnabla_XC_{X_0})Y)+\omega(C_{X_0}\tnabla_XY)$, similarly for $Y(\omega(C_{X_0}X))$ and the fact that $\tnabla$ is torsionless. 
 The tensor $C_{X_0}$ is a vector-valued one-form, and therefore
\[
(d_{\tnabla}C_{X_0})(X,Y)
=\tnabla_X(C_{X_0}Y)-\tnabla_Y(C_{X_0}X)-C_{X_0}[X,Y].
\]
Since $\tnabla$ is torsionless, $[X,Y]=\tnabla_XY-\tnabla_YX$, and hence, expanding and simplifying one gets
\[
(d_{\tnabla}C_{X_0})(X,Y)
=(\tnabla_X C_{X_0})Y-(\tnabla_Y C_{X_0})X.
\]
Thus the last term in  formula \eqref{eq:3.5} is precisely
\[
\omega\bigl((d_{\tnabla}C_{X_0})(X,Y)\bigr),
\]
and it vanishes by $d_{\tnabla}C_{X_0}=0$.  Hence from the first part of \eqref{eq:3.5} and the definition of $\calH$ we get:
\begin{equation}\label{eq:H-CX-symmetric}
        \calH(X,C_{X_0}Y)=\calH(Y,C_{X_0}X).
\end{equation}
Since $\calH$ is symmetric, \eqref{eq:H-CX-symmetric} says that $C_{X_0}$ is self-adjoint with respect to $\calH$:
\begin{equation}\label{eq:CX-self-adjoint}
        \calH(X,C_{X_0}Y)=\calH(C_{X_0}X,Y).
\end{equation}
Therefore every power $C_{X_0}^k$ is self-adjoint.  By the cyclicity assumption, every multiplication operator $C_Y$ is a polynomial in $C_{X_0}$ with functional coefficients.  It follows that every $C_Y$ is self-adjoint with respect to $\calH$:
\begin{equation}\label{eq:CY-self-adjoint}
        \calH(X,Y\circ Z)=\calH(Y\circ X,Z),
\end{equation}
which is precisely \eqref{eq:H-invariant}.
\end{proof}

The second lemma turns an invariant Hessian (that is $\calH$) into the 3RC-curvature identity tested against the corresponding conservation-law one-form.

\begin{lemma}\label{lem:curvature-pairing}
Let $\nabla$ be a torsionless connection compatible with $\circ$ in the sense
\begin{equation}\label{eq:nabla-compatible-general}
        (\nabla_X\circ)(Y,Z)=(\nabla_Y\circ)(X,Z).
\end{equation}
Let $\omega$ be a one-form and put $\calH:=\nabla\omega$.  Assume that $\calH$ is invariant with respect to $\circ$.  Then
\begin{equation}\label{eq:curvature-pairing}
\omega\Bigl(
        R^{\nabla}(X,Y)(Z\circ W)
        +R^{\nabla}(Y,Z)(X\circ W)
        +R^{\nabla}(Z,X)(Y\circ W)
\Bigr)=0
\end{equation}
for all local vector fields $X,Y,Z,W$.
\end{lemma}

\begin{proof}
We give the details of the two steps used in the proof: first the invariant form
$\calH$ is written through a one-form, and then the skew part of its covariant
derivative is identified with the curvature applied to $\omega$.

Define the one-form
\begin{equation}\label{eq:theta}
        \theta(Y):=\calH(e,Y).
\end{equation}
Since $e$ is the unit and $\calH$ is invariant with respect to the product,
\begin{equation}\label{eq:H-theta}
        \calH(Y,Z)=\theta(Y\circ Z)
\end{equation}
for all local vector fields $Y,Z$.  Indeed,
\[
        \theta(Y\circ Z)
        =\calH(e,Y\circ Z)
        =\calH(e\circ Y,Z)
        =\calH(Y,Z).
\]
Thus the tensor $\calH$ is completely encoded by the one-form $\theta$ and the
multiplication $\circ$.

We now differentiate \eqref{eq:H-theta}.  We get:
\begin{align*}
(\nabla_X\calH)(Y,Z)
&=X\bigl(\calH(Y,Z)\bigr)-\calH(\nabla_XY,Z)-\calH(Y,\nabla_XZ)  \\
&=X\bigl(\theta(Y\circ Z)\bigr)
  -\theta((\nabla_XY)\circ Z)-\theta(Y\circ \nabla_XZ).
\end{align*}
On the other hand,
\[
(\nabla_X\theta)(Y\circ Z)
= X\bigl(\theta(Y\circ Z)\bigr)-\theta\bigl(\nabla_X(Y\circ Z)\bigr).
\]
Solving for $ X\bigl(\theta(Y\circ Z)\bigr)$ in the last formula and substituting in  the previous one gives
\begin{align}
(\nabla_X\calH)(Y,Z)
&=(\nabla_X\theta)(Y\circ Z) 
  +\theta\bigl(\nabla_X(Y\circ Z)-(\nabla_XY)\circ Z-Y\circ \nabla_XZ\bigr) \nonumber\\
&=(\nabla_X\theta)(Y\circ Z)+\theta\bigl((\nabla_X\circ)(Y,Z)\bigr).
\label{eq:cov-der-H-detailed}
\end{align}
This identity is only the covariant derivative of \eqref{eq:H-theta}; no further
assumption is being made.

Define
\begin{equation}\label{eq:K-definition}
        K(X,Y,Z):=(\nabla_X\calH)(Y,Z)-(\nabla_Y\calH)(X,Z).
\end{equation}
Using \eqref{eq:cov-der-H-detailed} twice, we obtain
\begin{align*}
K(X,Y,Z)
&=(\nabla_X\theta)(Y\circ Z)-(\nabla_Y\theta)(X\circ Z) \\
&\quad +\theta\bigl((\nabla_X\circ)(Y,Z)-(\nabla_Y\circ)(X,Z)\bigr).
\end{align*}
The last term vanishes by the compatibility condition \eqref{eq:nabla-compatible-general}.  Hence
\begin{equation}\label{eq:K-formula-expanded}
        K(X,Y,Z)
        =(\nabla_X\theta)(Y\circ Z)-(\nabla_Y\theta)(X\circ Z).
\end{equation}

We now apply \eqref{eq:K-formula-expanded} to the three triples
\[
        (X,Y,Z\circ W),\qquad (Y,Z,X\circ W),\qquad (Z,X,Y\circ W).
\]
This gives
\begin{align*}
K(X,Y,Z\circ W)
&=(\nabla_X\theta)(Y\circ Z\circ W)
  -(\nabla_Y\theta)(X\circ Z\circ W),\\
K(Y,Z,X\circ W)
&=(\nabla_Y\theta)(Z\circ X\circ W)
  -(\nabla_Z\theta)(Y\circ X\circ W),\\
K(Z,X,Y\circ W)
&=(\nabla_Z\theta)(X\circ Y\circ W)
  -(\nabla_X\theta)(Z\circ Y\circ W).
\end{align*}
Since the product is commutative and associative,  the six terms above
cancel in pairs, and we obtain
\begin{equation}\label{eq:K-cyclic-expanded}
        K(X,Y,Z\circ W)+K(Y,Z,X\circ W)+K(Z,X,Y\circ W)=0.
\end{equation}

It remains to identify $K$ with curvature evaluated against $\omega$.   Recall that 
\[
        \calH(Y,Z)=(\nabla_Y\omega)(Z),
\]
and it is assumed to be invariant with respect to the product $\circ$
Therefore, using the definition of the covariant derivative,
\begin{align}
(\nabla_X\calH)(Y,Z)
&=X\bigl(\calH(Y,Z)\bigr)-\calH(\nabla_XY,Z)-\calH(Y,\nabla_XZ) \nonumber\\
&=X\bigl((\nabla_Y\omega)(Z)\bigr)
  -(\nabla_{\nabla_XY}\omega)(Z)
  -(\nabla_Y\omega)(\nabla_XZ).      \label{eq:auxiliary1}
\end{align}
On the other hand, applying the definition of the covariant derivative of a one-form to the one-form $\nabla_Y\omega$, we get
\[
        (\nabla_X(\nabla_Y\omega))(Z)
        =X\bigl((\nabla_Y\omega)(Z)\bigr)
        -(\nabla_Y\omega)(\nabla_XZ).
\]
Solving for $(\nabla_Y\omega)(\nabla_XZ)$ in the last equation and substituting in\eqref{eq:auxiliary1} gives
\begin{equation}\label{eq:cov-H-second-omega-X}
        (\nabla_X\calH)(Y,Z)
        =(\nabla_X(\nabla_Y\omega))(Z)
        -(\nabla_{\nabla_XY}\omega)(Z).
\end{equation}
Interchanging $X$ and $Y$ gives similarly
\begin{equation}\label{eq:cov-H-second-omega-Y}
        (\nabla_Y\calH)(X,Z)
        =(\nabla_Y(\nabla_X\omega))(Z)
        -(\nabla_{\nabla_YX}\omega)(Z).
\end{equation}
Subtracting \eqref{eq:cov-H-second-omega-Y} from \eqref{eq:cov-H-second-omega-X}, and using that $\nabla$ is torsionless we obtain
\begin{align}
K(X,Y,Z)
&=(\nabla_X\calH)(Y,Z)-(\nabla_Y\calH)(X,Z) \nonumber\\
&=\bigl(
        \nabla_X\nabla_Y\omega
        -\nabla_Y\nabla_X\omega
        -\nabla_{[X,Y]}\omega
  \bigr)(Z).                         \label{eq:K-curvature-on-forms}
\end{align}

We now identify the expression in \eqref{eq:K-curvature-on-forms} with the curvature on vector fields.  
Applying the curvature commutator to the scalar function $\omega(Z)$, and using the Leibniz rule  gives, after a straightforward computation:
\begin{equation}\label{eq:curvature-pairing-zero}
        0
        =R^\nabla(X,Y)\bigl(\omega(Z)\bigr)
        =\bigl(R^\nabla(X,Y)\omega\bigr)(Z)
         +\omega\bigl(R^\nabla(X,Y)Z\bigr).
\end{equation}
Hence \eqref{eq:curvature-pairing-zero} gives
\begin{equation}\label{eq:one-form-curvature}
\bigl(\nabla_X\nabla_Y\omega-\nabla_Y\nabla_X\omega-\nabla_{[X,Y]}\omega\bigr)(Z)
        =-\omega\bigl(R^\nabla(X,Y)Z\bigr).
\end{equation}

Combining \eqref{eq:K-curvature-on-forms} with \eqref{eq:one-form-curvature}, we get
\begin{equation}\label{eq:Ricci-identity-expanded}
        K(X,Y,Z)=-\omega\bigl(R^\nabla(X,Y)Z\bigr).
\end{equation}
Substituting \eqref{eq:Ricci-identity-expanded} into the cyclic identity
\eqref{eq:K-cyclic-expanded} gives
\[
\omega\Bigl(
        R^{\nabla}(X,Y)(Z\circ W)
        +R^{\nabla}(Y,Z)(X\circ W)
        +R^{\nabla}(Z,X)(Y\circ W)
\Bigr)=0,
\]
which is \eqref{eq:curvature-pairing}.
\end{proof}

\subsection{The conservation-law theorem}
\begin{theorem}\label{thm:conservative-integrable}
Let $(M,\circ,e)$ be an F-manifold and consider the F-system
\beq\label{Fsys2}
u^i_t=(X_0\circ u_x)^i=c^i_{jk}X_0^ju^k_x,\qquad i\inw.
\eeq
Assume that $X_0$ is cyclic, and let $\tnabla$ be the associated torsionless connection.  If the F-system is locally conservative in the sense of Definition \ref{def:conservative}, then $\tnabla$ satisfies
\begin{equation}\label{eq:main-curvature-condition}
        R^{\tnabla}(X,Y)(Z\circ W)
        +R^{\tnabla}(Y,Z)(X\circ W)
        +R^{\tnabla}(Z,X)(Y\circ W)=0
\end{equation}
for all local vector fields $X,Y,Z,W$.  Conversely, in the analytic setting, if $\tnabla$ satisfies \eqref{eq:main-curvature-condition} then the F-system is locally conservative.
\end{theorem}

\begin{proof}
The proof is local.  Choose conserved densities and fluxes
\[
        h^1,\ldots,h^n,
        \qquad F^1,\ldots,F^n,
\]
as in Definition \ref{def:conservative}.  Set
\begin{equation}\label{eq:omega-a}
        \omega^a:=\dd h^a,
        \qquad a=1,\ldots,n.
\end{equation}
Since the densities are independent, the one-forms
\[
        \omega^1,\ldots,\omega^n
\]
form a local coframe.  In particular, a vector field $V$ is zero if and only if
$\omega^a(V)=0$ for every $a=1,\ldots,n$.

Each $\omega^a$ is closed, because it is exact.  Moreover, by the conservation-law
identity \eqref{eq:flux-differential},
\[
        C_{X_0}^*\omega^a=\dd F^a.
\]
Therefore $C_{X_0}^*\omega^a$ is closed as well.  Thus, for each $a$, the one-form
$\omega^a$ satisfies precisely the hypotheses of Lemma \ref{lem:invariant-Hessian}
with respect to the associated connection $\tnabla$.

Applying Lemma \ref{lem:invariant-Hessian}, we obtain for each $a$ a symmetric
bilinear form
\begin{equation}\label{eq:Ha-definition}
        \calH^a:=\tnabla\omega^a,
        \qquad
        \calH^a(X,Y)=(\tnabla_X\omega^a)(Y),
\end{equation}
which is invariant with respect to the product:
\begin{equation}\label{eq:Ha-invariant}
        \calH^a(X\circ Y,Z)=\calH^a(X,Y\circ Z).
\end{equation}
Let us record explicitly where the hypotheses enter.  The closedness of
$\omega^a$ gives the symmetry of $\calH^a$ because $\tnabla$ is torsionless.  The
closedness of $C_{X_0}^*\omega^a$, together with $d_{\tnabla}C_{X_0}=0$, gives the
self-adjointness of $C_{X_0}$ with respect to $\calH^a$.  The cyclicity of $X_0$ then implies that multiplication by any vector field is a
polynomial in $C_{X_0}$, and hence all multiplication operators are self-adjoint
with respect to $\calH^a$.  This is exactly the invariance
\eqref{eq:Ha-invariant}.

Now fix arbitrary local vector fields $X,Y,Z,W$ and define the curvature
combination
\begin{equation}\label{eq:curvature-combination-V}
\mathcal V(X,Y,Z,W):=
        R^{\tnabla}(X,Y)(Z\circ W)
        +R^{\tnabla}(Y,Z)(X\circ W)
        +R^{\tnabla}(Z,X)(Y\circ W).
\end{equation}
We shall show that this vector field is annihilated by the coframe
$\omega^1,\ldots,\omega^n$.

The connection $\tnabla$ is torsionless by construction, and it is compatible with
$\circ$ in the sense \eqref{eq:nabla-compatible-general} by the full symmetry of
$\tnabla$ established in Proposition \ref{prop:fullsym}.  Hence all hypotheses of Lemma
\ref{lem:curvature-pairing} are satisfied with
\[
        \nabla=\tnabla,
        \qquad \omega=\omega^a,
        \qquad \calH=\calH^a.
\]
Therefore Lemma \ref{lem:curvature-pairing} gives
\begin{equation}\label{eq:tested-curvature-expanded}
        \omega^a\bigl(\mathcal V(X,Y,Z,W)\bigr)=0,
        \qquad a=1,\ldots,n.
\end{equation}
Since the $\omega^a$ form a coframe, \eqref{eq:tested-curvature-expanded}
implies
\[
        \mathcal V(X,Y,Z,W)=0.
\]
By the definition \eqref{eq:curvature-combination-V} of $\mathcal V$, this is
precisely \eqref{eq:main-curvature-condition}. This proves the first part of the theorem. 
In the analytic setting, the converse statement follows from Remark \ref{coframe}.

\end{proof}

\begin{remark}
Conservative F-systems include as a particular case F-systems which are Hamiltonian  with respect to a local Hamiltonian structure. Thus the above arguments provides an alternative proof of generalised Novikov's conjecture in the local case. 
\end{remark}

\section{Conclusions}
Tsarev's theory has been developed for strictly hyperbolic quasilinear first order systems of PDEs admitting Riemann invariants. Tsarev defined integrability as compatibility of the linear system of PDEs defining the symmetries of the system. The same integrability conditions ensure the compatibility of the linear system for densities of conservation laws and of Tsarev's system providing the metrics involved in the definition of the associated Hamiltonian structures.

One of the key ingredients of Tsarev's theory is the possibility of determining the solutions of the system using its symmetries via the generalised hodograph method. For a given initial value problem, Tsarev showed that the symmetry defining the solution is uniquely determined by the initial datum thanks to a slight extension of a classical theorem of Darboux. The same theorem allows an extension of Tsarev's theory to a class of non-diagonalisable regular F-systems called Darboux-Tsarev systems (see \cite{LPVG2,LPVG3}).

As in Tsarev's framework, one can define the integrability of this class of systems as compatibility of the linear system for the symmetries. This leads to the geometric conditions 
\beq\label{shc-final}
\tilde R^s_{lmi}c^j_{ks}+\tilde R^s_{lik}c^j_{ms}+\tilde R^s_{lkm}c^j_{is}=0,\qquad
\tilde R^j_{skl}c^s_{mi}+\tilde R^j_{smk}c^s_{li}+\tilde R^j_{slm}c^s_{ki}=0,
\eeq
involving the Riemann tensor of a connection uniquely defined by the system and the structure functions of the underlying Hertling-Manin F-product.

In this paper we have significantly extended the theory in  two  directions.
\begin{itemize}
\item Relaxing the regularity assumption and considering a wider class of F-manifolds.  The
 only assumption we did is the ciclicity of the vector field defining the F-system. This is the minimal
   assumption we need to construct the natural connection that controls the integrability of the  system. In this framework, we prove a generalisation of a conjecture of Novikov (proved by Tsarev in the semisimple case)  relating the existence of a Hamiltonian structure and integrability. 
\item Relaxing the  Darboux-Tsarev hypothesis that allows to extend Tsarev's  theory in the regular setting.  This requires  the analyticity assumptions of the Cauchy-Kowalevski-Cartan-K\"ahler setting. Using these assumptions, we show that an F-system is integrable if and only if it is locally conservative and we prove the existence of a family of symmetries providing the unique local analytic solution of the Cauchy problem through the generalised hodograph method. 
\end{itemize}
The  above results clarify that the compatibility conditions \eqref{shc-final} are not only necessary but also sufficient conditions for the integrability of this wide class of systems, which, as far as we know, includes all known examples of quasilinear systems of evolutionary first-order partial differential equations.

\end{document}